\documentclass[opre,nonblindrev]{informs3} %

\DoubleSpacedXI %

\usepackage{natbib}
 \bibpunct[, ]{(}{)}{,}{a}{}{,}%
 \def\bibfont{\small}%

\usepackage[normalem]{ulem}
\usepackage{tikz}
\usepackage{pgfplots}
\usepackage{environ}
\usetikzlibrary{quotes,angles}

\usepackage{mathtools}
\newcommand{\RR}[0]{\mathbb{R}}

\newcommand{\PP}[0]{\mathbf{P}}
\newcommand{\eps}{\varepsilon } %
\newcommand{\tra}{t } %
\newcommand{\all}{x } %
\newcommand{\ps}{p^\star } %
\newcommand{\rs}{R^\star } %
\newcommand{\vsr}[0]{{v^\star}}
\newcommand{\w}[0]{w}
\newcommand{\E}{\mathsf{E}} %
\newcommand{\bv}[0]{\overline{v}}

\newcommand{\vm}[0]{\ps+\delta}
\newcommand{\vz}[0]{\ps-\delta}
\newcommand{\xr}[0]{x_{\textup{\texttt{pd}}}}
\newcommand{\dxr}[0]{\dot{x}_{\textup{\texttt{pd}}}}
\newcommand{\yr}[0]{y_{\textup{\texttt{pd}}}}
\newcommand{\dyr}[0]{\dot{y}_{\textup{\texttt{pd}}}}
\newcommand{\tr}[0]{t_{\textup{\texttt{pd}}}}
\newcommand{\mr}[0]{m_{\textup{\texttt{pd}}}}
\newcommand{\ir}{{\scriptscriptstyle \text{\rm IR}}}
\newcommand{\ic}{{\scriptscriptstyle \text{\rm IC}}}
\newcommand{\Lic}{\lambda^{\ic}}
\newcommand{\Lir}{\lambda^{\ir}}
\newcommand{\xd}[0]{x_{\textup{\texttt{d}}}}
\newcommand{\td}[0]{t_{\textup{\texttt{d}}}}
\newcommand{\md}[0]{m_{\textup{\texttt{d}}}}
\newcommand{\vd}[0]{v_{\textup{\texttt{d}}}^\star }

\newcommand{\1}[1]{\mathbf{1}_{\{#1\}}}

\newcommand{\pd}[0]{p}  %
\newcommand{\pt}[0]{s}  %

\DeclarePairedDelimiter\floor{\lfloor}{\rfloor}

\newcommand{\twopartdef}[4]{
\begin{dcases*} #1 &\mbox{if } $#2$\\
#3 &\mbox{if } $#4$
\end{dcases*}
}

\newcommand{\threepartdef}[6]{
 \begin{dcases*} #1 &\mbox{if } $#2$\\
#3 &\mbox{if } $#4$\\
#5 &\mbox{if } $#6$ \end{dcases*}
}

\usepackage{cleveref}
\usepgfplotslibrary{fillbetween}
\usetikzlibrary{plotmarks}

\TheoremsNumberedThrough     %
\ECRepeatTheorems

\EquationsNumberedThrough    %

\begin{document}

\RUNAUTHOR{Balseiro, Besbes and Castro}

\RUNTITLE{Mechanism Design under Approximate Incentive Compatibility}

\TITLE{ {Mechanism Design under \\Approximate Incentive Compatibility}}

\ARTICLEAUTHORS{%
\AUTHOR{Santiago R. Balseiro}
\AFF{Columbia Business School} %
\AUTHOR{Omar Besbes}
\AFF{Columbia Business School}
\AUTHOR{Francisco Castro}
\AFF{UCLA Anderson School of Management}
} %

\ABSTRACT{%
A fundamental assumption in classical mechanism design is that buyers are perfect optimizers. However, in practice, buyers may be limited by their computational capabilities or a lack of information, and may not be able to perfectly optimize their response to a mechanism. This has motivated the introduction of approximate incentive compatibility (IC) as an appealing solution concept for practical mechanism design. While most of the literature has focused on the analysis of particular approximate IC mechanisms, this paper is the first to study the design of \textit{optimal} mechanisms in the space of approximate IC mechanisms and to explore how much revenue can be garnered by moving from exact to approximate incentive constraints. In particular, we study the problem of a seller facing one buyer with private values and analyze optimal selling mechanisms under $\eps$-incentive compatibility.  We establish that the gains that can be garnered depend on the local curvature of the seller's revenue function around the optimal posted price when the buyer is a perfect optimizer. If the revenue function behaves locally like an $\alpha$-power for $\alpha \in (1,\infty)$, then no mechanism can garner gains higher than order $\eps^{\alpha/(2\alpha-1)}$. This improves upon state-of-the-art results which imply maximum gains of $\eps^{1/2}$ by providing the first parametric bounds that capture the impact of revenue function's curvature on revenue gains. Furthermore, we establish that an optimal mechanism needs to randomize as soon as $\eps>0$ and construct a randomized mechanism that is guaranteed to achieve order $\eps^{\alpha/(2\alpha-1)}$ additional revenues, leading to a tight characterization of the revenue implications of approximate IC constraints. Our study sheds light on a novel class of optimization problems and the challenges that emerge when relaxing IC constraints. In particular, it brings forward the need to optimize not only over allocations and payments but also over best responses, and we develop a new framework to address this challenge.
}%


\maketitle
\newpage
\section{Introduction}
\label{sec:intro}

From housing allocation to online advertising, market design is having a profound impact on the design, implementation, and operations of markets. An idea fundamental to good market design is the correct assessment of participants' incentives, usually through an exact equilibrium concept. Indeed, most of the literature in mechanism design has assumed that agents make decisions by perfectly optimizing their utility functions, even when facing arbitrarily complex mechanisms. This idea has impacted how practitioners design and operationalize their market solutions, and it has led to a long line of research that has provided a good understanding of optimal mechanism design with \textit{exact} incentive compatibility constraints. However, in real-world applications, markets are complex and participants are not always able to best respond to market conditions. Lack of information or experience, technological limitations, or even behavioral biases might prevent buyers from  accurately assessing their utility from a transaction, and thus buyers might make only near-optimal decisions. %
 Our focus in this paper is to explore the consequences of these limitations, and to understand the potential value that market designers can derive from them. In particular, the main question we tackle is:

\begin{quote}
\textit{What is an optimal or near-optimal selling mechanism when buyers are not perfect optimizers, and  how does the problem structure impact the level of additional revenue performance that can be gained?}
\end{quote}

There has been significant interest in studying specific mechanisms and quantifying the extent to which they are incentive compatible or not (see Section~\ref{sec:literature_review}). This line of work studies specific instances in the space of feasible mechanisms under relaxed incentive compatibility constraints. In addition, many studies have shown how to transform an approximately incentive-compatible mechanism into an exactly incentive compatible one while controlling for the loss in revenues. However, to the best of our knowledge, the study of the \textit{optimization} of mechanisms under relaxed incentive compatibility constraints is new. As we will see, this leads to a new class of optimization  problems that is fundamentally different in nature than classical mechanism design. Furthermore, we will derive new insights into the relationship between the perturbation of IC constraints and the revenue performance that can be achieved.

To make progress toward this question, we anchor our analysis around the classical Myerson setting with a single buyer with private values.  In this case, a posted price mechanism is known to be optimal when the buyer is a perfect optimizer. To capture that the buyer is not a perfect optimizer, we allow the incentive compatibility (IC) constraints to be satisfied up to a ``small'' $\eps>0$, i.e., the buyer is satisficing and near-optimizer, and he can select a reporting strategy that is $\eps$ away from optimal. { Studying the case of small $\eps$ is natural as agents, while satisficing, might not be willing to forgo large gains from not perfectly optimizing.} We will refer to this constraint as $\eps$-IC. Essentially, we study a classical mechanism design problem in which the IC constraint is slightly relaxed.

It is worth noting that the classical setting with IC constraints is an \textit{infinite-dimensional linear program} that admits a very simple solution through a posted price (see, e.g., \citealt{rileyzeckhauser1983}). The problem with $\eps$-IC constraints is also an infinite-dimensional linear program, but, as we will see, it leads to a rich new class of problems. In light of the question above, we are interested in quantifying the difference between the values of the two problems as a function of $\eps$, but also in understanding the underlying structure of near-optimal mechanisms in the $\eps$-IC setting and the impact of the buyer's distribution of private information on performance.

\paragraph{Existing benchmarks.}
A na\"ive comparison might lead one to expect the difference between the values of these two linear programs to scale (at most) in a piecewise linear fashion (as a function of $\eps$) based on intuition gleaned from finite dimensional linear programs. However, in the class of mechanism design problems, we will see that the picture is more subtle and rich,  since we are relaxing an \textit{infinite number of constraints}.
On another hand, an application of the powerful rounding argument of Nisan, which allows to transform an $\eps$-IC mechanism into an IC mechanism while bounding the revenue loss (see \Cref{sec:Nisan}), enables one to readily obtain an upper bound on the maximal additional revenues that could be garnered of order $\eps^{1/2}$. %

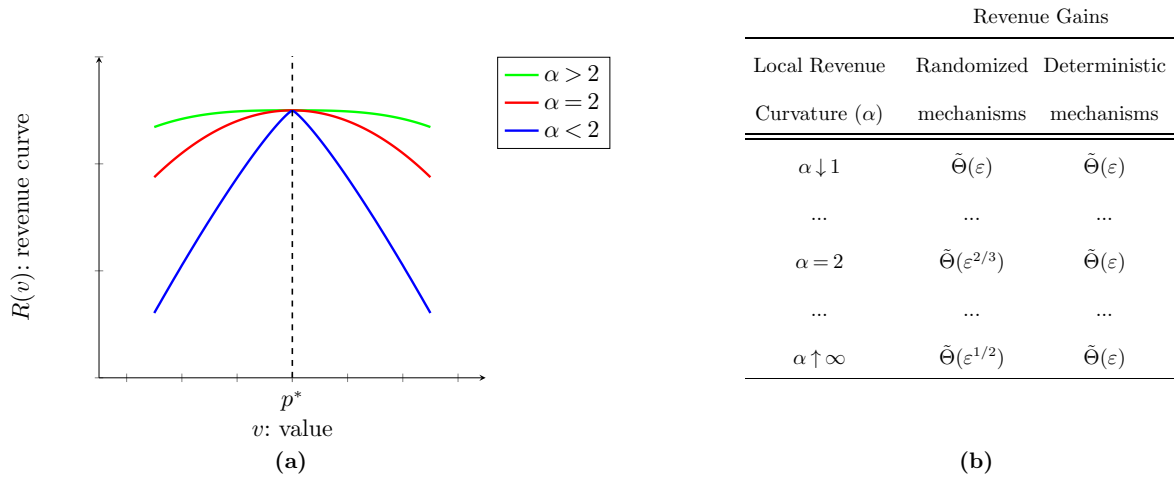
\begin{figure}[t]
\begin{center}
\scalebox{1}{\begin{tikzpicture}[scale=0.75,every node/.style={scale=1.1}]

        \begin{axis}[
        ylabel style={align=center},
        xlabel={$v$: value},
        xlabel style={align=right},
        ylabel={$R(v)$: revenue curve} ,
        axis y line=left,
        axis x line=bottom,
        xmin=1/4-0.1, xmax=3/4+0.1, ymin=0, ymax=0.3,
        extra x ticks={0.5},
        extra x tick labels={$p^*$},
        legend entries = {{$\alpha>2$},{$\alpha=2$},{$\alpha<2$}},
        legend pos=outer north east,
        yticklabels={,,},
        xticklabels={,,},
        ]

        \addplot [very thick, green, domain=1/4:3/4, samples=100]
      {1/4 - abs(x-1/2)^(3)};
        \addplot [very thick, red, domain=1/4:3/4, samples=100]
        {1/4 - abs(x-1/2)^2};
        \addplot [very thick, blue, domain=1/4:3/4, samples=100]
        {1/4 - abs(x-1/2)^(1.2)};

        \addplot[dashed, thick, black] coordinates {(1/2,0)(1/2,1)};
        \end{axis}
        
\node[scale=0.65,xshift = 16cm,yshift = 3.6cm]  {%
  \begin{tabular}{ccc}
	\begin{tabular}{ cccc }
			  & & \multicolumn{2}{c}{Revenue Gains}  \\ \hline 
			Local Revenue  &    & Randomized   &  Deterministic      \\
	Curvature ($\alpha$) & 	                     & mechanisms  & mechanisms   \\
			                     \hline\hline
$\alpha \downarrow 1$	& &$\tilde \Theta(\eps)$ & $\tilde \Theta(\eps)$ \\
...	&  & ...  & ...  \\
$\alpha =2$	& &$\tilde \Theta(\eps^{2/3})$ & $\tilde \Theta(\eps)$ \\
...	&  & ...  & ...  \\
$\alpha \uparrow \infty$	& &$\tilde \Theta(\eps^{1/2})$ & $\tilde \Theta(\eps)$ \\
			\hline
		\end{tabular}
  \end{tabular}
};

\node[scale = 0.7] at (3.4,-1.5){ \textbf{(a)}};
\node[scale = 0.7] at (2.55+13,-1.5){ \textbf{(b)}};

\end{tikzpicture}}
\end{center}
\caption{\textbf{(a)} Local curvature of the revenue function $R(v) = v (1 - F(v))$ around the optimal posted price $\ps$. Here, $F(v)$ denotes the cumulative distribution function of the buyer's value.
\textbf{(b)} Revenue improvement for $\eps$-optimizers as a function of $\eps$. %
 In this paper, we provide matching upper and lower bounds showing that the optimal rate is of order $\eps^{\alpha/(2\alpha-1)}$ for randomized mechanisms where $\alpha\in(1,\infty)$ captures the local curvature of the revenue function. We also show that the rate is linear for deterministic mechanisms.
}
\label{fig:curvature}
\end{figure}

\paragraph{Contributions.}
Our main contribution is to initiate the study of the implications of $\eps$-IC constraints from an optimization perspective. While most of the related literature has focused on measuring deviations from IC constraints, on studying mechanisms that are $\eps$-IC, or on reductions from $\eps$-IC to IC, our focus is instead on understanding the structure of \textit{optimal} $\eps$-IC mechanisms and their revenue guarantees.

Our objective is to study the impact of the buyer's distribution of values on the revenue that can be garnered by relaxing the incentive compatibility constraints. Notably, we identify that the revenue gains when $\eps$ is small are governed by the local curvature of the revenue function around the optimal posted price when the buyer is a perfect optimizer. More formally, denote by $F(v)$ the cumulative distribution function of the buyer's value and by $R(v) = v (1-F(v))$ the revenue function, which measures the seller's expected revenue when the posted price is $v$ and the buyer is a perfect optimizer. Denote by $\ps$ an optimal posted price. Roughly speaking, we say that the revenue function admits local $\alpha$-power envelopes for $\alpha \in (1,\infty)$ if
\[
    R(\ps) - R(v) \asymp  |v - \ps|^\alpha\,,
\]
for posted prices $v$ close to the optimal, that is, the revenue loss of choosing a suboptimal posted price decays at a rate of $\alpha$.\footnote{ Our working definition of local envelopes in Definition~\ref{def1} is actually slightly stronger as it is stated in terms of the derivative of the revenue function, but these are equivalent in many cases of interest.} See Figure~\ref{fig:curvature} (a) for an illustration of different local behaviors.

We first explore upper bounds on the achievable performance. As articulated above, we first review how Nisan's classical rounding argument  enables one to readily obtain an  impossibility result through an upper bound on the maximal additional revenues that could be garnered, given by order $\eps^{1/2}$ (\Cref{prop:smoothing}). This yields a powerful, uniform guarantee that is independent of the curvature of the revenue function. Our first main result is a parametric upper bound on the revenue gains (that can be obtained when relaxing the IC constraints) that depends on the local curvature of the revenue function. In particular, we establish in \Cref{thm:up-bd-per} that no (randomized) mechanism can yield revenue gains of order more than  $\eps^{\alpha/(2\alpha-1)}$, which yield an improvement on classical results for every $\alpha \in (1,\infty)$. Interestingly, in the limit when $\alpha \uparrow \infty$, i.e., when the revenue curve is locally flat, we recover the classical bound of $\eps^{1/2}$ from Nissan. In the limit when $\alpha \downarrow 1$, i.e., when the revenue curve has a kink, we recover the $\eps$ gains suggested by sensitivity analysis for finite dimensional linear programming. Finally, in the prototypical case of $\alpha=2$, i.e., when the revenue curve is smooth and locally quadratic, we obtain a novel bound of $\eps^{2/3}$ on the maximum revenue gains that can be garnered.

We derive our upper bound through a novel duality argument that also highlights the different nature of objects that emerge in this class of problems, and the interplay between $\eps$-IC and revenues.  Our approach involves guessing a best response for the buyer and then relaxing, for each value, the IC constraints of all reports except the one corresponding to the report made by the best response. Because these pairs of values and reports induce a path in the two-dimensional Euclidean space, we dub the resulting problem as the \emph{path-based relaxation}. We further upper bound the path-based relaxation by considering the Lagrangian dual problem obtained after dualizing the remaining constraints. The most challenging part of our analysis involves bounding the value of the latter problem. Here, our choice of the best response ends up playing a key role.

Our second set of main results pertains to achievability. We start our search for nearly optimal mechanisms, which would yield lower bounds on attainable revenues, by exploring the space of deterministic mechanisms. Our study of deterministic mechanisms is motivated by the fact that an optimal mechanism with exact IC constraints is deterministic. In that space, we establish in \Cref{prop:hard/soft-floor} that the gains are of order $\eps$ and derive an optimal mechanism. %

We then expand our search to include randomized mechanisms. We establish in \Cref{thm:low-bd-per} that there exists a mechanism that leads to gains of at least order $\eps^{\alpha/(2\alpha-1)}$.  Notably, our approach is constructive in that we exhibit a family of mechanisms that delivers such supra-linear gains for every $\alpha \in (1,\infty)$. These are an appropriately constructed  perturbation of the optimal  mechanism when $\eps=0$. The construction of this family is also instructive in that it highlights the special nature of the problem at hand. In particular, the mechanism is randomized and its allocations are characterized by a set of ordinary and delayed differential equations. Such delayed differential equations appear critical for understanding the class of $\eps$-IC mechanism design problems.

We summarize in \Cref{fig:curvature} (b) the results in the present paper. Notably, our bounds provide a sharp characterization of the performance implications of approximate incentive compatibility. The gains stemming from it are exactly of order $\eps^{\alpha/(2\alpha-1)}$, thus yielding a spectrum of gains depending on the local curvature $\alpha \in (1,\infty)$ of the revenue function. Furthermore, our results imply that  an optimal mechanism \textit{has to randomize}.  This is true even if, for example,  the underlying distribution is regular. This unveils an interesting phase transition. For $\eps=0$ the optimal allocation is deterministic---a posted-price mechanism---but as soon as $\eps$ becomes positive, the optimal mechanism requires randomization.  Our results also shed light on the geometry of the set of incentive-compatible mechanisms. Because mechanism design problems are linear programs, by Bauer's maximum principle they attain their maximum at an extreme point of the feasible set. Interestingly, in the  case of exact IC constraints, extreme points are deterministic posted-price mechanisms~\citep{manelli2007multidimensional}. A takeaway from our results is that as soon as $\eps > 0$, the set of extreme points becomes richer as optimal mechanisms must randomize.

Our performance analysis focuses on the case of $\eps$ small to derive clean analytical insights on the limiting revenue gains when $\eps \downarrow 0$. We remark, however, that our approach is also valid to develop upper and lower bounds for arbitrary $\eps$. Our path-based relaxation dualizes $\eps$-IC constraints along the path of a best response and, by weak duality, this approach yields upper bounds for any value of $\eps$. Moreover, by optimizing over the paths we can readily obtain upper bounds for every instance. Similarly, our lower bound provides a framework to construct a feasible $\eps-$IC mechanism through a set of ordinary and  delayed differential equations. Both methods are general and provide a computational framework for $\eps-$IC mechanism design for values of $\eps$ that are not necessarily small. \label{ref2-1}

From a methodological perspective, an important and novel takeaway emerges. While in classical mechanism design with exact IC constraints, the problem can be formulated  as an optimization over allocations and payments, now, as one relaxes the IC constraints to be approximately satisfied, an additional endogenous object emerges in the optimization, through the best response of the buyer (we discuss this in  \Cref{sec:formulation} after we introduce the problem formally). In the same way that the allocation pins down payments for an exactly IC mechanism, we show that the best response can be the central object that can be used to optimize over $\eps$-IC mechanisms. For example, in our path-based relaxation, a best response fully determines an upper bound on the achievable performance. Similarly, in constructing our proposed mechanisms, our choice of the best response (together with some reasonable assumptions on the mechanism) pins down both the allocation and payments. Optimizing over best responses, and the associated methodology, might have implications beyond the exact problem studied in the present paper.

\subsection{Literature Review}\label{sec:literature_review}

Our paper builds on the classical mechanism design literature for revenue maximization under exact IC constraints, e.g., \citet{myerson1979incentive} and \citet{myerson1981optimal}. More specifically, the setting of selling to a single buyer relates to the work of \citet{rileyzeckhauser1983}. Under an exact IC constraint,  \citet{rileyzeckhauser1983} establish that an optimal mechanism is simply a posted price. In contrast, as soon as IC is relaxed to hold approximately, we establish that any deterministic allocation mechanism is suboptimal. %

There has recently been a surge of interest in measuring IC and relaxing IC in the context of selling mechanisms. \citet{milgrom2011critical} highlights that understanding notions of approximate IC and their implications for performance is one of the four critical issues in the practice of market design. In the present paper, for a class of selling problems, we study the implications of relaxing IC constraints for both the structure of optimal mechanisms and the associated performance, while  also shedding light on the rich structure of the associated optimization problems that emerge.

There are various studies that analyze mechanisms that are approximately IC in various contexts. \cite{hartline2010bayesian} and \cite{hartline2015non} study efficient welfare black-box reductions that turn an allocation algorithm into a Bayesian incentive-compatible mechanism with minimal welfare and revenue losses in a single-dimensional setting. Several papers study similar reductions in more general settings such as multi-dimensional agents' types, discrete or continuous types, or combinations thereof.  In higher multi-dimensional settings, \cite{bei2011bayesian} provide reductions that are approximately incentive compatible for the case of discrete type space, while \cite{hartline2011bayesian} and \cite{hartline2015bayesian} find reductions that are exact IC. Until recently, the problem for continuous type space only had $\eps-$IC reductions but this gap was closed by \cite{dughmi2021bernoulli}. Finally, \citet{cai2021efficient} provide an efficient revenue preserving (with small loss) transformation from $\eps-$IC to exact IC
 and \cite{conitzer2020welfare} provide a transformation that preserves revenue and incurs negligible revenue loss.

We also refer the reader to \cite{Caroll2013}  and \cite{Dutting2021} that motivate and study  notions of $\eps-$IC in the context of voting rules and contract design, respectively.

\citet{nazerzadeh2013dynamic} and \citet{kanoria2017dynamic} consider approximately IC mechanisms in a dynamic context with learning. \citet{balseiro2015repeated} and \citet{balseiro2019dynamic} focus on such an environment in the context of multi-period settings with cumulative budget constraints. \citet{Gorokh2019arti} also exhibits a mechanism that satisfies some form of approximate IC when converting one-shot monetary mechanisms to dynamic mechanisms with artificial currencies. While the studies above use related notions of approximate IC, they differ in their focus. In these studies, the main motivation for the introduction of approximate IC is computational and analytical tractability since, in many settings, even computing a near-optimal, exactly IC mechanism is not possible. To overcome this challenge, these papers consider tractable mechanisms that are approximately IC and whose performance is close to that of an optimal, exactly IC mechanism. By contrast, in the present study, we ask what should be an \textit{optimal} mechanism if one optimizes over the space of approximate IC mechanisms. An interesting and novel takeaway of our work is that relaxing the IC constraint leads to supralinear gains.

The relaxation of IC constrains has also emerged in other general market design settings. We refer the reader to \citet{Lubin-2012-Approx} and \citet{Azevedo-2018-SPL} as well as references therein for an overview of such other applications. At a high level, these studies view approximate IC as a desirable property of a mechanism being studied.

With the emergence of black-box mechanisms in practice, another related line of research has focused on measuring {how far a non-truhtfhul mechanism is from incentive compatible}. Examples of recent work in this area include \citet{Lahaie2018-IC}, \citet{balcan2019estimating}, \citet{Deng2019-IC}, \citet{Feng-2019-IC}, \citet{Deng-2020-IC}, and \citet{colinibaldeschi2020envy}.

Finally, our work expands the set of techniques and ideas that have been used in the literature to analyze approximate incentive-compatible mechanisms. A central approach is the rounding argument often attributed to Noam Nisan and, to the best of our knowledge, first used by \citet{balcan2005mechanism} to turn any $\eps$-IC mechanism into an exactly IC mechanism while, at the same time, limiting the revenue losses. This idea and the associated guarantees have been extended by \citet{daskalakis2012symmetries} and \citet{rubinstein2015simple} to multiple bidders settings. %

\section{Problem Formulation}\label{sec:formulation}
We consider the classical setting  developed in \cite{myerson1981optimal} with a seller (she) selling a single item to one buyer (he). The buyer's value for the item is drawn from a distribution $F$ with support $\mathcal{S}$. The distribution $F$ is common knowledge.
The buyer's set of messages is denoted by $\Theta$. The seller aims to design an indirect selling mechanism given by $\hat{\all}:\Theta\rightarrow \RR$ and $\hat{\tra}:\Theta\rightarrow \RR$, where $\hat{\all}$ denotes the  allocation probability and $\hat{\tra}$  the  transfers. The buyer's strategic response to the mechanism is a function $\hat{\theta}: \mathcal{S}\rightarrow \Theta$. Given a particular mechanism and response, the buyer's utility is given by $v\cdot \hat{\all}(\hat{\theta}(v))-\hat{\tra}(\hat{\theta}(v))$.
Our critical modeling assumption is that the buyer is not a perfect optimizer. We capture this by imposing that the incentive compatibility
constraints are satisfied up to $\eps$.
More precisely, the buyer aims to select a reporting strategy that ensures that he collects non-negative utility from participating,   and the buyer is satisficing in that he aims to collect the maximum surplus up to $\varepsilon$, i.e., for all $v\in \mathcal{S}$, $v\cdot \hat{\all}(\hat{\theta}(v))-\hat{\tra}(\hat{\theta}(v))\geq
 \sup_{\hat{\theta}'\in\Theta} \{v\cdot \hat{\all}(\hat{\theta}')-\hat{\tra}(\hat{\theta}') \}
-\eps.$ When there are multiple $\eps$-optimal best responses for the buyer, we assume that the buyer chooses the one that is the most favorable to the principal. %
  An implication of this assumption, which is pervasive in the mechanism design literature, is that the seller's problem reduces to simultaneously choosing a mechanism together with a best response for the buyer, which together are individually rational and approximately incentive compatible.   The seller aims to maximize the expected revenues from trade $\E_v[\hat{\tra}(\hat{\theta}(v))]$ subject to the above constraints.

\paragraph{Reformulation via  the revelation principle.}
The key change compared to a classical mechanism design problem is that the buyer is not a perfect optimizer.
 However, one may use  classical arguments (see \citealt{myerson1979incentive}) to reduce attention to  direct revelation mechanisms (see Appendix~\ref{app-rev} for details). In particular, without loss of optimality, one may restrict attention to the following  seller's problem
\begin{flalign}\label{eq:probeps}\tag{$\mathcal{P}_\eps$}
 \sup_{\all(\cdot),\tra(\cdot)} & \E_v[\tra(v)]\\\label{eq:IR} \tag{IR}
\text{s.t. }& v\cdot \all(v)-\tra(v)\geq 0,\quad \forall v\in \mathcal{S},\\\label{eq:IC} \tag{$\mbox{IC}_{\varepsilon}$}
&v\cdot \all(v)-\tra(v)\geq
 v\cdot \all(v')-\tra(v')
-\eps,\quad \forall v\in \mathcal{S},\quad v'\in \mathcal{S},\\\nonumber
&\all:\mathcal{S} \rightarrow [0,1] \quad \text{and}\quad \tra:\mathcal{S}\rightarrow \RR.
\end{flalign}
Problem \eqref{eq:probeps} will be our main focus of analysis. Our goal is to shed light on the structure of an optimal or near-optimal solution to \eqref{eq:probeps} and provide a comprehensive characterization of the added value of \eqref{eq:IC} compared to imposing exact IC.

Note that when $\eps=0$,  \eqref{eq:probeps} corresponds to the standard mechanism design problem. In the remainder of the paper we will use $(\mathcal{P}_0)$ to refer to the latter problem.  We denote by $\mathcal{M}(\varepsilon)$  the set of feasible mechanisms for problem  \eqref{eq:probeps} and use  $m$ to denote an element of this set. For a mechanism $m=(x,t) \in \mathcal{M}(\varepsilon)$, we denote by $\Pi(m)$ the corresponding revenues.
Note that for any $\varepsilon>0$, $ \mathcal{M}(0) \subset \mathcal{M}(\varepsilon) $.

In what follows we will be interested in quantifying the revenue implications associated with  relaxing the set of feasible mechanisms from $\mathcal{M}(0)$ to $\mathcal{M}(\varepsilon)$.
 Let
 \begin{equation*}
 \Pi^\star(\mathcal{M}(\eps))\triangleq \sup_{m\in \mathcal{M}(\eps)} \Pi(m)
 \end{equation*}
 denote the optimal revenue under constraints $\mathcal{M}(\eps)$.  We are interested in quantifying the difference in revenues
\begin{equation}\label{eq:key-Metric}
 \Pi^\star(\mathcal{M}(\eps))- \Pi^\star(\mathcal{M}(0)),
\end{equation}
and how it changes as a function $\varepsilon$ when the latter is small.
It is worth noting that by comparing $ \Pi^\star(\mathcal{M}(\eps))$ to $ \Pi^\star(\mathcal{M}(0))$, we are essentially comparing the values of  two infinite-dimensional linear programs, where in one of these, the right-hand side of the incentive compatibility constraint is relaxed. As noted in the introduction, a na\"ive comparison,  based on the intuition gleaned from finite dimensional linear programs,
may lead one to expect that the difference above  scales (at most) in a piecewise linear fashion. However, we will see that this not the case, as we are relaxing an infinite number of constraints.

We aim to quantify the gap in  \Cref{eq:key-Metric} in an  instance dependent manner, i.e., to understand the magnitude of this gap as a function of the underlying distribution. %
We first introduce the following technical assumption on the distribution of values.

\begin{assumption} \label{assumption1}
The distribution $F$ admits a density, $f$, and has support in $\mathcal{S}=[0,\bv]$ with $0 \le \bv <\infty$. Moreover, the density $f$ has bounded variation.
\end{assumption}
This assumption is mild and is satisfied by a broad set of distributions commonly studied in the literature. We assume that the lower bound of the support is zero to simplify the analysis, but our result could be extended to distributions whose support does not include zero.
Note that bounded variation implies that $f(x)$ is bounded. We use $\bar{f}$ to denote an upper bound.

To study instance dependent bounds, we  parametrize distribution families through their local behavior around the price that maximizes their corresponding revenue function. Formally, let $\bar{F}(v)$ denote $1-F(v)$, and let $R(v)$  denote $v\cdot \bar{F}(v)$. We define an optimal price and optimal revenue by
\begin{equation}
\ps \in \argmax_{v\in[0,\bv]} R(v),\quad \text{and}\quad \rs = \max_{v\in[0,\bv]} R(v).
\end{equation}
Our parametrization of distributions is then given by the following definition.
\begin{definition}[Local $\alpha$-power envelopes] \label{def1} We say that the revenue function $R(v)$ admits local $\alpha$-power envelopes if
there exists a unique interior solution $\ps$, and
there exists $\alpha \in (1,+\infty)$, and positive constants $\kappa_L,\kappa_U$ and a neighborhood of $\ps$, $\mathcal{N}_\ell=(\ps-\ell,\ps+\ell)\subset (0,\bv)$, such that
$$\kappa_L\alpha\cdot |v-\ps|^\alpha \leq (\ps-v)\cdot \dot{R}(v)\leq \kappa_U\alpha\cdot |v-\ps|^\alpha, \quad \forall v\in\mathcal{N}_\ell.
$$
\end{definition}
This definition characterizes the curvature of the revenue function around $\ps$.
A distribution with a local $\alpha$-power envelope is such that its associated revenue function has a power of $\alpha$ local behavior around the unique optimal price. For example, when $\alpha=2$ then the behavior is quadratic around $\ps$. The latter case can be thought as the prototypical one, and  holds for a wide range of distributions  (e.g., uniform, exponential,...). A distribution in the limiting case when $\alpha\uparrow \infty$ is almost flat (and approaches the isorevenue distribution) around the optimal price; while a distribution
in the limiting case when $\alpha\downarrow 1$ has a kink around $\ps$. That is, the local $\alpha$-power envelopes notation characterizes a very broad set of  curvatures around the optimal price $\ps$. We will show that the local behavior of $R(v)$, as captured by  $\alpha$, is a key driver of performance. %

Finally, we note that \Cref{def1} implies that
\begin{equation}\label{eq:def1-intuitive}
\kappa_L\cdot |v - \ps|^\alpha \le R(\ps) - R(v) \le \kappa_U \cdot|v - \ps|^\alpha.
\end{equation}
Moreover if $R(p)$ is locally concave, the condition above and \Cref{def1} are equivalent (see \Cref{sec:app_rev_aux-11} in the appendix for a formal statement and proof).

\paragraph{Notation.} We next introduce some notation  that will be used in the rest of the paper. In what follows we will be interested in quantities as $\eps \downarrow 0$. Let $g$ and $h$ be two functions, we write
$g(\eps)=\mathcal{O}(h(\eps))$  if there exists a positive constant $C$ and $\eps_0>0$ such that  $g(\eps)\leq C\cdot h(\eps)$ for all $\eps \in (0,\eps_0]$.
Additionally,
$g(\eps)=\tilde{\mathcal{O}}(h(\eps))$  if there exists a positive constant $C$ and $\eps_0>0$ such that  $g(\eps)\leq C\cdot h(\eps)\cdot\log(1/\eps)$ for all $\eps \in (0,\eps_0]$.
We write
$g(\eps)=\Omega(h(\eps))$  if there exists a positive constant $C$ and $\eps_0>0$ such that  $C\cdot h(\eps)\le g(\eps)$ for all $\eps \in (0,\eps_0]$. %

\subsection{Initial Remarks on the Impact of Relaxing IC Constraints}
Before we proceed with our main analysis we elaborate on the challenges associated with the $\eps$-IC mechanism design problem.
The coming discussion will also allow to lay down the main distinctive  features of \eqref{eq:probeps} compared to
$(\mathcal{P}_0)$.
In particular,  we will highlight how the arguments and reductions from classical mechanism design  no longer apply as soon as $\eps$ becomes positive.

The classical mechanism design problem, $(\mathcal{P}_0)$, is amenable to a reformulation that greatly simplifies its analysis and that leads to a simple posted-price solution. Fundamentally, there are three key properties that enable this:
\begin{enumerate}

\item[(P1)] The allocation $\all(\cdot)$ is monotone non-decreasing. This is a consequence of $(\text{IC}_0)$.

\item[(P2)] By $(\text{IC}_0)$, truthful reporting is an optimal reporting strategy. This obviates the need to keep track of the buyer's optimal reporting strategy.

\item[(P3)] The transfers $t(\cdot)$ can be expressed as a linear functional of the allocation $x(\cdot)$. This is obtained by leveraging $(\text{IC}_0)$ in conjunction with an application of the envelope theorem.%

\end{enumerate}

These three properties lead to a reduction of $(\mathcal{P}_0)$  as
\begin{flalign}\label{eq:probstan}%
\max_{\all(\cdot)}& \left\{ \int_{\mathcal{S}} \all(v)\psi(v)f(v)dv  \: : \: \text{s.t}\quad  \all(\cdot) \text{ is non-decreasing} \right\},
\end{flalign}
where we denote the \textit{virtual value} function by $\psi:\mathcal{S}\rightarrow \RR$ with $\psi(v) \triangleq v-(1-F(v))/f(v)$ for all $v\in \mathcal{S}$. This reformulation was developed and solved using ironing in \cite{myerson1981optimal}. A more direct approach (without ironing) can be found in \cite{rileyzeckhauser1983}. It can be shown that the  optimal allocation is a posted-price mechanism $\ps$ that maximizes the revenue function
\begin{equation*}%
\ps \in \arg\max_{v\in\mathcal{S}} R(v). %
\end{equation*}

Let us now consider the case $\eps>0$ and consider the impact on the three properties above.  First, $(\text{IC}_0)$ implies that the allocation is a non-decreasing function, see property (P1).  However, as soon as  $\eps>0$,  a feasible allocation is not necessarily a non-decreasing function. The incentive compatibility constraint $(\text{IC}_\eps)$ only implies the following  ``approximate'' monotonicity property
\begin{equation*}%
(v-v')(\all(v) - \all(v')) \ge -2\eps,\quad \forall v,v'\in \mathcal{S}.
\end{equation*}

Second,  $(\text{IC}_0)$ implies that it is optimal for the buyer to report his type, see property (P2). When $\eps>0$, this is no longer necessarily true. Let  $u(v)=\max_{w\in \mathcal{S}} \{vx(w)-t(w)\}$ be the optimal utility of a buyer with value $v$,
and  $\vsr(v):\mathcal{S}\rightarrow \mathcal{S}$  be a \textit{best reporting mapping} of the buyer defined by
\begin{equation}\label{eq:best-str-Map}
\vsr(v)\in \argmax_{w\in \mathcal{S}}\Big\{v\cdot \all(w)-\tra(w)\Big\}.
\end{equation}
While  $(\text{IC}_0)$ implies that $\vsr(v)=v$ is a best reporting strategy, $(\text{IC}_\eps)$ only implies
\begin{equation}\label{eq:IC-no-tight}
v\cdot \all(v)-\tra(v)\ge u(v)-\eps, \quad \forall v\in \mathcal{S}.
\end{equation}
That is, one only knows that truthful reporting leads to a utility in the interval $[u(v)-\eps,u(v)]$. %
Consequentially, in the $\eps$-IC setting, from an optimization perspective, the need to account for the endogenous object $\vsr(\cdot)$ is not obviated. As a matter of fact, we will see that this object plays a crucial role in both deriving impossibility results but also in optimizing over mechanisms.

Finally, another important difference that emerges relates to the transfers $\tra(\cdot)$, and the lack of a strong  characterization  readily available such as in property (P3). In particular, when $\eps>0$, according to the envelope theorem (see e.g., \citet{milgrom2002envelope}) we have
\begin{equation}\label{eq:env-report}
u(v) = u(0)+\int_{0}^{v}\all(\vsr(s))ds, \quad \forall v\in \mathcal{S}.
\end{equation}
   By combining \cref{eq:IC-no-tight} and \cref{eq:env-report} we only obtain an inequality for $t(\cdot)$ as a function of $x(\cdot)$, $u(0)$ and $\vsr(\cdot)$,  as opposed to a full characterization of $\tra(\cdot)$. That is, we have
   \begin{equation}\label{eq:env-UB-t}
\tra(v) \:\le\: v\cdot \all(v)- u(0) -\int_{0}^{v}\all(\vsr(s))ds+\eps, \quad \forall v\in \mathcal{S}.
\end{equation}

As the above illustrates, as soon as $\epsilon>0$,  the seller's problem becomes one of a different nature. In particular, in the $\eps$-IC case, one must solve jointly for the allocation and transfers by also tracking a best reporting mapping. As we shall see, the latter will play a critical role in our approach.

\section{Upper Bound}\label{sec:upperboundper}

We begin our analysis of \eqref{eq:probeps} for general mechanisms by developing an upper bound for \cref{eq:key-Metric}.
As an initial benchmark, we review a known upper bound of $\mathcal{O}(\eps^{1/2})$ . This known upper bound provides an initial range of what is achievable for the optimal gains on the revenue improvement when the buyer is not a perfect optimizer. Then, we establish our first main result: a new instance dependent
upper bound of  a different order
$\tilde{\mathcal{O}}(\eps^{\alpha/(2\alpha-1)})$ on the revenue improvement that is based on  a novel approach and techniques.

\subsection{An Initial Upper Bound via Nisan's Rounding Argument} \label{sec:Nisan}
We first present an  upper bound  that follows directly from  a classical rounding argument,  attributed to Noam Nisan (we are grateful to Anonymous for this reference).  In mechanism design, the rounding argument has been used extensively to show that an approximate incentive-compatible mechanism can be transformed into an exact incentive-compatible mechanism (see e.g., \citealt{balcan2005mechanism}).
Here, one may use it to bound the maximal revenues one could gain when solving under approximate, as opposed to exact, incentive compatibility. We provide a proof in the appendix for completeness.

\begin{proposition}[Upper Bound via Rounding]\label{prop:smoothing}
\begin{equation*}
\Pi^\star(\mathcal{M}(\eps)) \: -\: \Pi^\star(\mathcal{M}(0)) \:=\:\\ \mathcal{O}(\eps^{1/2}).
\end{equation*}
\end{proposition}
The rounding argument provides a first step towards gauging the impact of imperfect optimizers on the seller's revenue. It provides an instance independent impossibility result on the incremental revenues that can be garnered by the seller.
In the present paper, we aim at understanding the instance dependent optimal performance, and, in turn, a natural question is whether the bound above is tight once we specify a family of distributions. In other words, we ask whether it is possible to improve upon this bound by considering the characteristics of a given family of distributions. Furthermore, we are also interested in \textit{how} to achieve optimal or near-optimal performance.
 We next develop a novel path-based duality  approach to obtain a new impossibility result that leads to an improved upper bound for \cref{eq:key-Metric}.

\subsection{Upper Bound via Path-Based Duality Approach}\label{sec:path-based-upper-bound}

Our first main theorem provides a new impossibility result. %

\begin{theorem}[Performance Upper Bound]\label{thm:up-bd-per} Suppose \Cref{assumption1} holds, and that the  revenue function admits local $\alpha$-power envelopes in a neighborhood $\mathcal{N}_\ell$ of $\ps$ and
$
\inf_{p\notin \mathcal{N}_\ell} |\dot{R}(p)|>0.
$
Then
\begin{equation*}
\Pi^\star(\mathcal{M}(\eps)) \: -\: \Pi^\star(\mathcal{M}(0)) \:=\:\tilde{\mathcal{O}}(\eps^{\alpha/(2\alpha-1)}).
\end{equation*}
\end{theorem}
This result shows that the seller can benefit by at most $\eps^{\alpha/(2\alpha-1)}$ up to polylogarithmic terms by offering a mechanism that is $\eps$-incentive compatible. (We conjecture that  it might be possible to remove the polylogarithmic terms from the upper bound.)  \Cref{thm:up-bd-per} presents an instance dependent impossibility result in the presence of imperfect optimizer. Furthermore, our family of bounds (across of $\alpha$) establishes  that the limit on achievable performance is driven by the local behavior of the revenue function around the optimal posted price.

Additionally, since for any $\alpha\in(1,+\infty)$ we have that $1/2< \alpha /(2\alpha-1)$, our result also implies the $\eps^{1/2}$ scaling (obtained through the rounding argument) for the class of distributions that satisfies the assumptions of the theorem. We note that the last assumption in \Cref{thm:up-bd-per} implies that the revenue function does not have local maxima besides $\ps$.

\subsubsection{Main proof ideas}
Our upper bound is based on an appropriate relaxation that can be analyzed through a duality argument. For our duality approach, we will relax most IC constraints except those that go through a carefully constructed path. This is one of our critical observations and it is what makes the Proof of \Cref{thm:up-bd-per} of a different nature than previous approaches for classical mechanism design problems. Indeed, constructing the path is analogous to choosing a best reporting function for the $\eps$-IC problem.  As we will see, this is an intricate task that requires going beyond approaches that consider local deviations in the incentive compatibility constraints and, instead, entails optimizing over best-response functions. Next, we explain our approach in detail.

\paragraph{Path-based relaxation of IC constraints.} To construct the path, let us first fix $\mu$ and $\nu_0$ such that $0 < \mu<\nu_0 < \bv$. The values of $\mu$ and $\nu_0$ will parametrize the path. Define the candidate best response function $\vsr:[\mu,\bv]\rightarrow [0,\nu_0]$ such that $\vsr(v)$ is strictly increasing and continuous, $\vsr(\mu)=0$ and $\vsr(\bv)=\nu_0$. Later we show that this candidate is indeed a best response function. We also assume that
$\vsr(v)< v$ for all $v$ so that the best response always under-reports. The function $w:[0,\nu_0]\rightarrow [\mu,\bv]$ is the inverse of $\vsr(\cdot)$. We note that $\vsr(\cdot)$ is constructed to have all the properties that are desirable on a reporting function. Later, we will optimize over choices of the parameters that satisfy the aforementioned properties.

 Given $\mu$ and $\nu_0$, we consider the following  relaxed version of  \eqref{eq:probeps} in which the IR constraint is only imposed on $[0,\mu]$ and the $\eps$-IC constraint is imposed on the path given by $\vsr(\cdot)$:
\begin{flalign}\label{eq:probepsdel}\tag{$\mathcal{P}_{\text{path}}$}
\max_{\all(\cdot),\tra(\cdot)} & \E_v[\tra(v)]\\\label{eq:IRd} \tag{$\text{IR}_\mu$}
\text{s.t. } &  v\cdot \all(v)-\tra(v)\geq 0,\quad v\in [0,\mu]  \\\label{eq:ICd} \tag{$\text{IC}_\mu$}
&v\cdot \all(v)-\tra(v)\geq
 v\cdot \all(\vsr(v))-\tra(\vsr(v))
-\eps,\quad \forall v\in [\mu,\bv],\\\nonumber
&\all:\Theta\rightarrow [0,1] \quad \text{and}\quad \tra:\Theta\rightarrow \RR.
\end{flalign}
We let $\mathcal{M}_{\text{path}}(\eps)$ denote the set of feasible mechanisms.
Clearly, $\Pi^\star(\mathcal{M}(\eps))\leq \Pi^\star(\mathcal{M}_{\text{path}}(\eps))$. The first set of constraints in \eqref{eq:probepsdel} come from the individual rationality constraints in the original problem. However, we only consider \eqref{eq:IR} for types with low values (below $\mu$). Intuitively, individual rationality constraints are always binding for low  value buyers. Meanwhile, incentive compatibility constraints are in general binding for high value buyers. This leads to the  second set of constraints which come from the original approximate incentive compatibility constraints. Here we consider the pairs $(v,\vsr(v))$ for $v\ge \mu$. That is, for any type  $v\geq \mu$ we consider a downward misreport of size $v-\vsr(v)$. The latter choice captures an essential property of most optimal contracts in mechanism design, namely, higher types have an incentive to report low values and so the mechanism should prevent such deviations.
We note that there are many choices for an upper bound problem such as  \eqref{eq:probepsdel}. However, the above problem strikes a balance in the sense that it captures key properties of optimal mechanisms (we keep the constraints that are most likely to bind), and it also parametrizes these properties in a simple and, as we will shortly see, tractable way.

For each constraint \eqref{eq:IRd} and \eqref{eq:ICd}, we introduce  dual variables $\lambda^\ir : [0,\mu] \rightarrow \mathbb R_+$ and
 $\lambda^\ic : [\mu,\bar v] \rightarrow \mathbb R_+$, respectively. After dualizing these constraints and optimizing over $(\all,\tra)$ we obtain the following result.
 \begin{lemma}[Weak duality]\label{lem:weak-dual-upper}
 Consider the dual problem
\begin{flalign}
\label{eq:dualproblem}\tag{$\mathcal{D}$}
 \Pi^\star(\mathcal{D})\triangleq \min_{
 \substack{
 \Lic : [\mu,\bv] \rightarrow \mathbb R_+
}
 }& \Phi_1(\Lic)+ \Phi_2(\Lic)\\\nonumber
\text{s.t.}\:\:\:\:\:\:\:\: & \Lic(v) = f(v),\quad \text{a.e on } [\nu_0,\bv], \\ \nonumber
&\Lic(v) = f(v) + \dot{\w}(v)\Lic(\w(v)),\quad \text{a.e on } [\mu,\nu_0],
 \end{flalign}
 where
 \begin{equation*}
\Phi_1(\Lic)=\eps \int_{\mu}^{\bar v} \Lic(v) \text{d} v\quad \text{and}\quad
\Phi_2(\Lic)=\int_{0}^{\bar v} \Big( v f(v) - \dot{\w}(v)(v-\w(v)) \Lic(\w(v) ) \1{v \in [0, \nu_0]} \Big)^+ \text{d} v.
  \end{equation*}
 Then,
 \begin{equation*}
  \Pi^\star(\mathcal{M}_{\text{path}}(\eps))\leq \Pi^\star(\mathcal{D}).
 \end{equation*}
 \end{lemma}
The result above is obtained by first forming the Lagrangian for \eqref{eq:probepsdel}. Since the transfer $\tra(\cdot)$ is a free variable, all terms that multiply $t(v)$ must equal zero for (almost) every realization of the values. This leads to a set of conditions the dual variables should satisfy which, in turn, completely pin down the dual variable of the IR constraint $\Lir$ and give a functional equation that determines the dual variable $\Lic$. Therefore, in the dual problem we can eliminate $\Lir$ and state the problem in terms of $\Lic$ directly. Our dual problem is an infinite-dimensional mathematical program with constraints given by functional equations.

 The new upper bound objective has two terms. The first, $\Phi_1(\lambda^\ic)$, comes from the approximate incentive constraint \eqref{eq:ICd} while the second term, $\Phi_2(\lambda^\ic)$, comes from optimizing $\all(v)\in [0,1]$. Due to the functional equation that defines $\Lic$,  problem \eqref{eq:dualproblem} does not have a straightforward solution. Moreover, its solution depends on our choice of $\w(\cdot)$ which we have not explicitly defined yet. In order to circumvent this, we provide further bounds for the dual problem. In what follows, we develop structural properties for \eqref{eq:dualproblem}. These properties will provide guidelines for how to choose $\w(\cdot)$ and will lead us to develop a tight  bound for $\Pi^\star(\mathcal{D})$.

\paragraph{Structural properties.} We provide some definitions and a lemma that simplifies $\Lic(\cdot)$. We define a sequence of thresholds $\{\nu_k\}$ iteratively by
\begin{equation*}
\nu_{k} = w(\nu_{k+1}), \quad k\geq 0, \quad \nu_{-1}=\bv.
\end{equation*}
We let $K$ be the first index such that $\nu_K\leq \mu$ and $\nu_{K-1}\geq \mu$. Since $w(\cdot)$ is strictly increasing, this sequence is well defined. Also, because $w(x)\ge  x$ we have that
$\nu_k\geq \nu_{k+1}$. We will refer to $K$ as the number of steps while every interval of the form $[\nu_{k},\nu_{k-1}]$ will be referred to as a step.  Intuitively, $K$ gives the number of times we need to iterate $w(\cdot)$ to reach the origin starting from $\nu_0$. Geometrically, this sequence is constructed by repeatedly moving leftward and downward until we touch the best response and the $45^\circ$ line, respectively. See \Cref{fig:steps} for a graphical representation of $\{\nu_k\}$.
\begin{figure}[t]
\begin{center}
\scalebox{0.9}{\begin{tikzpicture}[baseline=0pt]

\draw[->,line width=0.3mm] (0,0)--(0,5.5);
\draw[->,line width=0.3mm] (0,0)--(5,0);
\node at (-0.5,5.2){$\w(v)$};
\node at (5.1,-0.25){$v$};

\draw[dashed,line width=0.4mm](4,0)--(4,4.5);
\node at (4,-0.25) {$\nu_0$};
\draw[dashed,line width=0.4mm](4,4.5)--(0,4.5);
\node at (-0.25,4.5) {$\bv$};

\draw[line width=0.4mm](0,0)--(4,4);
\draw
    (3,0) coordinate (a) node[right] {}
    -- (0,0) coordinate (b) node[left] {}
    -- (2,2) coordinate (c) node[above right] {}
    pic[" $45^{\circ}$", draw=orange, <->, angle eccentricity=0.7, angle radius=2.35cm]
    {angle=a--b--c};

\draw[line width=0.4mm, domain=0:4, smooth, variable=\x, blue] plot ({\x}, {4.5/(1+0.275*(\x-4)*(\x-4))});

\draw[line width=0.3mm,red](4,4)--(3.3,4);
\draw[line width=0.3mm,red](3.3,4)--(3.3,3.3);
\draw[line width=0.3mm,red](3.3,3.3)--(2.85,3.3);
\draw[line width=0.3mm,red](2.85,3.3)--(2.85,2.85);
\draw[line width=0.3mm,red](2.85,2.85)--(2.55,2.85);
\draw[line width=0.3mm,red](2.55,2.85)--(2.55,2.55);
\draw[line width=0.3mm,red](2.55,2.55)--(2.3,2.55);
\draw[line width=0.3mm,red](2.3,2.55)--(2.3,2.3);
\draw[line width=0.3mm,red](2.3,2.3)--(2.1,2.3);
\draw[line width=0.3mm,red](2.1,2.1)--(2.1,2.3);
\draw[line width=0.3mm,red](2.1,2.1)--(1.95,2.1);
\draw[line width=0.3mm,red](1.95,1.95)--(1.95,2.1);
\draw[line width=0.3mm,red](1.95,1.95)--(1.825,1.95);
\draw[line width=0.3mm,red](1.825,1.825)--(1.825,1.95);
\draw[line width=0.3mm,red](1.825,1.825)--(1.7,1.825);
\draw[line width=0.3mm,red](1.7,1.7)--(1.7,1.825);
\draw[line width=0.3mm,red](1.7,1.7)--(1.55,1.7);
\draw[line width=0.3mm,red](1.55,1.55)--(1.55,1.7);
\draw[line width=0.3mm,red](1.55,1.55)--(1.4,1.55);
\draw[line width=0.3mm,red](1.4,1.4)--(1.4,1.55);
\draw[line width=0.3mm,red](1.4,1.4)--(1.15,1.4);
\draw[line width=0.3mm,red](1.15,1.15)--(1.15,1.4);
\draw[line width=0.3mm,red](1.15,1.15)--(0.8,1.15);
\draw[line width=0.3mm,red](0.8,0.8)--(0.8,1.15);
\draw[line width=0.3mm,red](0.78,0.78)--(0.0,0.78);

\draw[line width=0.3mm,blue](0,0.85)--(0.85,0.85);
\node at (-0.25,0.85) {$\mu$};

\draw[dashed,line width=0.3mm,red](3.3,0)--(3.3,3.3) node at (3.3,-0.25) {$\nu_1$};
\draw[dashed,line width=0.3mm,red](2.85,0)--(2.85,2.85) node at (2.9,-0.25) {$\nu_2$};

\draw[dashed,line width=0.3mm,red](2.55,0)--(2.55,2.55) node at (2.5,-0.25) {$\nu_3$};

\draw[dashed,line width=0.3mm,red](1.15,0)--(1.15,1.15) node at (1.45,-0.25) {$\nu_{K-1}$};

\draw[dashed,line width=0.3mm,red](0.78,0)--(0.78,0.78) node at (0.78,-0.25) {$\nu_K$};

\begin{scope}[xshift=7cm]
\draw[->,line width=0.3mm] (0,0)--(0,5.5);
\draw[->,line width=0.3mm] (0,0)--(8,0);
\node at (8.1,-0.25){$v$};

\draw[dashed,line width=0.3mm] (0,2.1)--(8,2.1);
\node at (-0.2,2.1) {$0$};

\draw[dotted,line width=0.4mm,shift={(0,0.1)}] plot file {upperbound_1.data};
\draw[densely dashed,blue,line width=0.4mm,shift={(0,0.1)}] plot file {upperbound_2_up.data};
\draw[red,loosely dashed,line width=0.4mm,shift={(0,0.1)}] plot file {upperbound_3_low.data};
\draw[ black,line width=0.4mm,shift={(0,0.1)}] plot file {upperbound_4_delayed.data};

\draw[line width=0.3mm] (4.5,0)--(4.5,2.1);
\node at (4.5,-0.2) {$\ps$};

\draw[line width=0.3mm] (2.2,0)--(2.2,2.1);
\node at (2.2,-0.2) {$\ps_U$};

\draw[line width=0.3mm] (5.2,0)--(5.2,2.1);
\node at (5.2,-0.2) {$\ps_L$};

\draw[black,line width=0.4mm,shift={(0,0.1)},name path=A] plot file {upperbound_4_delayed_region.data};
\draw[black,line width=0mm,name path=B]  (5.2,2.1)--(7.5,2.1);
\tikzfillbetween[of=A and B]{red, opacity=0.2}; 

\draw[densely dashed,blue,line width=0.4mm,shift={(0,0.1)},name path=C] plot file {upperbound_2_up_region.data};
\draw[black,line width=0mm,name path=D]  (2.2,2.1)--(5.2,2.1);
\tikzfillbetween[of=C and D]{blue, opacity=0.2};

\draw[densely dashed, blue,line width=0.3mm] (0.5,4.5)--(1.0,4.5) node at (1.6, 4.5) {$\Delta_U(v)$};
\draw[dotted,line width=0.3mm] (0.5,4.5-0.5)--(1.0,4.5-0.5) node at (1.6, 4.55-0.5) {$-\dot{R}(v)$};
\draw[black,line width=0.3mm] (0.5,4.5-1)--(1.0,4.5-1) node at (1.6, 4.5-1) {$\Delta(v)$};
\draw[loosely dashed,red,line width=0.3mm] (0.5,4.5-1.5)--(1.0,4.5-1.5) node at (1.6, 4.5-1.5) {$\Delta_L(v)$};
\end{scope}

\node at (2,-1) {\textbf{(a)}};
\node at (7+4,-1) {\textbf{(b)}};

\end{tikzpicture}}
\end{center}
\caption{\textbf{(a)} Illustration of $\w$---the inverse of the best-response mapping $\vsr(\cdot)$---and the sequence $\{\nu_k\}$. In our relaxation \eqref{eq:probepsdel}, we drop all IC constraints except for the pairs $(v,\vsr(v))_{v \in [\mu, \bv]}$ or, equivalently, $(w(v), v)_{v \in [0, \nu_0]}$.
\textbf{(b)} Illustration for the bound on $\Phi_2(\Lic)$---the integral of $\Delta(v)$. For $v\ge \ps_L$, we bound the integral using \Cref{lem:bound_rev}, and for $v$ in $[\ps_U,\ps_L]$ we bound it with the area under $\Delta_U(v)$.
}
\label{fig:steps}
\end{figure}
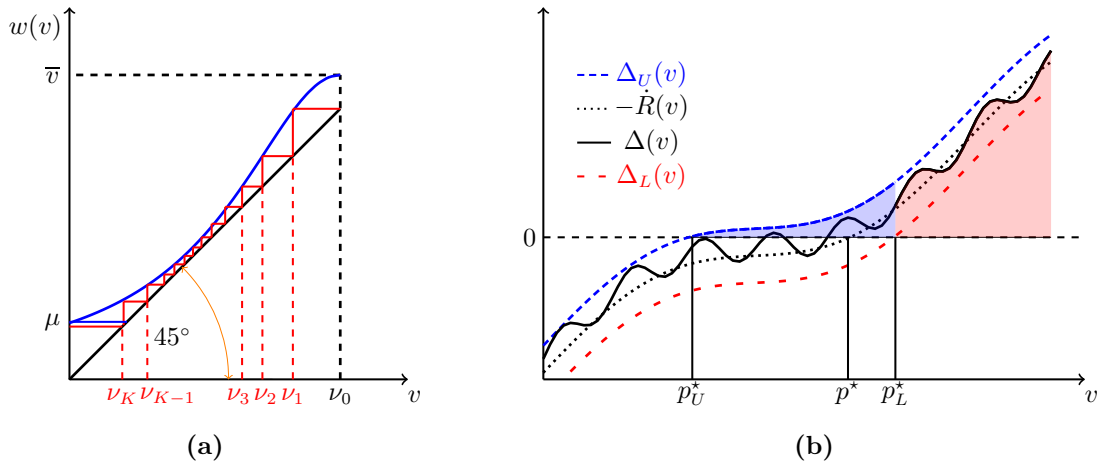

The purpose of defining the thresholds $\{\nu_k\}$ is that they allow to look at the functional equation that defines $\Lic$ in \eqref{eq:dualproblem} from a different perspective. In particular, we can obtain a quasi-closed form solution for $\Lic$ that depends on these thresholds by iteratively solving for $\Lic$ in each step.

\begin{lemma}[Solution to Functional Equation]
\label{lem:charac_lambda_ic}   A solution $\Lic(\cdot)$ to Problem \eqref{eq:dualproblem} is given by, for any $k\in \{0,\dots,K\}$ and for any $v\geq \mu$,
\begin{equation*}
\Lic(v) = \frac{d}{dv}\sum_{j=0}^k F(w^j(v)), \quad \text{a.e on } [\nu_k,\nu_{k-1}].
\end{equation*}
where $w^j(v)$ is the $j$th composition  of $w$ and $w^0(v)=v$.
\end{lemma}
Armed with \Cref{lem:charac_lambda_ic} we can derive upper bounds for $\Phi_1(\Lic)$ and
$\Phi_2(\Lic)$ that depend solely on properties of $\w(\cdot)$.
\begin{lemma} [Bound on Number of Steps]
\label{lem:steps}
There exists a constant $C>0$ such that
\begin{equation*}
\Phi_1(\Lic) \leq C\cdot\eps\cdot K.
\end{equation*}
\end{lemma}

\Cref{lem:steps} asserts that $\Phi_1(\Lic)$ is bounded above by $\eps\cdot K$ where $K$ is the number of steps.  We prove the result by using our formula for $\Lic$ given in \Cref{lem:charac_lambda_ic}, which allows us to easily integrate the dual variable in each step, and then telescoping the resulting sums.

\begin{lemma}[Dual Virtual Value Bound]
\label{lem:bound_rev}
Define the integrand of $\Phi_2(\Lic)$ by
\begin{equation*}
\Delta(v)\triangleq vf(v)-\dot{w}(v)(w(v)-v)\Lic(w(v)) \1{v \in [0, \nu_0]}.
\end{equation*}
Suppose that $\Delta(v)$ is non-negative in $[x,\nu_0]$ for some $x\in [0, \nu_0)$ then
\begin{equation*}
\Phi_2(\Lic) = \int_{x}^{\bv} \Delta(v) dv \leq R(x)+2\cdot (\w(x)-x).
\end{equation*}
\end{lemma}
\Cref{lem:bound_rev} is reminiscent of the analysis in classical mechanism design where it is shown that the integral of the virtual value times the density in an interval $[x,\bv]$ equals  the revenue function $R(x) = x\cdot \bar{F}(x)$.  However, in our dual problem, we do not have the standard virtual value $\psi(v)$ in the objective but a modified ``dual virtual value'' given by $\Delta(v)$, whose integral can be bounded in terms of the revenue function and an additional  correction term proportional to  $\w(x)-x$.

\paragraph{Designing a best-response mapping.} We leverage the structural properties developed in \Cref{lem:steps} and \Cref{lem:bound_rev} to build a candidate $\w(\cdot)$ that is both easy to manipulate---it allows for a closed-form approximation for $\Lic$---and that yields a tight upper bound. The previous lemmas  provide two guidelines for choosing $\w(\cdot)$. First, to apply \Cref{lem:bound_rev} we need to design the inverse reporting $w(\cdot)$ such that $\Delta(v)$ crosses zero only once and, at the crossing point, its distance to the $45^{\circ}$ line as measured by $w(x) - x$ is of order  $\eps^{\alpha/(2\alpha-1)}$. This implies that, as $\eps$ goes to zero, the inverse reporting function must get closer to the $45^{\circ}$ line. Second, from \Cref{lem:steps}, $\w(\cdot)$ must be such that the number of steps scales, at most, with order $\eps^{(1-\alpha)/(2\alpha-1)}$ (which goes to infinity as $\eps\downarrow0$) to yield a revenue improvement of $\eps^{\alpha/(2\alpha-1)}$. Therefore, because the number of steps increases as the inverse reporting function gets closer to the $45^{\circ}$ line, $w(\cdot)$ should not approach the $45^{\circ}$ line too fast. The two lemmas present a trade off: we would like to set $\w(x)-x$ small but not too small otherwise $K$ may be too large.

 A first, natural approach would be to choose $\w(x)=x+\delta$ with $\delta>0$, i.e., the buyer is under-reporting by an amount of $\delta$. This choice would converge to the truthful reporting function as $\delta$ decreases. Because the number of steps is proportional to $1/\delta$, assuming that we can use the lemmas above, it is possible to show that
$$
\Phi_1(\Lic)+\Phi_2(\Lic)\leq \frac{C\cdot \eps}{\delta} + \Pi^\star(\mathcal{M}(0)) + 2\cdot \delta,
$$
where we have used that, for any $x$, $R(x) \le \Pi^\star(\mathcal{M}(0))$.
 With this choice of $\w$ it is optimal to set $\delta \approx \epsilon^{1/2}$, which yields the upper bound $\mathcal{O}(\eps^{1/2})$ . (Choosing a smaller value of $\delta$ would necessarily lead to a too-large number of steps.) This simple choice of $\w(\cdot)$ is, unfortunately, not tight but provides an alternative proof to \Cref{prop:smoothing} using our duality based approach. The strength of our approach is that it provides  the flexibility to design the best-response function and choose a $\w(\cdot)$ that leads to an improved bound. To do so, we need $\w(\cdot)$ to be closer to the $45^\circ$ line while, at the same time, preventing the number of steps from growing too fast. The previous analysis shows this is impossible to achieve with a linear $\w(\cdot)$. We address this challenge by making $\w(\cdot)$ piecewise linear with a kink at $\ps$. By carefully adjusting the slope of $\w(\cdot)$ we can guarantee that the number of steps does not grow too fast while approaching the $45^\circ$ line at the right rate.

Our choice of $\w(\cdot)$ is given by\
\begin{equation}\label{eq:w_for_up_bound}
\w(v) =\twopartdef{m\cdot (v-\ps) + \ps+\eps^{1-\beta}}{v\geq \ps;}
{(2-m)\cdot (v-\ps) + \ps+\eps^{1-\beta}}{v\leq \ps,}
\quad
m=\frac{\bv-(\ps+\eps^{1-\beta})}{\nu_0-\ps},
\quad \nu_0=\bv-\eps^{\beta},
\end{equation}
for $\beta\in (0,1/2)$. The closest point of $\w(\cdot)$ to the $45^\circ$ line is at $\ps$ where the distance
 is $\eps^{1-\beta}$. The slopes of $\w(\cdot)$ are $m>1$ and $2-m<1$ to the right and left of $\ps$, respectively.
 This choice of $\w(\cdot)$ trades off tractability and tightness of our bound. Indeed, because $\w(\cdot)$ is a piecewise linear function we can obtain a simple approximation for $\Lic$ and, ultimately, analyze $\Delta(v)$ in a tractable way.  In order to obtain a tight bound, we constructed $\w(\cdot)$ such that as $\eps\downarrow 0$, $\w(\cdot)$ approaches the $45^\circ$ line in two ways: the slope converges to one and the distance converges to zero. By doing so, as we explain next, we can guarantee that $\w(\cdot)$ balances the number of steps $K$ and the difference $\w(x)-x$ in \Cref{lem:steps} and \Cref{lem:bound_rev}.

For $\w(\cdot)$ defined above, we leverage \Cref{lem:steps} and \Cref{lem:bound_rev} to obtain closed form bounds in terms of $\eps$ for the dual problem. We begin by bounding the number of steps.
\begin{proposition}[Number of Steps]\label{prop:numberOfSteps} Under $\w(\cdot)$ given by \cref{eq:w_for_up_bound} the number of steps, $K$, satisfies
$K=\mathcal{O}\left(\log(1/\eps)/\eps^{\beta}\right).$
\end{proposition}
The intuition behind the proposition is simple.
For ease of exposition, consider the number of steps above $\ps$.
Since $\w(\cdot)$ is linear with slope $m$, it must be that
$(\w(\nu_k)-\w(\nu_{k+1}))/(\nu_k-\nu_{k+1})=m$. But by the definition of $\nu_k$ this is equivalent to $(\w(\nu_{k+1})-\nu_{k+1})=(\w(\nu_k)-\nu_k)/m$. That is, the distance between $\w(\cdot)$ and the $45^\circ$ degree line decreases by a factor of $1/m$ in every step. In general, we have that $(\w(\nu_{k})-\nu_{k})=(\bv-\nu_0)/m^k$.
Because $\w(\ps)-\ps=\eps^{1-\beta}$,
the number of steps above $\ps$ is the value of $K$ such that $(\w(\nu_{K})-\nu_{K}) \approx \eps^{1-\beta}$  and, therefore,
\begin{equation*}
    K\approx \frac{\log((\bv-\nu_0)/\eps^{1-\beta})}{\log(m)}=\mathcal{O}\left(\log(1/\eps)/\eps^{\beta}\right),
\end{equation*}
where the last equality comes from a Taylor series around $\eps=0$ and from $m-1\approx \eps^{\beta}$. By combining this result and \Cref{lem:steps}, we deduce that $\Phi_1(\Lic)$ is upper bounded by a term proportional to $\log(1/\eps)\cdot \eps^{1-\beta}$.

Next, we use \Cref{lem:bound_rev} to bound $\Phi_2(\Lic)$. The key idea of our argument is to show that, given our choice of $\w(\cdot)$, there exists a vanishing neighborhood  $(\ps_U, \ps_L)$ of size $\eps^{\beta/(\alpha-1)}$ around $\ps$ such the integrand in $\Delta(v)$ is negative below and positive above the neighborhood.  Below the neighborhood, the integral is zero because we take the positive part of $\Delta(v)$. Inside the neighborhood, the integral of $\Delta(v)$ is at most $\mathcal{O}(\eps^{\beta+ \beta/(\alpha-1)})$ because both the integrand and the size of neighborhood are small. Above the neighborhood, we can use \Cref{lem:bound_rev} to obtain an upper bound for
$\Phi_2(\Lic)$ which, together with the fact that
$\w(v)-v\approx \eps^{1-\beta}$ for $v\approx \ps$, obtains an upper bound of  $\mathcal{O}(\eps^{\gamma})$, with $\gamma = \min\{1-\beta,\beta+ \beta/(\alpha-1)\}$. We illustrate this in \Cref{fig:steps} (a) and provide a more detailed explanation next.

By developing an approximation for $\Lic$, we can establish that
\begin{equation*}
    \Delta_L(v)\triangleq vf(v)-\bar{F}(v) - K_1\cdot \eps^{\beta}\leq \Delta(v) \leq vf(v)-\bar{F}(v) +K_2\cdot \eps^{\beta}\triangleq \Delta_U(v),
\end{equation*}
almost everywhere, that is, $\Delta(v)$ converges uniformly (almost everywhere) to the negative of the derivative of the revenue function $-R'(v) = vf(v)-\bar{F}(v)$. Note that we cannot simply use $\Delta_U(v)$ to bound $\Phi_2(\Lic)$ because
 that would lead to a bound of $\mathcal{O}(\eps^{\beta})$. Instead, we can use $\Delta_L(v)$ and $\Delta_U(v)$ to approximate where $\Delta(v)$ crosses zero and then apply \Cref{lem:bound_rev}.

 Using the lower bound above and %
the local $\alpha$-power envelopes of the revenue function given in \Cref{def1}, we can determine the right endpoint of our neighborhood to be $\ps_L \approx \ps + \eps^{\beta/(\alpha-1)}$ above which $\Delta_L(v)$
is non-negative and, therefore, $\Delta(v)$ is non-negative  above $\ps_L $. By \Cref{lem:bound_rev}, we have
\begin{equation}\label{eq:delta_long_int}
    \int_{\ps_L}^{\bv}\Delta(v)dv\leq R(\ps_L) +2(w(\ps_L)-\ps_L). %
\end{equation}
The left endpoint of the neighborhood is similarly determined to be $\ps_U \approx \ps - \eps^{\beta/(\alpha-1)}$ below which $\Delta(v)$ can be guaranteed to be non-positive. Now, within the neighborhood we cannot control the sign of $\Delta(v)$; however, we can make use of $\Delta_U(v)$ as follows:
\begin{equation}\label{eq:delta_small_int_a}
    \int_{\ps}^{\ps_L}\Delta(v)^+dv\leq \int_{\ps}^{\ps_L}\Delta_U(v)dv =
    R(\ps)-R(\ps_L)+\mathcal{O}(\eps^{\beta+ \beta/(\alpha-1)}),
\end{equation}
 by integrating the derivative of the revenue function which is negative above $\ps$ and using that $[\ps,\ps_L]$ has size $\eps^{\beta/(\alpha-1)}$.
 In $[\ps_U,\ps]$, $\Delta_U(v)^+$ is bounded above by $K_2 \eps^\beta$ because the derivative of the revenue function is positive in this interval. Hence, in this interval, the integral of $\Delta(v)^+$ is at most of size $\mathcal{O}(\eps^{\beta+ \beta/(\alpha-1)})$.
 Combining this with \eqref{eq:delta_long_int} and
 \eqref{eq:delta_small_int_a}:
 \begin{equation*}
    \int_{0}^{\bv}\Delta(v)^+dv\leq R(\ps)
    +2(w(\ps_L)-\ps_L)+\mathcal{O}(\eps^{\beta+ \beta/(\alpha-1)})
    = \Pi^\star(\mathcal{M}(0))+\mathcal{O}(\eps^{\beta+ \beta/(\alpha-1)}) + \eps^{1-\beta},
 \end{equation*}
 where we have used the definition of $\w(\cdot)$ and that $m-1\approx \eps^{\beta}$. We obtain \Cref{thm:up-bd-per} by combining the previous upper bound with \Cref{prop:numberOfSteps}, which gives a bound for the dual problem of order
$ \tilde{\mathcal{O}}(\eps^{\gamma})$ with
 $\gamma = \min\{1-\beta,\beta+ \beta/(\alpha-1)\}$. The latter is minimized by taking $\beta = (\alpha-1)/(2\alpha-1)$ which, in turn, yields $\gamma= \alpha/(2\alpha-1)$.

\section{Performance of deterministic mechanisms}\label{sec:floor}
In the previous section, we developed an upper bound on the revenue gains for the seller. We now begin our analysis of the achievable revenue gains by looking at the class of deterministic mechanisms, i.e., the class
\begin{eqnarray*}
\mathcal{M}_\texttt{d}(\varepsilon) = \left\{m \in \mathcal{M}(\varepsilon) : x(\cdot): \mathcal{S} \rightarrow \{0,1\} \right\}.
\end{eqnarray*}
By analyzing this class of mechanism, we will develop a lower bound for \Cref{eq:key-Metric} which, together with the upper bound, will define the feasible spectrum of achievable revenue performances.
 Deterministic mechanisms are also appealing in practice because the buyer can always expect the same outcome from the same report. It is important to note that when $\eps=0$, an optimal mechanism is deterministic and hence included in $\mathcal{M}_\texttt{d}(0)$. The next result characterizes the best performance of deterministic mechanisms.

\begin{theorem}[Deterministic mechanisms] \label{prop:hard/soft-floor}Suppose \Cref{assumption1} holds.
An optimal deterministic mechanism, $\md = (\xd, \td)$, is given by the hard/soft floor mechanism:
\begin{eqnarray*}
\xd(v) &=& \mathbf{1}\{v \in [\pd,\bv] \},\\
\td(v) &=& v \: \mathbf{1}\{v \in [\pd,s]\} + s \:\mathbf{1}\{v \in [s,\bv]\},
\end{eqnarray*}
where the hard floor $p$ and soft floor $s$  satisfy $s = p + \epsilon$ and are chosen to maximize revenue. Furthermore, the optimal performance satisfies
\begin{equation*}
 \Pi^\star(\mathcal{M}_\texttt{d}(\eps))  \: - \: \Pi^\star(\mathcal{M}(0))\: = \: \Theta(\eps).
\end{equation*}
\end{theorem}
The theorem establishes that over the entire class of $\eps$-IC deterministic mechanisms, the seller cannot  exploit the relaxation of the IC constraints to yield more than linear gains in $\eps$. Furthermore, this bound is shown to be  tight for deterministic mechanisms as one can exhibit a mechanism that yields such gains in performance. That is, there exist mechanisms that yield a linear improvement when relaxing the IC constraints.  This sets a lower bound on the performance that can be achieved in the approximate-IC setting.  Any optimal approximate-IC mechanism should deliver at least linear gains compared to the optimal performance with IC constraints. Furthermore, we know that it is impossible to break the linear barrier with deterministic mechanisms. This means that the spectrum of achievable performances ranges in general from $\eps$ to $\eps^{1/2}$, but from \Cref{thm:up-bd-per} this spectrum can be tightened: it becomes  $\eps$ to $\eps^{\alpha/(2\alpha-1)}$ when we account for the characteristics of a distribution. The next natural questions that emerge are whether this linear guarantee can be surpassed, and if so, how much more performance gain can be achieved and what mechanisms guarantee such improved performance, i.e., is the $\eps^{\alpha/(2\alpha-1)}$ performance achievable? In the next section, we will expand to broader families of \textit{randomized} mechanisms and provide a positive answer to these questions.

\paragraph{Discussion and proof ideas.}
The proof is organized around two main steps. We first establish that an optimal deterministic mechanism must take the form of a hard/soft floor mechanism, i.e., a mechanism characterized by a pair $(\pd,\pt)$ with  $\pd \le \pt \le \bv$ such that the allocation and payments are as in the statement of the theorem.
In this family of mechanisms, if the buyer's bid is between the hard and soft floor then he pays his own bid and is allocated the item. If the buyer's bid is above the soft floor then he pays the soft floor and is allocated the item. Otherwise the item is not allocated and the buyer does not incur any payment. Incidentally, this class of mechanisms that naturally emerges in the $\eps$-IC setting has been used by some online advertising platforms in the past to extract revenue from buyers with high values~\citep{mopub2013, Zeithammer2019soft}. %

In the second step of the proof, we show that there exists $r\in[0,\bv]$ such that the optimal hard/soft floor mechanism is  characterized by
$\pd=r$ and $\pt=r+\eps$. Note that this is an alternative way to show that when $\eps=0$ the optimal deterministic mechanism must be the optimal posted price $\ps$. Importantly, from this characterization one can readily obtain the upper bound in the theorem. Indeed, using integration by parts, we obtain that the revenue of $\md$ with hard/soft floors as above satisfies
\begin{equation*}
\Pi(\md)=r\cdot  \overline{F}(r) +\int_{r}^{r+\eps}\overline{F}(v)dv\: \leq\: R(r) +\eps \: \leq\: \Pi^\star(\mathcal{M}(0))+\eps,
\end{equation*}
where the first inequality follows because $\overline{F}(v) \le 1$ and the second by optimizing over $r$.
The lower bound is derived in a similar fashion except that we replace $r$ with $\ps$ and bound the integral above from below by $\overline{F}(\ps+\eps)$, which is positive for small enough $\epsilon >0$ because the optimal posted price is interior.

Interestingly, for hard/soft floor mechanisms, truthful reporting is not a best reporting strategy. A best reporting strategy is given by
\begin{equation}\label{eq:v_star}
\vd(v) = \twopartdef{0}{v\in[0,r);}{r}{v\in[r,\bv].}
\end{equation}
See \Cref{fig:reporting_fun} (b) for an illustration.  For $v\in [r,\bv]$ the best reporting strategy for the buyer is always to report a lower type, namely, $\vd(v)=r$ and gain a utility of $v-r$. Nevertheless, since the buyer is $\eps$-optimizer he will be willing to bid $v$ and take a loss of at most $\eps$. The mechanism exploits this fact to obtain an extra revenue  of $\int_{r}^{r+\eps}\overline{F}(v)dv=\mathcal{O}(\eps)$ compared to the optimal posted-price mechanism.

The main departure point of \eqref{eq:probeps} from  $(\mathcal{P}_0)$ is that in the former the reporting function is not necessarily determined by the incentive constraints and, therefore, one can perform the optimization in \eqref{eq:probeps} by searching over the space of reporting functions. In the space of deterministic mechanisms, if we fix the allocation to $\xd(\cdot)$ and choose as a reporting function $\vd(\cdot)$ as described above, then from \eqref{eq:IC} and \eqref{eq:IR} it is possible to infer that the transfer $\td(\cdot)$ must be zero below $r$ and bounded above by $\min\{\td(r)+\eps,v\}$ with $\td(r)\le r$. Therefore, if we optimize over the transfer, we will end up with $\td(\cdot)$ as in the statement of  \Cref{prop:hard/soft-floor}. In other words, in  mechanism design with approximate incentive compatibility we  have two degrees of freedom when designing mechanisms: allocations and reporting functions. In the following section, we will exploit this observation to analyze  the richer class of randomized mechanisms.

\section{Near Optimal Mechanisms}\label{sec:low-bound-perf}
In Section \ref{sec:upperboundper}, we saw that the relaxation to $\varepsilon$-IC mechanisms can lead
to at most  $\tilde{\mathcal{O}}(\eps^{\alpha/(2\alpha-1)})$ revenue gains for $\alpha\in (1,\infty)$ compared to
$(\mathcal{P}_0)$.  In this section, we establish that our upper bound is tight by developing a lower bound for the growth rate of
$\Pi^\star(\mathcal{M}(\eps)) \: -\: \Pi^\star(\mathcal{M}(0))$
 as a function of $\eps$. We also present some of the main proof ideas as they highlight new techniques and arguments that are critical  for the understanding of $\eps$-IC mechanism design problems.
  We present the main result of this section next.
\begin{theorem}[Performance Lower Bound]\label{thm:low-bd-per}Suppose \Cref{assumption1} holds, and that the  revenue function admits local $\alpha$-power envelopes in a neighborhood $\mathcal{N}_\ell$ of $\ps$.  Then
\begin{equation*}
\Pi^\star(\mathcal{M}(\eps)) \: -\: \Pi^\star(\mathcal{M}(0)) \:=\:\Omega(\eps^{\alpha/(2\alpha-1)}).
\end{equation*}
\end{theorem}
This result  has multiple implications. First, it shows that the seller can benefit by at least $\Omega(\eps^{\alpha/(2\alpha-1)} )$ by offering a mechanism that is $\eps$-IC as opposed to IC.   That is, when we search in the entire class of $\varepsilon$-IC mechanisms, $\mathcal{M}(\varepsilon)$, quite notably, it is indeed possible to achieve \textit{supralinear} gains  compared to when we restrict attention to $\mathcal{M}(0)$. In conjunction with Theorem \ref{thm:low-bd-per}, we obtain that there exist constants $\underline{C}, \bar{C} > 0$ such that for all $\eps$ small enough
\begin{equation*}
\underline{C} \eps^{\alpha/(2\alpha-1)} \: \le \: \Pi^\star(\mathcal{M}(\eps)) \: -\: \Pi^\star(\mathcal{M}(0)) \:\le \:\bar{C} \eps^{\alpha/(2\alpha-1)}\log(1/\eps),\quad \alpha\in (1,\infty).
\end{equation*}
In other words, we have established that it is possible to achieve supralinear growth and characterized the types of supralinear growths that are possible ranging from linear to square root performance.
 Interestingly, in contrast to classical linear programming, a linear perturbation to the infinite-dimensional linear program \eqref{eq:probeps} leads to a \textit{nonlinear} gain in the objective.
 In conjunction with \Cref{prop:hard/soft-floor}, this result also establishes the crucial role that randomization takes in the class of $\varepsilon$-IC problems: \textit{An optimal mechanism has to randomize as soon as $\eps>0$.
 } This implies that the extreme points in \eqref{eq:probeps} and $(\mathcal{P}_0)$ are of a very different nature. In $(\mathcal{P}_0)$ the extreme points are posted-price mechanisms; however, \Cref{thm:low-bd-per} establishes that perturbing the IC constraints in a mechanism design problem fundamentally changes the feasible region and the set of extreme points is richer. In \Cref{fig:alloc-tra}, we depict a near optimal  $\eps$-IC mechanism that exhibits randomization.
 This highlights the subtle and deep impact of relaxing IC constraints on the nature of the optimization problem, and the significant impact it has on the performance that can be achieved (see \citealt{martin2016slater} and \citealt{basu2017strong} for other examples of properties that are true in finite dimension but fail in infinite-dimensional optimization problems).

\begin{figure}[t]
\begin{center}
\scalebox{1}{\begin{tikzpicture}[baseline=0pt,scale=0.5]
\def\n{10}

\draw [->,black,line width=0.3mm] (0,0) -- (1.1*\n,0);
\draw [->,black,line width=0.3mm] (0,0) -- (0,0.7*\n);
\draw node at (1.1*\n+0.3,-0.8) { $v$};

\draw [-,line width=0.6mm] (0,0.3) -- (0,-0.3) node at (0,-0.8) {0};

\draw [-,line width=0.6mm]  node at (-1.1,7.1) { $x(v)$};

\draw [-,line width=0.6mm] (0.3,5.1) -- (-0.3,5.1) node at (-1.0,5.1) { \small $1$};

\draw[line width=0.4mm,shift={(0,0.1)},blue] plot file {alloc-L-bound.data};

\def\sh{14}
\draw [->,black,line width=0.3mm] (0+\sh,0) -- (1.1*\n+\sh,0);
\draw [->,black,line width=0.3mm] (0+\sh,0) -- (0+\sh,0.7*\n);
\draw node at (1.1*\n+0.3+\sh,-0.8) { $v$};

\draw [-,line width=0.6mm] (0+\sh,0.3) -- (0+\sh,-0.3) node at (0+\sh,-0.8) {0};

\draw [-,line width=0.6mm]  node at (-1.1+\sh,7.1) { $t(v)$};

\draw [-,line width=0.6mm] (0.3+\sh,5.1) -- (-0.3+\sh,5.1) node at (-1.0+\sh,5.1) {\small $0.5$};

\draw[line width=0.4mm,shift={(0+\sh,0.1)},blue] plot file {trans-L-bound.data};

\draw [dashed,line width=0.4mm,red] (0,0) -- (4.95,0);
\draw [dashed,line width=0.4mm,red] (4.95,0) -- (4.95,5.09);
\draw [dashed,line width=0.4mm,red] (4.95,5.09) -- (10,5.09);

\draw [dashed,line width=0.4mm,red] (0+\sh,0) -- (4.95+\sh,0);
\draw [dashed,line width=0.4mm,red] (4.95+\sh,0) -- (4.95+\sh,4.7);
\draw [dashed,line width=0.4mm,red] (5.4+\sh,5.05) -- (10+\sh,5.05);
\draw [dashed,line width=0.4mm,red] (4.95+\sh,4.7) -- (5.4+\sh,5.05);

\draw[] plot coordinates {(0.99+\sh,0)(1.98+\sh,0)(2.97+\sh,0)(3.96+\sh,0)(4.95+\sh,0)};

\draw [-,line width=0.6mm] (5.2,0.3) -- (5.2,-0.3) node at (5.2,-0.8) {\small $\ps$};
\draw [-,line width=0.6mm] (5.2+\sh,0.3) -- (5.2+\sh,-0.3) node at (5.2+\sh,-0.8) {\small $\ps$};

\draw [-,line width=0.4mm,blue] (10.4+\sh,4.1) -- (10.8+\sh,4.1) node at (14.0+\sh,4) {\small near-optimal ($\mr$)};
\draw [dashed,line width=0.4mm,red] (10.3+\sh,3.1) -- (10.9+\sh,3.1) node at (14.0+\sh,3) {\small deterministic ($\md$)};

\end{tikzpicture}}
\end{center}
\caption{Allocations and transfers of the lower bound mechanism and the optimal deterministic mechanism when $F$ is uniform in $[0,1]$, i.e., $\alpha=2$, and $\eps=0.001$. Because $\eps$ is small, the optimal posted-price mechanism when $\eps=0$ is almost indistinguishable from $\md$.
}
\label{fig:alloc-tra}
\end{figure}
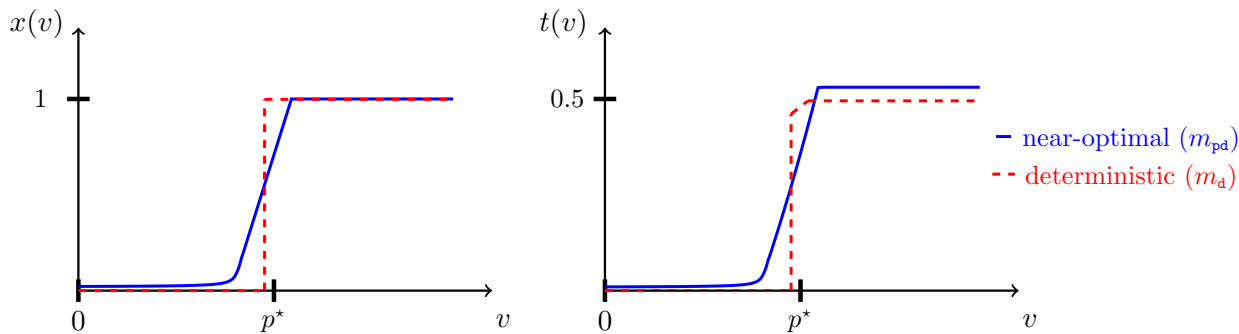

The Proof of \Cref{thm:low-bd-per} is constructive in that we judiciously construct a feasible mechanism that yields said performance. The proof  also sheds light on what are the key forces at play when optimizing over $\mathcal{M}(\varepsilon)$, and showcases important departures from the classical mechanism design approach. We next detail  the proof and its main building blocks.

\subsection{Proof of \Cref{thm:low-bd-per}}

The proof has three main components.  First we propose a family of mechanisms through an endogenous characterization.
This characterization critically depends on a carefully designed parametric family of best reporting mappings which indirectly induce the mechanisms.
In particular, we establish that for our proposed family, the allocation is fully characterized by a system of ordinary and delayed differential equations. We then solve this system and  establish that the induced mechanisms are feasible for \eqref{eq:probeps}. Finally, we characterize the performance of these mechanisms and,  by optimizing over their defining parameters,  we obtain the $\Omega(\eps^{\alpha/(2\alpha-1)} )$ revenue gains.

\subsubsection{Step 1: A Proposed Family of Mechanisms}\label{sec:prop-family}

Our main challenge is to exhibit  a feasible mechanism that is simple enough so that we can analyze its performance, but rich enough so that it can garner additional $\eps^{\alpha/(2\alpha-1)}$ revenue gains. We accomplish this in two steps. First, we propose a suitable family of best reporting functions $\vsr(\cdot|\mu,\delta)$ parametrized by $(\mu,\delta)$ (see \cref{eq-vs-fam}). Second, we establish a consistency property of $\vsr(\cdot|\mu,\delta)$: it
coincides with the best reporting function of a  family of mechanisms that we coin \textit{perturbed delayed mechanisms} (see \Cref{prop:alloc-tra-vs}).

\paragraph{Best reporting function.}
In \Cref{sec:formulation}, we saw that standard reductions no longer apply  when $\eps>0$.
In particular, after selecting a mechanism $m\in \mathcal{M}(\eps)$ the seller must also account for  the best reporting mappings $\vsr(\cdot)$ (c.f., \cref{eq:best-str-Map}). Consequently, the optimization must be indirectly performed over the  triplet $(\all, \tra,\vsr)$. One possible approach to come up with a feasible solution for \eqref{eq:probeps} is to define a pair $(\all,\tra)$ that is both IR and $\eps$-IC. However, such an approach is challenging given the lack of control on $\vsr(\cdot)$.
To bypass this issue, the approach we take is to define $\vsr(\cdot)$ as a natural  perturbation of $\vd(\cdot)$, defined in \Cref{eq:v_star} and depicted in \Cref{fig:reporting_fun} (b). Then we use the fact that $\vsr(v)$ must  maximize $v\cdot \all(w)-\tra(w)$ for some $m=(\all,\tra)\in \mathcal{M}(\eps)$ to reverse engineer  the mechanism $m$ that is consistent with $\vsr(\cdot)$. The perturbation of the reporting function will lead to randomization  in the mechanisms that induce it. %
 We define the best reporting mappings $\vsr(\cdot|\mu,\delta)$---a perturbation of $\vd(\cdot)$---parametrized by some positive constants $(\mu,\delta)$  as follows

\begin{equation}\label{eq-vs-fam}
\vsr(v|\mu,\delta)\triangleq \threepartdef{0}{v\in[0,\ps-\delta);}{v-\mu}{v\in[\ps-\delta,\ps+\delta+\mu];}{\ps+\delta}{v\in [\ps+\delta+\mu,\bv].}
\end{equation}
We assume that $\delta < \ps$ and $\ps+\delta+\mu < \bv$, which is possible because the optimal posted price $\ps$ is assumed to be interior.

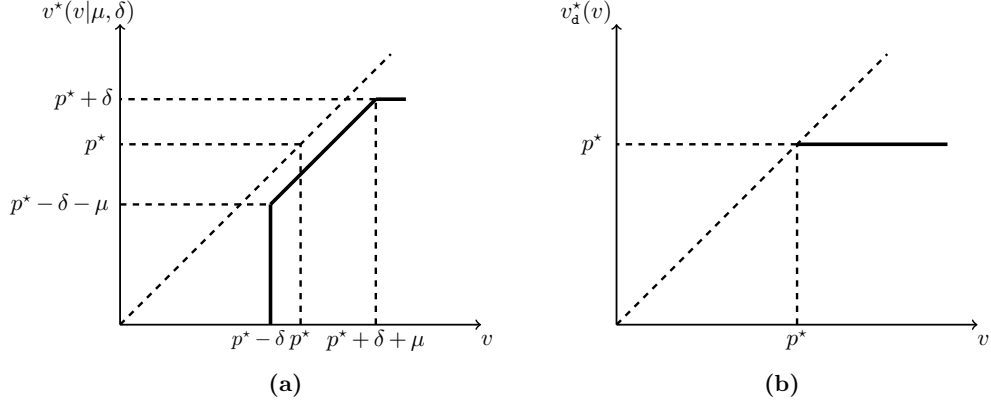
\begin{figure}[t]
\begin{center}
\scalebox{0.8}{\begin{tikzpicture}[baseline=0pt]

\draw[->,line width=0.3mm] (0,0)--(0,5);
\draw[->,line width=0.3mm] (0,0)--(6,0);
\node at (-0.5,5.2){$\vd(v)$};
\node at (6.1,-0.25){$v$};

\draw[dashed,line width=0.4mm](0,0)--(4.5,4.5);

\node at (-0.4,3 ){$\ps$};
\node at (3,-0.25 ){$\ps$};
\draw[dashed,line width=0.4mm](3,0)--(3,3);
\draw[dashed,line width=0.4mm](3,3)--(0,3);
\draw[line width=0.6mm](3,3)--(5.5,3);

\def\sh{-8.25}

\draw[->,line width=0.3mm] (0+\sh,0)--(0+\sh,5);
\draw[->,line width=0.3mm] (0+\sh,0)--(6+\sh,0);
\node at (-0.52+\sh,5.2){$\vsr(v|\mu,\delta)$};
\node at (6.1+\sh,-0.25){$v$};

\draw[dashed,line width=0.4mm](0+\sh,0)--(4.5+\sh,4.5);

\draw[line width=0.6mm](2.5+\sh,0)--(2.5+\sh,2) node at (2.3+\sh,-0.25){\small $\ps-\delta$};
\draw[dashed,line width=0.4mm] (0+\sh,2)--(2.5+\sh,2)node at (-1.0+\sh,2) {$\ps-\delta-\mu$};

\draw[line width=0.6mm](2.5+\sh,2)--(4.25+\sh,3.75);
\draw[dashed,line width=0.4mm] (4.25+\sh,0)--(4.25+\sh,3.75) node at (4.25+\sh,-0.25) {$\ps+\delta+\mu$};
\draw[dashed,line width=0.4mm] (4.25+\sh,3.75)--(0+\sh,3.75) node at (-0.6+\sh,3.75) {$\ps+\delta$};

\draw[line width=0.6mm](4.25+\sh,3.75)--(4.75+\sh,3.75);

\node at (-0.4+\sh,3 ){$\ps$};
\node at (3+\sh,-0.25 ){$\ps$};
\draw[dashed,line width=0.4mm](3+\sh,0)--(3+\sh,3);
\draw[dashed,line width=0.4mm](3+\sh,3)--(0+\sh,3);

\node at (2.75,-1){\textbf{(b)}};
\node at (2.75+\sh,-1){\textbf{(a)}};

\end{tikzpicture}}
\end{center}
\caption{Illustration of  \textbf{(a)} the parametrized best reporting function and \textbf{(b)} the best  reporting function for a hard/soft floor mechanism with $\pd=\ps$ and $\pt=\ps+\eps$.}
\label{fig:reporting_fun}
\end{figure}
For a graphical illustration of $\vsr(v|\mu,\delta)$ see \Cref{fig:reporting_fun} (a).
In what follows, when clear from context, we drop the dependence on $(\mu,\delta)$ and use $\vsr(v)$ to denote $\vsr(v|\mu,\delta)$.
The definition of  $\vsr(\cdot)$ implies that under the associated mechanism the best reporting function always misreports the value in a downward fashion, $\vsr(v)\leq v$. It equals 0  for  small enough valuations (below $\ps-\delta$),  $v-\mu$ for intermediate valuations (in $[\ps-\delta,\ps+\delta+\mu]$) and $\ps+\delta$ for large valuations (above $\ps+\delta$).

Note that $\vsr(\cdot)$ generalizes the best reporting functions of the hard/soft floor mechanism $(\ps,\ps+\eps)$. Indeed,
for $\delta=0$ both functions coincide, except in  $[\ps,\ps+\mu]$ where $\vsr(\cdot)$ is linear. However,
as $\mu$ becomes smaller the functions become indistinguishable, i.e., $\lim_{\mu\downarrow0}\vsr(v|\mu,0) =\vd(v)$.

\paragraph{Perturbed delayed mechanisms.}
We first establish that it is possible to characterize the allocation and payments associated to $\vsr(\cdot)$ through a set of ordinary and delayed differential equations. Before stating the result, we let $w(\cdot)$ denote the generalized inverse of $\vsr(\cdot)$, which is given by
\begin{equation*}
w(v)\triangleq \twopartdef{\ps-\delta}{v\in[0,\ps-\delta-\mu];}{v+\mu}{v\in[\ps-\delta-\mu, \ps+\delta].}
\end{equation*}

\begin{proposition}[Perturbed delayed mechanisms]\label{prop:alloc-tra-vs}
Let $\xr:[0,\ps + \delta]\rightarrow \mathbb R$ be a solution of the system
\begin{flalign}\label{eq:ode-1}\tag{ODE}
\dxr(v)&=\frac{\xr(v)}{\w(v)-v}, \:\: \forall v\in [0,\ps-\delta], \qquad \xr(0)=\frac{\eps}{\ps-\delta}, \\\label{eq:dde-1}\tag{DDE}
\dxr(v)&=\frac{\xr(v)-\xr(\vsr(v))}{w(v)-v}, \:\: \forall v\in [\ps-\delta,\ps+\delta].
\end{flalign}
Moreover, let $\xr(v)= 1$ for all $v\in (\ps+\delta,\bv]$ and let $\tr:[0,\bv]\rightarrow \RR$
 be given by
\begin{equation*}
\tr(v) = \twopartdef{v\cdot \xr(v)}{v< \ps-\delta;}{v\cdot \xr(v)- \int_{0}^{v}\xr(\vsr(s))ds +\eps}{v\geq \ps-\delta.}
\end{equation*}
If $\lim_{v\uparrow\ps+\delta}\xr(v)= 1$, then $\vsr(\cdot)$, defined in \eqref{eq-vs-fam}, is a best reporting function for $\mr=(\xr,\tr)$.
\end{proposition}
This result provides a partial characterization of the family of mechanisms that lead to $\vsr(\cdot)$. The allocation, and as a consequence the transfers, are determined by a set of ordinary \eqref{eq:ode-1} and  \textit{delayed differential equations} \eqref{eq:dde-1}. That is, \Cref{prop:alloc-tra-vs} equips us with a precise methodology to generate mechanisms. (In the next section we will show that there is a range of values  $(\mu,\delta)$ for which $\mr$ is indeed an $\eps$-IC mechanism.)
Given that $\xr$ must be found by solving a delayed differential equation and that it depends on $\vsr(\cdot)$ (a perturbation of $\vd(\cdot)$), we name the family of mechanisms in \Cref{prop:alloc-tra-vs} the \textit{perturbed delayed mechanisms}.

Another important implication of
\Cref{prop:alloc-tra-vs} is that it showcases the relationship between the allocations and transfers as a function of the best reporting function they induce via \eqref{eq:ode-1}-\eqref{eq:dde-1}. This highlights that the class of $\eps$-IC mechanisms is of a different nature compared to the class of IC mechanism.

To obtain \Cref{prop:alloc-tra-vs} we establish necessary conditions for the allocations and transfers that induce the postulated best reporting mapping. Under differentiability conditions, this leads to  the set  of ordinary and delayed differential equations for the allocation. Next we provide a heuristic derivation of the differential equations for $\xr$ and transfers. A formal proof of the proposition can be found in the appendix.

Any mechanism $m=(\all,\tra)$ that is compatible with $\vsr(\cdot)$ must satisfy condition \eqref{eq:best-str-Map}, i.e., the best response function should maximize the buyer's utility.
If we further assume that the allocations and transfers are differentiable, and impose that the first order conditions
are satisfied at every $\vsr(v)$, we deduce that  $v\cdot \dot{\all}(\vsr(v))=\dot{\tra}(\vsr(v))$ for all $ v\in [0,\bv]$. Recalling that $w(\cdot)$ is the generalized inverse of $\vsr(\cdot)$, we then have
\begin{equation}\label{eq:fo-cond}
w(v)\cdot \dot{\all}(v)=\dot{\tra}(v), \quad \forall v\in [0,\ps+\delta].
\end{equation}
Note that because $\vsr(v)=0$ for $v\leq \ps-\delta$ this equation is valid only for $v\in [\ps-\delta+\mu,\ps+\delta]$. By extending the equation for
$v\le \ps -\delta+\mu$, we are simply imposing more structure on $(x,t)$.
The equation above characterizes the transfers (up to a constant) in $[0,\ps+\delta]$ as a function of $w(\cdot)$ (or, equivalently, as a function of $\vsr(\cdot)$) and $\all(\cdot)$.
We would like, however, to obtain a characterization of both transfers and allocations as a function of solely $\vsr(\cdot)$. Using the structure of the optimal deterministic mechanism and  the upper bound derived in Section~\ref{sec:path-based-upper-bound} as motivation, we impose two additional conditions that will enable us to accomplish this goal. First, we assume that
\eqref{eq:IR} binds for all low types  $v\in [0,\ps-\delta]$, that is,  $v\cdot \all(v)-\tra(v)=0$. By differentiating,  we obtain
$\all(v)+v\cdot \dot{\all}(v)=\dot{\tra}(v)$ for all $ v\in [0,\ps-\delta]$. Combining this with Eq. \eqref{eq:fo-cond} yields
\begin{equation}\label{eq:ode-w}
\all(v)+(v-w(v))\cdot \dot{\all}(v)=0, \quad \forall v\in [0,\ps-\delta].
\end{equation}
Second, we assume that \eqref{eq:IC} binds for high types $v\in [\ps-\delta,\bv]$. When \eqref{eq:IC} holds, condition~\eqref{eq:env-UB-t} holds with equality and we obtain
\begin{equation*}
v\cdot \all(v)-\tra(v)= u(0)+\int_{0}^{v}\all(\vsr(s))ds -\eps,\quad \forall v\in [\ps-\delta,\bv].
\end{equation*}
Differentiating both sides above obtains $\all(v) + v\dot{\all}(v)-\dot{\tra}(v)=\all(\vsr(v))$ for all $v\in [\ps-\delta,\bv]$, which in conjuction with \cref{eq:fo-cond} implies
\begin{equation}\label{eq:dde-w}
\all(v) + (v-w(v))\dot{\all}(v)=\all(\vsr(v)),\quad \forall v\in [\ps-\delta,\ps+\delta].
\end{equation}
In summary, \cref{eq:ode-w} and \cref{eq:dde-w} together provide a characterization (up to a constant) for the allocation as a function
of our choice of best reporting function $\vsr(\cdot)$. Note that the solution to these equations only yields the allocation in $[0,\ps+\delta]$. For higher values we always allocate. Since we are assuming that \eqref{eq:IR} and \eqref{eq:IC} bind in $[0,\ps-\delta]$ and $[\ps-\delta,\bv]$, respectively, we also obtain the transfers for all values.
Finally, observe that combining   \eqref{eq:IR} and \eqref{eq:IC}  at $\ps-\delta$
implies that $x(0)=\eps/(\ps-\delta)$. This leads to the systems of ordinary and delayed differential equations in \Cref{prop:alloc-tra-vs}.

\subsubsection{Step 2: Feasibility of proposed mechanisms}
In this section, we establish that the perturbed delayed mechanism from \Cref{prop:alloc-tra-vs} is feasible, that is,
$\mr\in \mathcal{M}(\eps)$. To do so, we take two steps. We first argue that the system \eqref{eq:ode-1}-\eqref{eq:dde-1} admits a unique solution and that  the solution is monotone. Then, we establish that for a range of parameters $(\mu,\delta)$ the  mechanism is feasible.

\begin{proposition}[Well-posedness]\label{prop:well-posed}
The perturbed delayed mechanism $\mr$ from \Cref{prop:alloc-tra-vs} is well defined, that is, the system \eqref{eq:ode-1}-\eqref{eq:dde-1} has a unique solution.
In addition, $\xr$ is strictly positive,  monotone non-decreasing in $[0,\ps + \delta]$, and admits a closed-form representation.
\end{proposition}
We establish the proposition by analyzing the system of differential equations. Note that
\eqref{eq:ode-1} can be solved by simple integration. However, \eqref{eq:dde-1} is a delayed differential equation
for which standard methods of ordinary differential equations do not apply.  Performing a change of variables and using the method of steps (see, e.g., \citealt{driver2012ordinary}), we can derive a closed-form solution for the allocation.  We provide the closed-form solution in the proof.

\Cref{fig:alloc-tra} provides an example when $F$ is a uniform
distribution. There are several interesting observations. First, note that $\xr(\cdot)$ is a non-decreasing and randomized allocation. Second,
both allocations and transfers represent a perturbation of the optimal hard/soft floor mechanism $\mr=(0.5\cdot(1-\eps),0.5\cdot(1+\eps))$.
Indeed, one can establish that as $\eps\downarrow0$, $\mr$ becomes $\md$.
Now that we have proved that $\mr$ can be explicitly determined, we can leverage \Cref{prop:alloc-tra-vs} to show that
$\mr$ is an $\eps$-IC mechanism.

\begin{proposition}[Feasibility]\label{prop:feas}
Fix $\eps>0$ and $\mu>\eps\cdot e$ then there exists $\delta(\mu,\eps)>0$  in $\left[ \frac{\mu^2}{4\eps}-\mu, \frac{ \mu^2}{4\eps}\right]$
such that
$\lim_{v\uparrow\ps+\delta}\xr(v)= 1$. For $\mu$ small such that $(\ps-\delta,\ps+\delta+\mu)\subset[0,\bv]$ the perturbed delayed mechanism is feasible, i.e., $\mr\in \mathcal{M}(\eps)$.
\end{proposition}
Note that \Cref{prop:alloc-tra-vs} provides a sufficient condition for $\vsr(\cdot)$ to be a best reporting function for $\mr$.
Under this condition one can readily verify that $\mr\in\mathcal{M}(\eps)$.  To see why \eqref{eq:IC} holds,
note that because $\vsr(\cdot)$ is a best reporting function we have that $\max_{w}\{v\cdot \xr(w)-\tr(w)\}
=v\cdot \xr(\vsr(v))-\tr(\vsr(v))$. This, together with the
envelope theorem (see \eqref{eq:env-report}), implies that \eqref{eq:IC} can be cast as
\begin{equation*}
v\cdot \xr(v)-\tr(v)\geq \int_{0}^v \xr(\vsr(s))ds-\epsilon,\quad \forall v\in [0,\bv].
\end{equation*}
If $v\leq \ps-\delta$, then the left hand-side above is 0; while the right hand-side is $v\xr(0)-\eps$ which, by the boundary condition of $\xr(\cdot)$, is non-negative. If $v\geq \ps-\delta$ then the inequality above is binding.
That is, $\mr$ verifies \eqref{eq:IC}. The \eqref{eq:IR} constraint can be  verified in a similar fashion. We note that the condition $(\ps-\delta,\ps+\delta+\mu)\subset[0,\bv]$ is needed to ensure that $\vsr(\cdot)$ is properly defined in $[0,\bv]$ and, from the characterization of $\delta(\mu,\eps)$, it is satisfied when $\mu$ is small. In \Cref{sec:per-charac}, we show how to properly choose $\mu$.

Therefore, the crucial step to establish \Cref{prop:feas} is to show that the sufficient condition in \Cref{prop:alloc-tra-vs} can be met. To do so, we show it is always possible to find $(\mu,\delta)$ such that $\xr((\ps+\delta)^-)=1$. Using that $\xr(\cdot)$ is non-decreasing and that it belongs to $[0,1]$, yields the desired conclusion. We note that establishing that $\xr((\ps+\delta)^-)=1$ requires to carefully choose $(\mu,\delta)$ so the solution to \eqref{eq:ode-1}-\eqref{eq:dde-1} hits 1 at $\ps+\delta$. The proof involves an elegant application of renewal theory---we provide more details in the appendix.

\subsubsection{Step 3: Performance characterization}\label{sec:per-charac}

From \Cref{prop:feas}, we have derived a feasible mechanism. Next, we analyze its performance. The following result derives an explicit form for the performance of a perturbed delayed mechanism.
\begin{proposition}[Performance Analysis]\label{prop:low-bd-obj}  Fix $\mu,\delta>0$ and $\eps>0$ such that the mechanism $\mr$ is feasible. Then the revenues garnered under $\mr$ are given by
\begin{flalign*}
\Pi(\mr)=\int_{0}^{\ps+\delta}\xr(v)\psi(v|w(\cdot))f(v)dv
+R(\vm)+\int_{\vm}^{\vm+\mu}\overline{F}(v)dv
\end{flalign*}
where
\begin{equation*}
\psi(v|w(\cdot))= v-\frac{\overline{F}(v+\mu)}{f(v)}\cdot \dot{w}(v),\quad \forall v\in [0,\ps+\delta].
\end{equation*}
\end{proposition}
The result above provides an expression for the revenues that highlights the different forces at play. A particular term is notable in the objective above, the first integral. This integral resembles the classical objective reformulation  for the (IC$_0$) problem (see \Cref{sec:formulation}, \cref{eq:probstan}).
Na\"ively, one may be tempted to try to optimize point-wise (as in a classical argument), but such an approach fails here as $\xr(v)$ depends on the entire profile of the best reporting strategy $\vsr(\cdot)$ and its inverse. The latter, $\w(\cdot)$, comes into play in the
term  $\psi(v|w(\cdot))$ that, due to its resemblance to the standard virtual value,  we coin ``\textit{delayed virtual value}.''
This further highlights the difference between the class of $\eps$-IC problems and the standard IC ones. Even when values are regular, the presence of $\eps$-IC demands a global optimization approach, while in the IC setting local optimization is enough to characterize an optimal mechanism.

Next, we analyze and lower bound each term in $\Pi(\mr)$.
Consider the second and third terms in  $\Pi(\mr)$.
Because distributions are monotone, we have
\begin{flalign*}
R(\vm)+\int_{\vm}^{\vm+\mu}\overline{F}(v)dv&\ge R(\ps+\delta)+
\overline{F}(\vm+\mu)\cdot \mu \\
&\geq \rs-\kappa_U\delta^\alpha + \bar{F}(\ps+\delta+\mu)\mu := (A),
\end{flalign*}
 where the last inequality follows from a Taylor expansion of the revenue function and from \Cref{def1}, which implies that the derivative of $R(\cdot)$ is bounded below by
 $-\kappa_U \alpha(v-\ps)^{\alpha-1}$ for $v\geq \ps$.
 Using (A) and the bounds for $\delta$ in \Cref{prop:feas} we choose $\mu$ that enables us to achieve the right scaling.
If we choose $\mu=\mu(\eps)=K\cdot \eps^{\alpha/(2\alpha-1)}$ for some positive constant $K>0$, we have that
$ \delta(\mu(\eps),\eps)\leq \mu(\eps)^2/(4\eps) \leq  K^2 \eps^{1/(2\alpha-1)/4}$.
 This delivers,
\begin{equation*}
(A)\geq \Pi^\star(\mathcal{M}(0)) +\left(\bar{F}(\ps+\delta+\mu)\cdot K -\kappa_U \left(\frac{K^2}{4}\right)^\alpha \right)\cdot \eps^{\alpha/(2\alpha-1)},
\end{equation*}
 In the equation above $K$ needs to be taken small enough so that $
\bar{F}(\ps+\delta+\mu) >\kappa_U\left(\frac{1}{4}\right)^\alpha K^{2\alpha-1}
>0$ (which is possible because $\alpha>1/2$, $\delta(\eps)+\mu(\eps)\downarrow 0$ and $\overline{F}(\ps)>0$ ).

This argument shows that the last two terms in \Cref{prop:low-bd-obj} already deliver the desired $\eps^{\alpha/(2\alpha-1)}$ rate.
We highlight that this choice of $\mu(\eps)$ as in \Cref{prop:feas}  is optimal. That is, $\Omega(\eps^{\alpha/(2\alpha-1)})$ is the best guarantee that can be achieved in the class of perturbed delayed mechanisms.
We still need to show, however, that the first term in \Cref{prop:low-bd-obj} does not negatively affect the order of magnitude of the performance gains. We provide a complete, and more formal analysis of $\Pi(\mr)$ in the proof of \Cref{thm:low-bd-per}.

We note that for the family of perturbed delayed  mechanisms---which generalizes optimal deterministic ones---allocations not only exhibit randomization, but they also randomize  even the smallest type, $\xr(0)>0$. A sharp departure from the standard IC mechanism design literature.

\paragraph{Discussion.} Interestingly, the second and third terms of $\Pi(\mr)$ which, as we have seen, provide the $\eps^{\alpha/(2\alpha-1)}$ gains in revenue, coincide with the expected revenue of a hard/soft floor mechanism with hard floor $\ps+\delta$ and soft floor $\ps+\delta+\mu$. Because the revenue function $R(\cdot)$ is smooth around its interior optimizer $\ps$, shifting the hard floor by an amount $\delta$ from the optimal price $\ps$ leads to a loss of order $-\delta^\alpha \approx -\eps^{\alpha/(2\alpha-1)}$ (see \Cref{eq:def1-intuitive}). The seller's gains are introduced by offering a region of size $\mu$ in which the hard/soft floor mechanism charges the bid. Note that in the deterministic case, this region can have length at most $\eps$, because otherwise \eqref{eq:IC} would be violated. Interestingly, randomization allows one to offer a region of larger size $\mu \approx \eps^{\alpha/(2\alpha-1)}$. By carefully choosing the parameters, we can guarantee that the benefits of the region in which bids are paid outweigh the loss of picking a suboptimal hard floor. These terms combined lead to the $\eps^{\alpha/(2\alpha-1)}$ revenue gains.

\section{Conclusion}

 In the present paper, we have initiated the study of the \textit{optimization} of mechanisms under $\eps$-IC constraints. We have highlighted the rich structure underlying this class of optimization problems and  characterized the gains in performance that are achievable. These are, to the best of our knowledge, the first results of this nature in the literature.

This paper opens up various avenues for future research. %
 First,  characterizing the structure of an optimal mechanism in this setting is an open question. While the results in the paper establish the need for randomization and construct mechanisms  with supralinear gains,  various arguments developed are fairly general, yield bounds for arbitrary value of $\eps$,  and may prove useful in characterizing an optimal mechanism in the future.  Another important avenue of research pertains to studying settings with an arbitrary number of buyers and analyzing how the gains grow as a function of this number.

Another possible implication of our analysis is that the classical rounding argument discussed in Section~\ref{sec:Nisan} might not be tight for all mechanisms. This argument is used extensively in the literature on black-box reductions for mechanism design as it provides an approach to construct a exactly IC mechanism from an $\eps$-IC mechanisms while losing an optimal revenue of $\eps^{1/2}$. Because the revenue gains of our near-optimal $\eps$-IC mechanism are $\eps^{\alpha/(2\alpha-1)}$, this suggests that the rounding argument might introduce revenue losses that are higher than necessary for some mechanisms and some distributions. We hope that the ideas developed in this paper also open the door to developing new techniques to construct exactly IC mechanisms from $\eps$-IC ones.

While we have focused on a static setting, the present analysis also opens possibilities for the design and optimization of dynamic mechanisms  under (dynamic) $\eps$-IC constraints. In such settings, an IC mechanism may often be intractable, however, under $\eps$-IC constraints, many potentially ``simple'' mechanisms become feasible. The present analysis may suggest local  perturbations to mechanisms to improve performance while maintaining $\eps$-IC. Additionally, the present paper suggests a refined metric to assess near-optimality of mechanisms when relaxing to $\eps$-IC constraints. The performance of an $\eps$-IC mechanism should not be judged just by whether it approaches the performance of an optimal IC mechanism (first order optimality), but also by the order of magnitude of the gap between the performance of the $\eps$-IC mechanism and an optimal IC mechanism (second order optimality).

\clearpage
\newpage
\begin{center}
\LARGE \textbf{Online Appendix: Mechanism Design under Approximate Incentive Compatibility}
\end{center}

\begin{APPENDICES}

\section{Proofs for  Section \ref{sec:upperboundper}}

\begin{proof}{\underline{\bfseries\sffamily Proof of \Cref{prop:smoothing}}}The following argument follows from a classic rounding argument attributed to Nisan.
Let $m=(\all,\tra)$ be an $\eps$-IC mechanism. Define an alternative mechanism $\tilde{m}$ as follows. For $\delta\in (0,1)$, we let
\begin{equation*}
v^\star_\delta(v)\in \argmax_{w\in \mathcal{S}}\Big\{v\cdot \all(w)-(1-\delta)\tra(w)\Big\}.
\end{equation*}
Then we define the mechanism $\tilde{m}\triangleq(\tilde{\all},\tilde{\tra})=(\all(v^\star_\delta(\cdot)),(1-\delta)\tra(v^\star_\delta(\cdot)))$. Therefore,  $\tilde{m}$ is a $0$-IC mechanism and for $v'=v^\star_\delta(v)$
we have that
\begin{flalign*}
v\cdot\all(v)-\tra(v)&\ge v\cdot \all(v')-\tra(v')-\eps,\\
 v\cdot \all(v')-(1-\delta)\tra(v')&\ge v\cdot \all(v)-(1-\delta)\tra(v),
\end{flalign*}
where the first inequality holds because $(x,t)$ is $\eps$-IC, and the second comes from the definition of $v^\star_\delta(v)$.
Combining the inequalities above we obtain
\begin{equation*}
    \tra(v')\ge \tra(v)-\frac{\eps}{\delta}.
\end{equation*}
Now since $\tilde{m}$ is $0$-IC, we have
\begin{equation*}
    \Pi^\star(\mathcal{M}(0)) \ge \Pi(\tilde{m})=
    \E_v[\tilde{\tra}(v)]=(1-\delta)\E_v[\tra(v^\star_\delta(v))]\ge
    (1-\delta)\left(\E_v[\tra(v)]-\frac{\eps}{\delta}\right).
\end{equation*}
The latter can be rearranged to yield the following
\begin{equation*}
    \Pi(m)-\Pi^\star(\mathcal{M}(0))\leq \frac{\delta}{1-\delta}\Pi^\star(\mathcal{M}(0))
    +\frac{\eps}{\delta}\leq 2\delta\Pi^\star(\mathcal{M}(0))+\frac{\eps}{\delta},
\end{equation*}
where in the last inequality we considered $\delta\leq 1/2$. Choosing $\delta=\sqrt{\eps}$ delivers the result.
\end{proof}

\begin{proof}{\underline{\bfseries\sffamily Proof of \Cref{lem:weak-dual-upper}}}

Given the definitions preceding the lemma, the Lagrangian is:

\begin{flalign*}
L
&=\int_{0}^{\bv}t(v)f(v)dv +\int_{0}^{\mu}\Lir(v)(vx(v)-t(v))dv
+\int_{\mu}^{\bv}\Lic(v)\Big(vx(v)-t(v)-vx(\vsr(v))+t(\vsr(v))+\eps\Big)dv\\
&=\eps \int_{\mu}^{\bv}\Lic(v)dv
+\int_{0}^{\bv} t(v)
\Big(f(v)-\Lir(v)\1{v\in[0,\mu]}
-\Lic(v)\1{v\in[\mu,\bv]}+\Lic(\w(v))\dot{\w}(v)\1{v\in(0,\nu_0)}
\Big)dv\\
&+\int_{0}^{\bv}x(v)\Big(
\Lir(v)v\1{v\in[0,\mu]}
+\Lic(v)v\1{v\in[\mu,\bv]}-\Lic(\w(v))\w(v)\dot{\w}(v)\1{v\in(0,\nu_0)}
\Big)dv,
\end{flalign*}
where we have used the change of variables $u=\vsr(v)$, the fact that $\w$ is the inverse of $\vsr$, that $\vsr$ is continuous and increasing (hence almost everywhere differentiable).
Since $t(v)$ is a free variable, we must have that
\begin{equation*}
f(v)-\Lir(v)\1{v\in[0,\mu]}
-\Lic(v)\1{v\in[\mu,\bv]}+\Lic(w(v))\dot{w}(v)\1{v\in(0,\nu_0)}=0
,\quad \text{a.e on } [0,\bv].
\end{equation*}
From this, we can deduce the following
\begin{flalign*}
\Lic(v) &= f(v),\quad \text{a.e on } [\nu_0,\bv] \\
\Lic(v) &= f(v) + \dot{\w}(v)\Lic(\w(v)),\quad \text{a.e on } [\mu,\nu_0] \\
\Lir(v) &= f(v) + \dot{\w}(v)\Lic(\w(v)),\quad \text{a.e on } [0,\mu]
\end{flalign*}
Note that since $\w$ is increasing then it is differentiable with positive derivative almost everywhere. As a consequence, $\Lic$ and $\Lir$ are non-negative and are well defined almost everywhere in their domains.

Going back to the Lagrangian, we find that
 \begin{equation*}
 L = \varepsilon \int_{\mu}^{\bv}\Lic(v)dv +
 \int_{0}^{\bv} \left(vf(v)-\dot{w}(v)(w(v)-v)\Lic(w(v))\1{v\in[0,\nu_0]}\right)^+dv.
 \end{equation*}
Since we have relaxed the constraints outside the path $(v, \vsr(v))$, only considered the IR constraints for $v\leq \mu$, and  the dual variables are non-negative, the above provides an upper bound for the primal problem. This concludes the proof.
\end{proof}

\begin{proof}{\underline{\bfseries\sffamily Proof of \Cref{lem:charac_lambda_ic}}}
It is easy to see that this is  true for $k=0$. Consider $v\in [\nu_{k},\nu_{k-1}]$ for some arbitrary $k\geq 1$, in this case,  $w(v) \in[\nu_{k-1},\nu_{k-2}]$. Hence, by induction, we have (a.e)
\begin{flalign*}
\lambda(v) &=  f(v) + \lambda(w(v))\dot{w}(v)\\
&= f(v) +\frac{d}{dx}\sum_{j=0}^{k-1} F(w^j(x))\Big|_{x=w(v)}\\
&= f(v) +\sum_{j=0}^{k-1} f(w^j(w(v)))\dot{(w^j)}(w(v))\dot{w}(v)\\
&= f(v) +\frac{d}{dv}\sum_{j=0}^{k-1} F(w^{j+1}(v))\\
&= \frac{d}{dv}\sum_{j=0}^{k} F(w^{j}(v)).
\end{flalign*}
\end{proof}

\begin{proof}{\underline{\bfseries\sffamily Proof of \Cref{lem:steps}}}
\begin{flalign*}
\int_{\mu}^{\bv} \lambda(v)dv&= \int_{\mu}^{\nu_0}\lambda(v)dv+\int_{\nu_0}^{\bv}\lambda(v)dv\\
&= \int_{\mu}^{\nu_0}\lambda(v)dv+\int_{\nu_0}^{\bv}f(v)dv\\
&= \sum_{k=1}^{K-1} \int_{\nu_k }^{\nu_{k-1}}\lambda(v)dv +\int_{\mu}^{\nu_{K-1}}\lambda(v)dv +\int_{\nu_0}^{\bv}f(v)dv\\
&= \sum_{k=1}^{K-1} \int_{\nu_k}^{\nu_{k-1}}\left(
\frac{d}{dv}\sum_{j=0}^{k} F(w^{j}(v))
\right)dv
+\int_{\mu}^{\nu_{K-1}}\left(
\frac{d}{dv}\sum_{j=0}^{K} F(w^{j}(v))
\right)dv
+\int_{\nu_0}^{\bv}f(v)dv\\
&= \sum_{k=1}^{K-1} \left(
\sum_{j=0}^{k} F(w^{j}(v))\right)\Big |_{\nu_k}^{\nu_{k-1}}
+
\left(
\sum_{j=0}^{K} F(w^{j}(v))
\right)\Big|_{\mu}^{\nu_{K-1}}
 +\int_{\nu_0}^{\bv}f(v)dv\\
 &= \sum_{k=1}^{K-1} \left(
\sum_{j=0}^{k} F(w^{j}(\nu_{k-1}))- F(w^{j}(\nu_k))\right)
+
\left(
\sum_{j=0}^{K} F(w^{j}(\nu_{K-1}))-F(w^{j}(\mu))
\right)
+\int_{\nu_0}^{\bv}f(v)dv\\
  &= \sum_{k=1}^{K-1} \left(
\sum_{j=0}^{k} (F(\nu_{k-j-1})-F(\nu_{k-j}))\right)
+
\left(
\sum_{j=0}^{K} (F(\nu_{K-1-j})-F(w^{j}(\mu)))
\right)
+\int_{\nu_0}^{\bv}f(v)dv\\
&\leq  \sum_{k=1}^{K-1} (1-F(\nu_k))
+
\left(
\sum_{j=0}^{K} (F(\nu_{K-1-j})-F(\nu_{K-j}))
\right)
+\int_{\nu_0}^{\bv}f(v)dv\\
 &=  \sum_{k=1}^{K-1} (1-F(\nu_k))
+(1-F(\nu_{K}))+\int_{\nu_0}^{\bv}f(v)dv\\
&\leq K+1,
\end{flalign*}
where in the above we used the property that $\w^j(\nu_k)=\nu_{k-j}$ and \Cref{lem:charac_lambda_ic}. In the first inequality, we used that $\w^j(\mu)\geq \nu_{K-j}$ which is true because $\mu\ge \nu_K$.
In the last inequality, we used that
$\int_{\nu_0}^{\bv}f(v)dv\leq 1$. This concludes the proof.
\end{proof}

\begin{proof}{\underline{\bfseries\sffamily Proof of \Cref{lem:bound_rev}}}
We separate the computation of the integral $\int_{x}^{\bv} \Delta(v) dv$ into three pieces. We do the computation for each piece for an arbitrary interval, $[a,b]$, first. We then combine all the pieces and intervals.

We will make the following definitions.
Since $\lambda(v)$ is defined differently
in different interval of the form $[\nu_k,\nu_{k-1}]$ (see \Cref{lem:charac_lambda_ic}), we will consider first  $[a,b] \subset[\nu_k,\nu_{k-1}]$ for some $k\geq 1$.
We also define
\begin{equation*}
H_{k}(v) \triangleq\sum_{j=0}^k F(w^j(v)), \quad \text{for } v\in [\nu_k,\nu_{k-1}].
\end{equation*}
Note that from \Cref{lem:charac_lambda_ic} for $v\in [\nu_k,\nu_{k-1}]$, we have that $\lambda(v)=\dot{H}_k(v)$ a.e.

\underline{First piece.}
\begin{equation*}
A_1 = \int_a^b v f(v)dv = vF(v)\Big|_a^b-\int_{a}^bF(v)dv.
\end{equation*}

\underline{Second piece.}
\begin{flalign*}
A_2 &= \int_a^b \w(v)\dot{\w}(v)\lambda(\w(v))dv \\
&= \int_{\w(a)}^{\w(b)}v \lambda(v)dv\\
&= \int_{\w(a)}^{\w(b)}v \dot{H}_{k-1}(v)dv\\
&= v H_{k-1}(v)\Big|_{\w(a)}^{\w(b)} - \int_{\w(a)}^{\w(b)}H_{k-1}(v)dv\\
&= \w(b) H_{k-1}(\w(b))-\w(a) H_{k-1}(\w(a)) - \int_{\w(a)}^{\w(b)}H_{k-1}(v)dv,
\end{flalign*}
where in the second equality we used the change of variables $u=w(v)$. In the third equality, we used that
because $[a,b] \subset[\nu_k,\nu_{k-1}]$ we must have
$[\w(a),\w(b)] \subset[\nu_{k-1},\nu_{k-2}]$
and, therefore, $\lambda(v)=\dot{H}_{k-1}(v)$ a.e.

\underline{Third piece.}
\begin{flalign*}
A_3 &= \int_a^b v\dot{\w}(v)\lambda(\w(v))dv \\
&= \int_a^b v\dot{\w}(v)\dot{H}_{k-1}(\w(v))dv \\
&= v H_{k-1}(\w(v))\Big|_{a}^{b} - \int_{a}^{b}H_{k-1}(\w(v))dv\\
&= b H_{k-1}(\w(b))-a H_{k-1}(\w(a)) - \int_{a}^{b}H_{k-1}(\w(v))dv
\end{flalign*}

Using the expressions above we can compute the integrals in $[x,\nu_0]$. We will need the following properties
\begin{equation}\label{eq:up_prop1}
H_{k}(\nu_{k-1})= \sum_{j=0}^k F(w^j(\nu_{k-1}))=\sum_{j=0}^k F(\nu_{k-1-j})=\sum_{j=0}^k F(\nu_{j-1}),
\end{equation}
and
\begin{equation}\label{eq:up_prop2}
H_{k}(\nu_{k}) = \sum_{j=0}^k F(w^j(\nu_{k}))=\sum_{j=0}^k F(\nu_{k-j})=\sum_{j=0}^{k+1} F(\nu_{j-1}) - 1,
\end{equation}
and
\begin{equation}\label{eq:up_prop3}
H_{k}(w(v)) = \sum_{j=0}^k F(w^j(w(v)))=\sum_{j=0}^k F(w^{j+1}(v))
=\sum_{j=1}^{k+1} F(w^{j}(v))= H_{k+1}(v)-F(v)
\end{equation}

 Next we exploit the properties above to compute the three integrals in $[x,\nu_0]$. Let $k$ be such that $\nu_k\leq x\leq \nu_{k-1}$. For the second piece we have
\begin{flalign*}
\int_{x}^{\nu_0} \w(v)\dot{\w}(v)\lambda(\w(v))dv &=\sum_{j=1}^{k-1}\int_{\nu_j}^{\nu_{j-1}}\w(v)\dot{\w}(v)\lambda(\w(v))dv +\int_{x}^{\nu_{k-1}}\w(v)\dot{\w}(v)\lambda(\w(v))dv\\
&=\sum_{j=1}^{k-1}\w(\nu_{j-1}) H_{j-1}(\w(\nu_{j-1}))-\w(\nu_{j}) H_{j-1}(\w(\nu_{j})) - \int_{\w(\nu_{j})}^{\w(\nu_{j-1})}H_{j-1}(v)dv\\
&+\w(\nu_{k-1}) H_{k-1}(\w(\nu_{k-1}))-\w(x) H_{k-1}(\w(x)) - \int_{\w(x)}^{\w(\nu_{k-1})}H_{k-1}(v)dv\\
&=\sum_{j=1}^{k-1}\nu_{j-2} H_{j-1}(\nu_{j-2})-\nu_{j-1} H_{j-1}(\nu_{j-1}) - \int_{\nu_{j-1}}^{\nu_{j-2}}H_{j-1}(v)dv\\
&+\nu_{k-2} H_{k-1}(\nu_{k-2})-\w(x) H_{k-1}(\w(x)) - \int_{\w(x)}^{\nu_{k-2}}H_{k-1}(v)dv
\end{flalign*}
For the third piece we have
\begin{flalign*}
\int_{x}^{\nu_0} v\dot{\w}(v)\lambda(\w(v))dv &=\sum_{j=1}^{k-1}\int_{\nu_j}^{\nu_{j-1}}v\dot{\w}(v)\lambda(\w(v))dv +\int_{x}^{\nu_{k-1}}v\dot{\w}(v)\lambda(\w(v))dv\\
&=\sum_{j=1}^{k-1} \nu_{j-1} H_{j-1}(\nu_{j-2})-\nu_{j} H_{j-1}(\nu_{j-1}) - \int_{\nu_{j}}^{\nu_{j-1}}H_{j-1}(\w(v))dv\\
&+ \nu_{k-1} H_{k-1}(\nu_{k-2})-x H_{k-1}(\w(x)) - \int_{x}^{\nu_{k-1}}H_{k-1}(\w(v))dv
\end{flalign*}

Combining the second and third pieces  and using properties \eqref{eq:up_prop1}, \eqref{eq:up_prop2} and \eqref{eq:up_prop3}
yields
\begin{flalign*}
\int_{x}^{\nu_0} (\w(v)-v)\dot{\w}(v)\lambda(\w(v))dv &=
\sum_{j=1}^{k-1}\left(\nu_{j-2} -\nu_{j-1} \right)H_{j-1}(\nu_{j-2})
-\left(\nu_{j-1}  -\nu_{j}  \right)H_{j-1}(\nu_{j-1})\\
& -\left( \int_{\nu_{j-1}}^{\nu_{j-2}}H_{j-1}(v)dv-
\int_{\nu_{j}}^{\nu_{j-1}}H_{j-1}(\w(v))dv
\right)\\
&+(\nu_{k-2}-\nu_{k-1}) H_{k-1}(\nu_{k-2})-(\w(x)-x) H_{k-1}(\w(x))\\
& - \left(\int_{\w(x)}^{\nu_{k-2}}H_{k-1}(v)dv-\int_{x}^{\nu_{k-1}}H_{k-1}(\w(v))dv\right)\\
&=
\sum_{j=1}^{k-1}\left(\nu_{j-2} -\nu_{j-1} \right)\sum_{i=0}^{j-1} F(\nu_{i-1})
-\left(\nu_{j-1}  -\nu_{j}  \right)\left(
\sum_{i=0}^{j} F(\nu_{i-1}) - 1
\right)\\
& -\left( \int_{\nu_{j-1}}^{\nu_{j-2}}H_{j-1}(v)dv - \int_{\nu_{j}}^{\nu_{j-1}}\left(H_{j}(v)-F(v)\right)dv
\right)\\
&+(\nu_{k-2}-\nu_{k-1}) H_{k-1}(\nu_{k-2})-(\w(x)-x) H_{k-1}(\w(x))\\
& - \left(\int_{\w(x)}^{\nu_{k-2}}H_{k-1}(v)dv-\int_{x}^{\nu_{k-1}}\left(H_{k}(v)-F(v)\right)dv\right).
\end{flalign*}
Note that
\begin{flalign*}
\sum_{j=1}^{k-1}\left(\nu_{j-2} -\nu_{j-1} \right)\sum_{i=0}^{j-1} F(\nu_{i-1})
-\left(\nu_{j-1}  -\nu_{j}  \right)\left(
\sum_{i=0}^{j} F(\nu_{i-1}) - 1
\right)&=-(\nu_{k-1}-\nu_0) + \left(\bv  -\nu_{0}  \right)\\
&- \left(\nu_{k-2}  -\nu_{k-1}  \right)
\sum_{i=0}^{k-1} F(\nu_{i-1}) \\
&=(\bv-\nu_{k-1}) -\left(\nu_{k-2}  -\nu_{k-1}  \right)H_{k-1}(\nu_{k-2}).
\end{flalign*}
We also have that
\begin{flalign*}
\sum_{j=1}^{k-1}\left( \int_{\nu_{j-1}}^{\nu_{j-2}}H_{j-1}(v)dv - \int_{\nu_{j}}^{\nu_{j-1}}\left(H_{j}(v)-F(v)\right)dv
\right) &= \int_{\nu_0}^{\bv}H_{0}(v)dv-\int_{\nu_{k-1}}^{\nu_{k-2}}H_{k-1}(v)dv + \int_{\nu_{k-1}}^{\nu_0}F(v)dv\\
&= \int_{\nu_{k-1}}^{\bv}F(v)dv-\int_{\nu_{k-1}}^{\nu_{k-2}}H_{k-1}(v)dv .
\end{flalign*}
In turn, we have that
\begin{flalign*}
\int_{x}^{\nu_0} (\w(v)-v)\dot{\w}(v)\lambda(\w(v))dv &=(\bv-\nu_{k-1}) -\left(\nu_{k-2}  -\nu_{k-1}  \right)H_{k-1}(\nu_{k-2})
-\int_{\nu_{k-1}}^{\bv}F(v)dv+\int_{\nu_{k-1}}^{\nu_{k-2}}H_{k-1}(v)dv \\
&+(\nu_{k-2}-\nu_{k-1}) H_{k-1}(\nu_{k-2})-(\w(x)-x) H_{k-1}(\w(x))\\
& - \left(\int_{\w(x)}^{\nu_{k-2}}H_{k-1}(v)dv-\int_{x}^{\nu_{k-1}}\left(H_{k}(v)-F(v)\right)dv\right)\\
&=(\bv-\nu_{k-1})-\int_{x}^{\bv}F(v)dv+\int_{x}^{w(x)}H_{k-1}(v)dv -(\w(x)-x) H_{k-1}(\w(x))\\
&+\int_{x}^{\nu_{k-1}}H_{k}(v)dv.
\end{flalign*}
Therefore,
\begin{flalign*}
\int_{x}^{\nu_0} \Delta(v) dv+ \int_{\nu_0}^{\bv} vf(v)dv&=  \nu_{k-1}-xF(x)
-\int_{x}^{w(x)}H_{k-1}(v)dv +(\w(x)-x) H_{k-1}(\w(x))-\int_{x}^{\nu_{k-1}}H_{k}(v)dv\\
&=  (\nu_{k-1}-x)+ x\bar{F}(x)
-\int_{x}^{w(x)}H_{k-1}(v)dv +(\w(x)-x) H_{k-1}(\w(x))\\
&-\int_{x}^{\nu_{k-1}}H_{k}(v)dv.
\end{flalign*}
Now, we have that $H_{k-1}(w(x))=H_{k}(x)-F(x)=H_{k-1}(x) + F(w^k(x))-F(x)$. Therefore,
\begin{flalign*}
(\w(x)-x) H_{k-1}(\w(x))&=(\w(x)-x) H_{k-1}(x) +(\w(x)-x)(F(w^k(x))-F(x))\\
&\leq \int_{x}^{w(x)}H_{k-1}(v)dv +(\w(x)-x)(F(w^k(x))-F(x)),
\end{flalign*}
where the inequality follows from $H_{k-1}(v)$ being non-decreasing (composition of non-decreasing functions). So finally, we deduce
\begin{flalign*}
\int_{x}^{\nu_0} \Delta(v) dv+ \int_{\nu_0}^{\bv} vf(v)dv&\leq   (\nu_{k-1}-x)+ x\bar{F}(x)
+(\w(x)-x)(F(w^k(x))-F(x))\\
&\leq   (w(x)-x)+ x\bar{F}(x)
+(\w(x)-x)\\
&=x\bar{F}(x)+2(\w(x)-x),
\end{flalign*}
where in the second inequality we used that $\nu_{k-1}\leq \w(x)$ and that $(F(w^k(x))-F(x))\leq 1$.
\end{proof}

\begin{proof}{\underline{\bfseries\sffamily Proof of \Cref{prop:numberOfSteps}}}
First, let us consider counting the number of steps in $[\ps,\nu_0]$. For $k$ such that $\nu_k\geq \ps$ we have
$$
\nu_k = \frac{1}{m^k}\left\{
\nu_0 +\ps(m^k-1) -\eps^{1-\beta} \frac{m^k-1}{m-1}\right\}.
$$
The number of steps in $v\geq \ps$ equals $k$ where $k$ is the largest $k$ such that $\nu_k\geq \ps$. We solve for $k$ such that
$$
\frac{1}{m^{k}}\left\{
\nu_0 +\ps(m^{k}-1) -\eps^{1-\beta} \frac{m^{k}-1}{m-1}\right\} \le \ps+\eps^{1-\beta}\le \frac{1}{m^{k-1}}\left\{
\nu_0 +\ps(m^{k-1}-1) -\eps^{1-\beta} \frac{m^{k-1}-1}{m-1}\right\}.
$$
This is equivalent to
$$
\left\{\frac{(\nu_0-\ps)(m-1)}{\eps^{1-\beta}}+1\right\}\frac{1}{m} \le m^{k}\le \frac{(\nu_0-\ps)(m-1)}{\eps^{1-\beta}}+1.
$$
Note that
\begin{flalign*}
\frac{(\nu_0-\ps)(m-1)}{\eps^{1-\beta} }
&= \frac{(\nu_0-\ps)}{\eps^{1-\beta} } \left( \frac{\bv-\ps-\eps^{1-\beta}}{\nu_0-\ps}-1\right)\\
&= \frac{(\nu_0-\ps)}{\eps^{1-\beta} } \left( \frac{\bv-\ps-\eps^{1-\beta}}{\nu_0-\ps}-\frac{\nu_0-\ps}{\nu_0-\ps}\right)\\
&= \frac{(\nu_0-\ps)}{\eps^{1-\beta} } \left( \frac{\bv-\ps-\eps^{1-\beta} -\nu_0+\ps}{\nu_0-\ps}\right)\\
&= \frac{1}{\eps^{1-\beta} } \left( \bv-\eps^{1-\beta} -\nu_0\right)\\
&= \frac{1}{\eps^{1-\beta} } \left( \bv -\nu_0\right) -1\\
&=\eps^{\beta-({1-\beta})} -1.
\end{flalign*}
In turn, we have
\begin{equation}\label{eq:rev-step-1}
\frac{\eps^{2\beta-1}}{m} \le m^{k}\le \eps^{2\beta-1},
\end{equation}
and, thus,
$$
\frac{\log\left(\frac{\eps^{2\beta-1}}{m}\right)}{\log(m)} \le k\le \frac{\log\left(\eps^{2\beta-1}\right)}{\log(m)},
$$
Note that,
$$
\log(m)= \log\left(m-1+1\right) = \log\left(\frac{\bv-\eps^{1-\beta} -\nu_0}{\nu_0-\ps}+1\right)
= \log\left(\frac{\eps^\beta-\eps^{1-\beta}}{\bv -\eps^\beta -\ps}+1\right) = \mathcal{O}(\eps^\beta),
$$
where in the last equality we used that $\beta< 1/2$. This shows that the number of steps above $\ps$ is of order
$\Theta\left(\log(\eps^{2\beta-1})/\eps^\beta\right)$.

Next, we bound the number of steps below $\ps$. Let $k^\star$ be the largest $k$ such that $\nu_k\geq \ps$. We can obtain the number of steps $j+1$ below $\ps$ by finding the largest $j$ such that $\nu_{k^\star+j}\geq \mu$. First note that
$$
\nu_{k^\star+j} =   \frac{1}{(2-m)^j}\left\{
\nu_{k^\star} +\ps((2-m)^j-1) -\eps^{1-\beta} \frac{(2-m)^j-1}{2-m-1}\right\}.
$$
We use the above to solve for the largest $j$ such that $\nu_{k^\star+j} \geq \mu$. Recall that
$$w(0)=\mu = -\ps(2-m) +\ps +\eps^{1-\beta}
= \ps(m-1) +\eps^{1-\beta}
$$
We have
$$
 \frac{1}{(2-m)^j}\left\{
\nu_{k^\star} +\ps((2-m)^j-1) -\eps^{1-\beta} \frac{(2-m)^j-1}{2-m-1}\right\}\geq \ps(m-1) +\eps^{1-\beta}
$$
which is equivalent to
$$
\eps^{1-\beta} -\ps(1-m) + \nu_{k^\star}(1-m)\leq (2-m)^{j+1}\left( \ps (m-1) +\eps^{1-\beta}
\right),
$$
hence  we have
$$
\log\left(1-\frac{ \nu_{k^\star}(m-1)}
{\ps(m-1)+\eps^{1-\beta}}
\right)\le (j+1) \log(2-m)
$$
or equivalently
$$
  (j+1)\le \frac{\log\left(1-\frac{ \nu_{k^\star}(m-1)}
{\ps(m-1)+\eps^{1-\beta}}
\right)}{\log(2-m)}
$$
We note that $\log(2-m)<0$ and
$$
\log(2-m) =  \log(1-(m-1))=\log\left(1- \frac{\eps^\beta-\eps^{1-\beta}}{\bv -\eps^\beta -\ps}\right)=\mathcal{O} (-\eps^{\beta}).
$$
Also note that by the definition of $\nu_{k^\star}$ we have
$$
\frac{\nu_{k^\star}-\ps -\eps^{1-\beta}}{m}+\ps \le \ps +\eps^{1-\beta}
\Leftrightarrow
\nu_{k^\star}  \le \ps + (m+1)\eps^{1-\beta}.
$$
So
$$
1-\frac{ \nu_{k^\star}(m-1)}
{\ps(m-1)+\eps^{1-\beta}} \geq
1-\frac{ (\ps + (m+1)\eps^{1-\beta})(m-1)}
{\ps(m-1)+\eps^{1-\beta}} =
\frac{  (2- m^2)\eps^{1-\beta}}
{\ps(m-1)+\eps^{1-\beta}}  =\mathcal{O}(\eps^{1-2\beta}).
$$
Putting all together we deduce that
$$
  (j+1) =\mathcal{O} \left(\frac{\log(\eps^{2\beta-1})}{\eps^\beta}\right).
 $$
In conclusion, the total number of steps is $\mathcal{O} \left(\frac{\log(\eps^{2\beta-1})}{\eps^\beta}\right)$. This concludes the proof.

\end{proof}

\begin{proof}{\underline{\bfseries\sffamily Proof of \Cref{thm:up-bd-per}}}
We show that
\begin{equation*}
    \Pi^\star(\mathcal{D})\leq \Pi^\star(\mathcal{M}(0)) +\tilde{\mathcal{O}}(\eps^{\alpha/(2\alpha-1)}).
\end{equation*}
By combining \Cref{lem:steps} and \Cref{prop:numberOfSteps},  we can deduce that
$\Pi^\star(\mathcal{D})-\Phi_2(\Lic)= \tilde{\mathcal{O}}(\eps^{1-\beta})$. It remains to provide an upper bound for $\Phi_2(\Lic)$.

We separate the proof into two cases:  $v\geq\ps$ and  $v\leq\ps$.
In both cases we provide bounds for $\Lic(\w(v))$ which we then use to further bound $\Phi_2(\Lic)$.
For $v\geq\ps$, we will provide an upper and a lower bound for $\Lic(\w(v))$. For $v<\ps$, we will provide an upper bound for $\Lic(\w(v))$. These bounds will be tight and will enable us to use \Cref{lem:bound_rev} to bound $\Phi_2(\Lic)$.

\underline{Case 1:} We consider the case  $v\geq\ps$. For ease of exposition, we use $\gamma$ to denote $1-\beta$. First, we characterize the difference $\nu_{k-1}-\nu_{k}$ for $\nu_k\geq \ps$.
\begin{flalign*}
\nu_{k-1}-\nu_{k} &= \frac{1}{m^{k-1}}\left\{
\nu_0 +\ps(m^{k-1}-1) -\eps^\gamma \frac{m^{k-1}-1}{m-1}\right\}
- \frac{1}{m^k}\left\{
\nu_0 +\ps(m^k-1) -\eps^\gamma \frac{m^k-1}{m-1}\right\}\\
&=  \frac{1}{m^k}\left\{
m\nu_0 +\ps(m^{k}-m) -\eps^\gamma \frac{m^{k}-m}{m-1}
-\nu_0 -\ps(m^k-1) +\eps^\gamma \frac{m^k-1}{m-1}\right\}\\
&= \frac{1}{m^k}\left\{
(m-1)\nu_0 -\ps(m-1) +\eps^\gamma\right\}\\
&= \frac{1}{m^k}\left\{
(m-1)(\nu_0 -\ps) +\eps^\gamma\right\}\\
&= \frac{1}{m^k}\left\{
\eps^\beta -\eps^\gamma +\eps^\gamma\right\}\\
&= \frac{\eps^\beta}{m^k}.
\end{flalign*}
Note that from \Cref{lem:charac_lambda_ic} for $k\geq1$, we have
$$
\lambda(w(v)) = \sum_{i=1}^k m^{i-1} f(w^i(v)), \quad a.e \text{ in } [\nu_k,\nu_{k-1}]\cap[\ps,\bv].
$$
First, note that the following holds for $v\in [\nu_k,\nu_{k-1}]\cap[\ps,\bv]$
$$
 \sum_{i=1}^k m^{i-1} f(w^i(v)) = m^{k-1}  \sum_{i=0}^{k-1} m^{-i} f(w^{k-i}(v)).
$$
In order to bound the expression above, we make use of the mean value theorem for integrals. We have that since $v\in[\nu_k,\nu_{k-1}]$ then $w^{k-i}(v) \in[\nu_i,\nu_{i-1}]$ and there exists $\xi \in[\nu_i,\nu_{i-1}]$ such that
\begin{flalign*}
\int_{\nu_i}^{\nu_{i-1}} f(v)dv &=  f(\xi)(\nu_{i-1}-\nu_{i} )\\
&=  (f(\xi)-f(w^{k-i}(v)))(\nu_{i-1}-\nu_{i} )+ f(w^{k-i}(v))(\nu_{i-1}-\nu_{i} )\\
&=  \underbrace{(f(\xi)-f(w^{k-i}(v)))}_{\chi_i(v)}\frac{\eps^\beta}{m^i}+ f(w^{k-i}(v))\frac{\eps^\beta}{m^i}.
\end{flalign*}
Then
$$
 m^{k-1}  \sum_{i=0}^{k-1}\int_{\nu_i}^{\nu_{i-1}} f(v)dv= m^{k-1}  \sum_{i=0}^{k-1}\chi_i(v)\frac{\eps^\beta}{m^i}+  m^{k-1}  \sum_{i=0}^{k-1}f(w^{k-i}(v))\frac{\eps^\beta}{m^i}
$$
Equivalently,
$$
\frac{m^{k-1}}{\eps^\beta} \bar{F}(\nu_{k-1})= m^{k-1}  \sum_{i=0}^{k-1}\frac{\chi_i(v)}{m^i}+  \lambda(w(v)).
$$
Note that
$$
\left|   \sum_{i=0}^{k-1}\frac{\chi_i(v)}{m^i} \right|\leq \sum_{i=0}^{k-1}\frac{|\chi_i(v)|}{m^i}
\leq \sum_{i=0}^{k-1}\frac{|\chi_i(v)|}{m^i} \leq \sum_{i=0}^{k-1}|\chi_i(v)|\leq V_0^{\bv}(f),
$$
where we have used that $m>1$ and that $V_0^{\bv}(f)$ denotes the total variation\footnote{The total variation of a function $f$ over $[0,\bv]$ is given by $V_0^{\bv}(f) = \sup_{P\in\mathcal{P} } \sum_{i=0}^{n_P-1} |f(x_{i+1})-f(x_i)|,$
where $\mathcal{P} $ is the set of partitions in $[0,\bv]$ and $n_P$ is the number of intervals in partition $P$.} $f(\cdot)$ which we assume to be bounded (see \Cref{assumption1}), and that
$\xi \in[\nu_i,\nu_{i-1}]$ and $w^{k-i}(v) \in[\nu_i,\nu_{i-1}]$. Hence, for $v\in [\nu_k,\nu_{k-1}]\cap[\ps,\bv]$ (a.e)
$$
  \frac{m^{k-1}}{\eps^\beta} \bar{F}(\nu_{k-1})  - m^{k-1}V_0^{\bv}(f)\leq \lambda(w(v)) \leq  \frac{m^{k-1}}{\eps^\beta} \bar{F}(\nu_{k-1})  +m^{k-1}V_0^{\bv}(f).
$$
Note that $\bar{F}(\nu_{k-1})\leq \bar{F}(v)$. Also,
$$
\bar{F}(\nu_{k-1}) = \bar{F}(v) -\int_{v}^{\nu_{k-1}}f(x)dx \geq \bar{F}(v) -\bar{f}(\nu_{k-1}-v)\geq \bar{F}(v) -\bar{f}\cdot (w(v)-v).
$$
Additionally, we note that $(w(v)-v)$ in increasing for $v\geq \ps$ hence $(w(v)-v)\leq \nu_{k-2}-\nu_{k-1} = \eps^\beta/m^{k-1}$. Then,
for $v\in [\nu_k,\nu_{k-1}]\cap[\ps,\bv]$ (a.e)
$$
  \frac{m^{k-1}}{\eps^\beta} \bar{F}(v) - \bar{f}  - m^{k-1}V_0^{\bv}(f)\leq \lambda(w(v)) \leq  \frac{m^{k-1}}{\eps^\beta} \bar{F}(v)  +m^{k-1}V_0^{\bv}(f).
$$

We thus have that the following holds almost everywhere for $v \in [\ps,\nu_0]$
\begin{flalign*}
\Delta(v) &= vf(v) - \dot{w}(v)(w(v)-v)\lambda(w(v))\\
&\geq vf(v) - m\frac{\eps^\beta}{m^{k-1}}\lambda(w(v)) \\
&\geq vf(v) - m\frac{\eps^\beta}{m^{k-1}}\left( \frac{m^{k-1}}{\eps^\beta} \bar{F}(v)  +m^{k-1}V_0^{\bv}(f)\right)\\
&= vf(v) - m\bar{F}(v)  -m\eps^\beta V_0^{\bv}(f)\\
&= vf(v) - \bar{F}(v) -(m-1)\bar{F}(v)  -m\eps^\beta V_0^{\bv}(f)\\
&= vf(v) - \bar{F}(v) -\frac{\eps^\beta-\eps^\gamma}{\nu_0-\ps}  -m\eps^\beta V_0^{\bv}(f)\\
&\ge vf(v) - \bar{F}(v) -C_L \eps^\beta,
\end{flalign*}
where $C_L$ is a positive constant. In the last inequality we have used $\eps>0$ small. We also have
\begin{flalign*}
\Delta(v) &= vf(v) - \dot{w}(v)(w(v)-v)\lambda(w(v))\\
&\leq vf(v) - \dot{w}(v)(w(\nu_k)-\nu_k)\lambda(w(v))\\
&\leq vf(v) - m\frac{\eps^\beta}{m^{k}}\lambda(w(v)) \\
&\leq  vf(v) - \frac{\eps^\beta}{m^{k-1}}\left( \frac{m^{k-1}}{\eps^\beta} \bar{F}(v)  -\bar{f} - m^{k-1}V_0^{\bv}(f)\right)\\
&\leq  vf(v) - \bar{F}(v)  +\eps^\beta \left\{\bar{f} +V_0^{\bv}(f)\right\}\\
&= vf(v) - \bar{F}(v)  +C_U\eps^\beta,
\end{flalign*}
where $C_U$ is a positive constant. We have thus shown that
\begin{equation*}
\Delta_L(v)\triangleq vf(v) - \bar{F}(v) -C_L \eps^\beta \leq \Delta(v)\leq  vf(v) - \bar{F}(v)  +C_U\eps^\beta\triangleq \Delta_U(v), \quad v\in [\ps,\nu_0] \: \: \text{ a.e}.
\end{equation*}
Note that $\Delta_L(v)=-\dot{R}(v)-C_L\eps^\beta$ and, hence, $\Delta_L(v)\ge 0$ if and only if
$\dot{R}(v)\leq -C_L\eps^\beta$. By \Cref{def1}, the latter is true for all $v$ such that $-\kappa_L\alpha(v-\ps)^{\alpha-1}\leq -C_L\eps^\beta$ or, equivalently, for all $v\in (\ps,\ps+\ell)$ such that $v\geq \ps + (C_L/(\kappa_L\alpha))^{1/(\alpha-1)}\eps^{\beta/(\alpha-1)}$. We define
\begin{equation*}
\ps_L=\ps + \underbrace{(C_L/(\kappa_L\alpha))^{1/(\alpha-1)}}_{C_L'(\alpha)}\eps^{\beta/(\alpha-1)},
\end{equation*}
and we consider $\eps>0$ small enough such that $\ps_L\in (\ps,\ps+\ell)$. Hence, for  $v\in [\ps_L,\ps+\ell)$ we have that $\Delta_L(v)\ge 0$. Now, for $v\geq \ps+\ell$ we use the last assumption in the statement of the theorem. From the assumption, we can see that it is always possible to consider $\eps>0$ small enough such that
$\sup_{v\geq \ps+\ell} \dot{R}(v) \leq -C_L\eps^\beta$ (because the sup is strictly negative). Therefore, for any $v\geq \ps+\ell$ we have that  $\Delta_L(v)\ge 0$. In conclusion, we can always consider $\eps>0$ small enough such that  for $v\geq \ps_L$ we have that  $\Delta_L(v)\ge 0$.  Hence, considering $\eps>0$ small enough  such that $\ps_L< \nu_0$, \Cref{lem:bound_rev} with $x=\ps_L$ implies
\begin{flalign*}
\int_{\ps_L}^{\bv}\Delta(v)^+dv &= \int_{\ps_L}^{\bv}\Delta(v)dv \\
&\leq \ps_L \bar{F}(\ps_L) + 2(w(\ps_L)-\ps_L)\\
&= R(\ps_L) + 2(w(\ps_L)-\ps_L)\\
&= R(\ps_L) + (m-1) C_L'(\alpha) \eps^{\frac{\beta}{\alpha-1}} +\eps^\gamma\\
&= R(\ps_L) + \mathcal{O}\left(\eps^{\beta+\frac{\beta}{\alpha-1}}\right) +\eps^\gamma,
\end{flalign*}
where we have used that $m-1 =\mathcal{O}(\eps^\beta)$. We also have, $\Delta_U(v)\geq 0$ for $v\geq \ps$ and
\begin{flalign*}
\int_{\ps}^{\ps_L}\Delta(v)^+dv &\leq \int_{\ps}^{\ps_L}\Delta_U(v)dv=R(\ps) -R(\ps_L) +C_U\eps^\beta C_L'(\alpha) \eps^{\frac{\beta}{\alpha-1}}.
\end{flalign*}
Then,
\begin{equation*}
\int_{\ps}^{\bv}\Delta(v)^+dv \leq R(\ps)+ \mathcal{O}\left(\eps^{\beta+\frac{\beta}{\alpha-1}}\right) +\eps^\gamma.
\end{equation*}
Now, we still need to bound $\Phi_2(\Lic)$ for $v\leq \ps$. Before we do that we can check that we are getting the right order. From the bound for the number of steps, we have a term of order $\eps\cdot \log(\eps^{\beta-\gamma})/\eps^\beta$. So if we take $\beta = (\alpha-1)/(2\alpha-1)$. We will have that
$$
\gamma = 1-\beta= \frac{\alpha}{2\alpha-1} \quad \text{and} \quad \beta + \frac{\beta}{\alpha-1}= \frac{\alpha}{2\alpha-1}.
$$
Note that $\gamma >\beta$ if and only if $\beta<1/2$ which is true if $\alpha<\infty$. In sum, this gives the upper bound of order $\frac{\alpha}{2\alpha-1}$. As desired.

\underline{Case 2:} We consider the case  $v\leq\ps$.
We will provide an upper bound for $\Delta(v)$ for $v\leq \ps$. Then we will show that the upper is negative except close to $\ps$. The area below the upper bound in such region will be of the desired order.  We have the following
\begin{flalign*}
\Delta(v) &= vf(v) - \dot{w}(v)(w(v)-v)\lambda(w(v))\\
&= vf(v) - (2-m)(w(v)-v)\lambda(w(v)).
\end{flalign*}
We now provide bounds for the terms $(w(v)-v)$ and $\lambda(w(v))$. Let $k^\star$ be the largest $k$ such that $\nu_k\geq \ps$. We consider $v\in[\nu_{k^\star+j},\nu_{k^\star+j-1}] \cap [0,\ps]$ for $j\geq 1$. We let $\bar{j}$ to be the largest $j$ such that $\nu_{k^\star+j}\geq 0$ and define
$\nu_{k^\star+\bar{j}+1}=0$. For $j \in\{1,\dots,\bar{j}\}$, we have
\begin{flalign}\nonumber
\nu_{k^\star+j-1}- \nu_{k^\star+j} &=\frac{1}{(2-m)^{j-1}}\left\{
\nu_{k^\star} +\ps((2-m)^{j-1}-1) -\eps^\gamma \frac{(2-m)^{j-1}-1}{2-m-1}
\right\}\\\nonumber
&-\frac{1}{(2-m)^j}\left\{
\nu_{k^\star} +\ps((2-m)^j-1) -\eps^\gamma \frac{(2-m)^j-1}{2-m-1}\right\}\\\nonumber
&=\frac{1}{(2-m)^{j}}\left\{
(1-m)\nu_{k^\star} +\ps(m-1) +\eps^\gamma
\right\}\\ \label{eq:rev-diff-below-1}
&=\frac{1}{(2-m)^{j}}\left\{
 (\ps-\nu_{k^\star})(m-1) +\eps^\gamma
\right\}.
\end{flalign}
Additionally, since $\nu_{k^\star+j}\geq 0$ we have that
$$
(2-m)^{j-1} ( (m-1)\ps +\eps^\gamma )  \geq  (\ps-\nu_{k^\star})(m-1)+\eps^\gamma,
$$
note that $\nu_{k^\star}-\ps \leq w(\ps) -\ps =\eps^\gamma$. Hence, $\nu_{k^\star+j}\geq 0$ implies that
\begin{equation}\label{eq:rev-j-below-2}
\frac{1}{(2-m)^{j-1}}\leq \frac{1}{(2-m)} \left( \frac{(m-1)\ps}{\eps^\gamma} +1 \right), \quad j \in\{1,\dots,\bar{j}\}.
\end{equation}

We use the above to provide lower bounds for both $(w(v)-v)$ and $\lambda(w(v))$. We start with $w(v)-v$. Note that for $v\leq \ps$, $w(v)-v$ is decreasing. Hence for $v\in[\nu_{k^\star+j},\nu_{k^\star+j-1}] \cap [0,\ps]$ with $j \in\{1,\dots,\bar{j}+1\}$ we have
$$
w(v)-v \geq (2-m)\frac{\eps^\gamma}{(2-m)^{j-1}}.
$$
Indeed, for $j=1$ we have
$$w(v)-v\geq w(\ps)-\ps =\eps^\gamma \geq (2-m)\eps^\gamma,$$
where we have used that $(2-m)<1$. For $j>1$ we have
\begin{flalign*}
w(v)-v &\geq  w(\nu_{k^\star+j-1})-\nu_{k^\star+j-1}\\
& = \nu_{k^\star+j-2}-\nu_{k^\star+j-1}\\
&=\frac{1}{(2-m)^{j-1}}\left\{
 (\ps-\nu_{k^\star})(m-1) +\eps^\gamma \right\}\\
 & \geq (2-m)\frac{\eps^\gamma}{(2-m)^{j-1}}.
\end{flalign*}
Using the above, so far we have that for $v\in[\nu_{k^\star+j},\nu_{k^\star+j-1}] \cap [0,\ps]$ for  $j \in\{1,\dots,\bar{j}+1\}$
\begin{flalign*}
\Delta(v) &\leq  vf(v) - (2-m)^2\frac{\eps^\gamma}{(2-m)^{j-1}}\lambda(w(v)).
\end{flalign*}
Next, we provide a lower bound for $\lambda(w(v))$. From \Cref{lem:charac_lambda_ic} and since $(w(v))\geq \mu$ we have for $v\in[\nu_{k^\star+j},\nu_{k^\star+j-1}] \cap [0,\ps]$ with  $j \in\{1,\dots,\bar{j}+1\}$ that (almost everywhere)
\begin{flalign*}
\lambda(w(v)) &\geq \sum_{i=0}^{j-1} f(w^{i+1}(v))(2-m)^i + (2-m)^{j-1}\sum_{i=j}^{k^\star+j-1} f(w^{i+1}(v))m^{i-j+1}\\
&=(2-m)^{j-1}\underbrace{\sum_{i=0}^{j-1} f(w^{j-i}(v)) (2-m)^{-i}}_{(A)} +(2-m)^{j-1}m^{k^\star}\underbrace{\sum_{i=0}^{k^\star-1} f(w^{k^\star-i+j}(v))m^{-i}}_{(B)}.
\end{flalign*}
Next, we bound each sum above. For $(A)$ we have that since $v\in[\nu_{k^\star+j},\nu_{k^\star+j-1}]$ then $w^{j-i}(v) \in [\nu_{k^\star+i},\nu_{k^\star+i-1}]$. So we can always find $\xi \in  [\nu_{k^\star+i},\nu_{k^\star+i-1}]$ such that
\begin{flalign*}
\int _{\nu_{k^\star+i}}^{\nu_{k^\star+i-1}} f(v)dv &= \underbrace{(f(\xi)- f(w^{j-i}(v)))}_{\triangleq \chi_{j,i}(v)}\cdot(\nu_{k^\star+i-1} - \nu_{k^\star+i}) + f(w^{j-i}(v))\cdot(\nu_{k^\star+i-1} - \nu_{k^\star+i})\\
&\leq |\chi_{j,i}(v)| \cdot(\nu_{k^\star+i-1} - \nu_{k^\star+i}) + f(w^{j-i}(v))\cdot(\nu_{k^\star+i-1} - \nu_{k^\star+i})\\
&\leq |\chi_{j,i}(v)| \cdot m\frac{\eps^\gamma}{(2-m)^i} + f(w^{j-i}(v))\cdot m\frac{\eps^\gamma}{(2-m)^i},
\end{flalign*}
where we have used that for $i=0 $ the difference $\nu_{k^\star+i-1} - \nu_{k^\star+i}$ equals $\eps^\beta/m^{k^\star}$, and that $1/m^{k^\star}$ is bounded above by $m \eps^\gamma/\eps^\beta$ (see \cref{eq:rev-step-1}). For $i>0$ the bound comes from \cref{eq:rev-diff-below-1}. Next, by the fact that
$(2-m)<1$, that $f$ has bounded variation and \cref{eq:rev-j-below-2} we have
$$
\sum_{i=0}^{j-1} \frac{|\chi_{j,i}(v)|}{(2-m)^i}\le \sum_{i=0}^{j-1} \frac{|\chi_{j,i}(v)|}{(2-m)^{j-1}} \le V_0^{\bv}(f)\frac{1}{(2-m)} \left( \frac{(m-1)\ps}{\eps^\gamma} +1 \right).
$$
Hence,
$$
(A)\geq \frac{1}{m\eps^{\gamma}}\int _{\nu_{k^\star+j-1}}^{\nu_{k^\star-1}} f(v)dv - V_0^{\bv}(f)\frac{1}{(2-m)} \left( \frac{(m-1)\ps}{\eps^\gamma} +1 \right).
$$

Now, for $(B)$ we have the following that since $v\in[\nu_{k^\star+j},\nu_{k^\star+j-1}]$ then $w^{k^\star-i+j}(v) \in [\nu_{i},\nu_{i-1}]$. So we can always find $\xi \in  [\nu_{i},\nu_{i-1}]$ such that
\begin{flalign*}
\int _{\nu_{i}}^{\nu_{i-1}} f(v)dv &= \underbrace{(f(\xi)- f(w^{k^\star-i+j}(v)))}_{\triangleq \chi_{j,i}(v)}\cdot(\nu_{i-1} - \nu_{i}) + f(w^{k^\star-i+j}(v))\cdot(\nu_{i-1} - \nu_{i})\\
&\leq |\chi_{j,i}(v)|\cdot(\nu_{i-1} - \nu_{i}) + f(w^{k^\star-i+j}(v))\cdot(\nu_{i-1} - \nu_{i})\\
& = |\chi_{j,i}(v)|\cdot\frac{\eps^\beta}{m^i} + f(w^{k^\star-i+j}(v))\cdot \frac{\eps^\beta}{m^i}.
\end{flalign*}
In turn, using that $m>1$ and that $f$ has bounded variation, we have that
$$
(B)\geq  \frac{1}{\eps^\beta}\int_{\nu_{k^\star-1}}^{\bv}f(v)dv - V_0^{\bv}(f).
$$
With this, we obtain a lower bound for $\lambda(w(v))$. Recall from \cref{eq:rev-step-1} that $ \eps^{\beta-\gamma} \geq m^{k^\star} \geq \eps^{\beta-\gamma}/m$ then
\begin{flalign*}
\frac{\lambda(w(v))}{(2-m)^{j-1}} &\geq
\frac{1}{m\eps^{\gamma}}\int _{\nu_{k^\star+j-1}}^{\nu_{k^\star-1}} f(v)dv - V_0^{\bv}(f)\frac{1}{(2-m)} \left( \frac{(m-1)\ps}{\eps^\gamma} +1 \right)
+ \frac{m^{k^\star}}{\eps^\beta}\int_{\nu_{k^\star-1}}^{\bv}f(v)dv -m^{k^\star} V_0^{\bv}(f)\\
&\geq
\frac{1}{m\eps^{\gamma}}\int _{\nu_{k^\star+j-1}}^{\nu_{k^\star-1}} f(v)dv - V_0^{\bv}(f)\frac{1}{(2-m)} \left( \frac{(m-1)\ps}{\eps^\gamma} +1 \right)
+ \frac{1}{m\eps^\gamma}\int_{\nu_{k^\star-1}}^{\bv}f(v)dv - \eps^{\beta-\gamma} V_0^{\bv}(f)\\
&=
\frac{1}{m\eps^{\gamma}}\bar{F}(\nu_{k^\star+j-1})  - V_0^{\bv}(f)\frac{1}{(2-m)} \left( \frac{(m-1)\ps}{\eps^\gamma} + 1 - (2-m)\eps^{\beta-\gamma}\right).
\end{flalign*}
Note that
\begin{flalign*}
\bar{F}(\nu_{k^\star+j-1}) &= \bar{F}(v) - \int_{v}^{\nu_{k^\star+j-1}} f(s)ds\\&
\geq   \bar{F}(v)  -\bar{f} (\nu_{k^\star+j-1}-\nu_{k^\star+j})\\
&\geq   \bar{F}(v)  -\bar{f} (\nu_{k^\star+j-1}-v)\\
&\geq \bar{F}(v) - \frac{\eps^\gamma}{(2-m)^j}\\
&\geq \bar{F}(v) - \eps^\gamma \frac{1}{(2-m)^2} \left( \frac{(m-1)\ps}{\eps^\gamma} +1 \right),
\end{flalign*}
where in the last inequality we have used \cref{eq:rev-j-below-2}. In turn,
\begin{flalign*}
\Delta(v)&\leq  vf(v) - (2-m)^2 \eps^\gamma \left(
\frac{1}{m\eps^{\gamma}}\left\{\bar{F}(v) - \eps^\gamma \frac{1}{(2-m)^2} \left( \frac{(m-1)\ps}{\eps^\gamma} +1 \right) \right\}\right. \\
&\left. - V_0^{\bv}(f)\frac{1}{(2-m)} \left( \frac{(m-1)\ps}{\eps^\gamma} + 1 - (2-m)\eps^{\beta-\gamma}\right)
\right)\\
&=  vf(v) - (2-m)^2\left(
\frac{1}{m}\left\{\bar{F}(v) - \eps^\gamma \frac{1}{(2-m)^2} \left( \frac{(m-1)\ps}{\eps^\gamma} +1 \right) \right\}\right. \\
&\left. - V_0^{\bv}(f)\frac{1}{(2-m)} \left( (m-1)\ps + \eps^\gamma - (2-m)\eps^{\beta}\right)
\right)\\
&=  vf(v) - \frac{(2-m)^2}{m}\bar{F}(v) + \frac{1}{m} \left( (m-1)\ps+\eps^\gamma \right) + V_0^{\bv}(f) (2-m) \left( (m-1)\ps + \eps^\gamma - (2-m)\eps^{\beta}\right)\\
&\leq vf(v) -\bar{F}(v)- \left( \frac{(2-m)^2}{m}-1\right) + \left(1  + V_0^{\bv}(f) \right)\left( (m-1)\ps+\eps^\gamma \right)\\
&\leq vf(v) -\bar{F}(v)+ \left(1- \frac{(2-m)^2}{m}\right) + \left(1  + V_0^{\bv}(f) \right)\left( (m-1)\ps+\eps^\beta \right),
\end{flalign*}
we note that the third term above is  positive for $\eps>0$ small enough and of order $\eps^\beta$. The fourth term is also of oder $\eps^\beta$ because $m-1$ is of that order. In conclusion, we proved  that the following holds almost everywhere in $[0,\ps]$
$$
\Delta(v)\leq -\dot{R}(v) +C\cdot \eps^\beta,
$$
for some positive constant $C$. From the last assumption in the statement of the theorem, we can deduce that for $\eps>0$ small enough $-\dot{R}(v) +C\cdot \eps^\beta<0$
for all $v\leq \ps +\ell$. From \Cref{def1}, for $v\in (\ps-\ell,\ps)$ we have that $-\dot{R}(v) +C\cdot \eps^\beta< 0$ if
$$
C\cdot \eps^\beta \leq \kappa_L \alpha (\ps-v)^{\alpha-1} \Leftrightarrow
v\leq  \ps-\left(\frac{C\cdot \eps^\beta}{\kappa_L \alpha}\right)^{\frac{1}{\alpha-1}} \triangleq p_L.
$$
Hence, $\eps>0$ small enough we have
$$
\int_{0}^{\ps} \Delta^+(v) dv = \int_{p_L}^{\ps} \Delta^+(v) dv \leq \int_{p_L}^{\ps} \left(-\dot{R}(v) +C\cdot \eps^\beta \right)^+ dv \leq C\eps^\beta \left(\frac{C\cdot \eps^\beta}{\kappa_L \alpha}\right)^{\frac{1}{\alpha-1}} = \mathcal{O}\left(\eps^{\beta+\frac{\beta}{\alpha-1}}\right),
$$
where we have used that $-\dot{R}(v)\leq 0$ for $v\leq \ps$ (this is implied by the assumptions in the theorem). To conclude note that
$\beta + \frac{\beta}{\alpha-1}= \frac{\alpha}{2\alpha-1}$ if we take $\beta = (\alpha-1)/(2\alpha-1)$. This gives us the right order and concludes the proof.

\end{proof}

\section{Proofs for Section \ref{sec:floor}}

\begin{proof}{\underline{\bfseries\sffamily Proof of \Cref{prop:hard/soft-floor}}}
We first establish that the optimal mechanism in $\md\in\mathcal{M}_\texttt{d}(\varepsilon)$ is hard/soft mechanism. We then analyze
the quantity $ \Pi(\md)  \: - \: \Pi^\star(\mathcal{M}(0))\:$, and prove parts \textit{i.)} and \textit{ii.)}.

Let $\md=(\xd,\td)$ be the optimal deterministic mechanism. We show that $\md$ is a hard/soft floor mechanism, that is,
\begin{eqnarray*}
\xd(v) &=& \mathbf{1}\{v \in [\pd,\bv] \},\\
\td(v) &=& v \: \mathbf{1}\{v \in [\pd,\pt)\} + s \:\mathbf{1}\{v \in [\pt,\bv]\},
\end{eqnarray*}
for some appropriately chosen $(\pd,\pt)$.

Without loss of generality assume that $\xd(\cdot)$ is right-continuous (this is possible because $F$ is absolutely continuous). Define
\begin{equation*}
r\triangleq \inf\{v\in[0,\bv]:\xd(v)=1\}.
\end{equation*}
Note that $r<\infty$. Indeed, if this is not true then $\xd(\cdot)\equiv 0$ which, together with \eqref{eq:IR}, implies that $\td(v)\leq0$
for all $v\in[0,\bv]$. This gives $\Pi(\md)\leq 0$ which is suboptimal because the associated mechanism to the posted price $\ps$
 is feasible in $\mathcal{M}_\texttt{d}(\varepsilon)$ and $\Pi^\star(\mathcal{M}(0))>0$. Moreover, note that because $\xd(\cdot)$ is right-continuous we have that
 $\xd(r)=1$, and $\xd(v)=0$ for all $v<r$.

Next we show that
\begin{equation}\label{eq:tm1-t-u}
\td(v)\leq \threepartdef{0}{v\in[0,r);}{v\cdot \xd(v)}{v\in[r,r+\eps);}{r+\eps}{v\in[r+\eps,\bv].}
\end{equation}
For $v<r$, \eqref{eq:IR} and the definition of $r$ imply that $\td(v)\leq v\cdot \xd(v)= 0$. For $v\in [r,r+\eps)$ the inequality follows from \eqref{eq:IR}. For $v\in[r+\eps,\bv]$, \eqref{eq:IC} yields
\begin{equation*}
\td(v) \leq \td(r) +v\cdot(\xd(v)-\xd(r))+\eps\leq r+\eps,
\end{equation*}
where the second inequality follows from $ \td(r)\leq r\cdot \xd(r)=r$ and $(\xd(v)-\xd(r))\leq 0$.

Since $(\xd,\td)$ is optimal, all the inequalities in  \cref{eq:tm1-t-u} must bind. Moreover, we must have  $\xd(v)=1$ for $v\in[r,r+\eps)$.
This means that $\md$ is a hard/soft floor mechanism with $\pd=r$ and  $\pt=r+\eps$. Note that this mechanism is also feasible. \
Indeed, the mechanism satisfies \eqref{eq:IR} because the transfers are never larger than buyers valuations. It also satisfies \eqref{eq:IC}.
A best reporting function is
\begin{equation*}
\vd(v) = \twopartdef{0}{v\in[0,r);}{r}{v\in[r,\bv].}
\end{equation*}
With this, we can see that $u(v)= (v-r) \mathbf{1}\{v \in [r,\bv] \}$. Using the characterization in \cref{eq:IC-no-tight} of \eqref{eq:IC}  and the definition of $\md$, it is not hard to check that $\md$ verifies \eqref{eq:IC}.
This shows that among the deterministic mechanisms, the hard/soft floor mechanisms are optimal. In turn, the optimal solution can be find by optimizing over $r$.

Next we analyze
the quantity $ \Pi(\md)  \: - \: \Pi^\star(\mathcal{M}(0))\:$. We use $\md(r)$ instead of $\md$ to highlight the dependence on $r$.
We have
\begin{equation*}
\Pi(\md(r))= \int_{r}^{r+\eps}vf(v)dv + (r+\eps)\cdot\overline{F}(r+\eps) =
r\cdot  \overline{F}(r) +\int_{r}^{r+\eps}\overline{F}(v)dv,
\end{equation*}
where in the second inequality we used integration by parts. Hence,
\begin{equation*}
\Pi(\md(r)) =r\cdot  \overline{F}(r) +\int_{r}^{r+\eps}\overline{F}(v)dv\: \leq\: r\cdot  \overline{F}(r) +\eps  \: \leq\: \Pi^\star(\mathcal{M}(0))+\eps,
\end{equation*}
where in the first inequality we used that $\overline{F}(v)\leq 1 $ and in the second we used that
$\Pi^\star(\mathcal{M}(0))=\ps \cdot \overline{F}(\ps)$. This shows that $ \Pi(\md)  \: - \: \Pi^\star(\mathcal{M}(0))\:$
is $\mathcal{O}(\eps)$. %

 We next show a lower bound on performance. By \cref{assumption1},  $\ps \in(0,\bv)$. Hence we can always find $\eps_0>0$ such that
for all $\eps\leq \eps_0$, $\ps+\eps\in (0,\bv)$. Now, by the mean value theorem, we have
that $\overline{F}(\ps+\eps)-\overline{F}(\ps)=-f(\xi(\eps))\cdot \eps$ for $\xi(\eps)\in[\ps,\ps+\eps]$. Therefore,
\begin{flalign*}
\Pi^\star(\mathcal{M}_\texttt{d}(\eps))&\geq \Pi(\md(\ps))\\
&=\Pi^\star(\mathcal{M}(0))+\int_{\ps}^{\ps+\eps}\overline{F}(v)dv\\
&=\Pi^\star(\mathcal{M}(0))+\overline{F}(\ps+\eps)\cdot \eps\\
&\ge \Pi^\star(\mathcal{M}(0))+ \overline{F}(\ps)\cdot \eps-\overline{f}\cdot\eps^2,
\end{flalign*}
where in the last inequality we used that $f(v)\leq \overline{f}$ for all $v\in[0,\bv]$ (\Cref{assumption1}).
 This completes the proof of the theorem.

\end{proof}

\section{Proofs for Section \ref{sec:low-bound-perf}}\label{app:sec4}

\begin{proof}{\underline{\bfseries\sffamily Proof of \Cref{prop:alloc-tra-vs}}}
For ease of notation we use $(x,t)$ in lieu of $(\xr,\tr)$.
Let $u(v)=\max_{w\in \mathcal{S}} \{v\cdot x(w)-t(w)\}$, we want to show that
\begin{equation*}%
 u(v)= v\cdot \all(\vsr(v))-\tra(\vsr(v)).
\end{equation*}

First observe that the generalize inverse of $\vsr(\cdot)$, $\w(\cdot)$, is
\begin{equation*}
\w(v)= \twopartdef{\ps-\delta}{v\in[0,\ps-\delta-\mu];}{v+\mu}{v\in[\ps-\delta-\mu, \ps+\delta].}
\end{equation*}
In turn, we have that   $w(v)\ge v$ for all $v\in[0,\ps+\delta]$. Let us analyze two cases, $v\le \vz$ (case 1) and $v\geq \vz$ (case 2).

\textbf{Case 1. } Suppose that $v\le \vz$. Define the function
\begin{equation*}
g(v')= v\cdot \all(v')-\tra(v'),\quad  v'\in[0,\bv],
\end{equation*}
We need to show that $g(\vsr(v))=g(0)\geq g(v')$ for all $v'\in[0,\bv]$. To prove this,
we argue that $g(\cdot)$ is non-increasing. If $v'\in[0,\vz]$ we have
\begin{equation*}
g(v') = (v-v')\cdot \all(v'),\quad \text{and}\quad \dot{g}(v')=-\all(v')+(v-v')\cdot \dot{\all}(v').
\end{equation*}
In $[0,\vz]$ the allocation $\all(\cdot)$ satisfies the differential equation
$(w(v')-v')\dot{\all}(v')=\all(v')$. In turn,
\begin{equation*}
 \dot{g}(v')=-(w(v')-v')\dot{\all}(v')+(v-v')\cdot \dot{\all}(v')=(v-w(v'))\cdot \dot{\all}(v').
\end{equation*}
Note that $w(\cdot)$ is  non-decreasing and $v\leq \vz=w(0)$ hence we have that
$v\leq w(v')$ for all $v'\in [0,\vz]$. Moreover, $\all$ is non-decreasing (see \Cref{prop:x-mon} in \Cref{sec:aux-sec4}) which implies that
$\dot{g}(v')\leq 0$ for all $v'\in[0,\vz]$. We have thus verified that $g$ is decreasing in $[0,\vz]$.

If $v'\in [\vz,\vm]$ we have that
\begin{equation*}
g(v') = (v-v')\cdot \all(v')+\int_0^{v'}\all(\vsr(s))ds-\epsilon,\quad \text{and}\quad \dot{g}(v')=-\all(v')+(v-v')\cdot \dot{\all}(v') +\all(\vsr(v')).
\end{equation*}
In this interval the allocation $\all$ is a solution to the delayed differential equation $\all(v') + (v' - w(v') ) \dot{\all}(v') =\all (\vsr(v'))$. As before this implies that $ \dot{g}(v')=(v-w(v'))\cdot \dot{\all}(v').$ Note that because
$\all$ is non-decreasing (see \Cref{prop:x-mon} in \Cref{sec:aux-sec4}) we have that $\dot{\all}(v')\ge 0$, and because $w(v')\ge v'$ and $v'\geq \vz\ge v$, we can conclude that $ \dot{g}(v')\le 0$. We have thus verified that $g$ is non-increasing in $[\vz,\vm]$.

If $v'\in[\vm,\bv]$ we have
\begin{equation*}
g(v') = (v-v')+\int_0^{v'}\all(\vsr(s))ds-\epsilon,\quad \text{and}\quad \dot{g}(v')=-1+\all(\vsr(v')),
\end{equation*}
where to derive $\dot{g}(v')$ we have used that $\all(v')=1$ for all $v'\in[\vm,\bv]$.
Note that since $\vsr(\cdot)$ is non-decreasing, we have that $\vsr(v')\leq \vm$ for all $v'$. In turn, because
$\all(\cdot)$ is non-decreasing (see \Cref{prop:x-mon} in \Cref{sec:aux-sec4}) we conclude that $\all(\vsr(v'))\leq \all(\vm)=1$. That is,
$ \dot{g}(v')\leq 0$, equivalently, $g$ is decreasing in $[\vm,\bv]$.

 Finally, since $\all$ and $\tra$ are continuous functions, $g$ is also a continuous function. This, together with $g$ being decreasing in each of the intervals $[0,\vz]$, $[\vz,\vm]$ and $[\vm,\bv]$, implies that $g$ is decreasing in $[0,\bv]$, as desired.

\textbf{Case 2.}
$v\ge \vz$: In the previous case we already computed the derivative of $g(\cdot)$ at different intervals. In particular, we have
\begin{equation*}
\dot{g}(v')=\twopartdef{(v-w(v'))\cdot \dot{\all}(v')}{v'\leq \ps+\delta}{-1+\all(\vsr(v'))}{v'\geq \ps+\delta}
\end{equation*}
Note that the sign of the derivative at $v'\in [0,\vm]$ depends on
how $v$ compares to $w(v')$. Observe that $w(\cdot)$ is a non-decreasing continuous function taking values in $[\vz,\vm+\mu]$. Therefore, for any $v\in [\vz,\vm+\mu]$
we have that if $v'<\vsr(v)$ then $w(v')\le v$, and if $v'>\vsr(v)$ then $w(v')\ge v$. Moreover, as we argued before
$x(\vsr(v'))\leq 1$.
That is, for any $v\in [\vz,\vm+\mu]$ the function $g(v')$ is first increasing for $v'<\vsr(v)$
and then decreasing for $v'>\vsr(v)$. Hence, for any $v\in [\vz,\vm+\mu]$  the maximum of $g$ is attained at $\vsr(v)$.

If $v\geq \vm+\mu$ then $\vsr(v)=\vm$, and for any $v'\leq \ps+\delta$ we have that
$\w(v')\leq \ps+\delta+\mu\leq v$. In turn, in this case $g(v')$ is increasing for all $v'\leq \ps+\delta$.
Since  we already saw that $g$ is  decreasing for $v'\ge \vm$ we conclude that $\vsr(v)$ attains the maximum of $g$. Thus the maximum is still achieved at $\vsr(v)$, as desired.

\end{proof}

\begin{proof}{\underline{\bfseries\sffamily Proof of \Cref{prop:well-posed}}}
Consider the change of variables
\begin{equation*}
\yr(v)=\xr(v)\cdot \frac{\mu}{\eps}\cdot \exp\left(-\frac{v-\ps+\delta+\mu}{\mu}\right).
\end{equation*}
To prove the statement,  we solve the resulting ordinary and delayed differential equations for $\yr(v)$.
Under this change of variables, \eqref{eq:ode-1}  for $v\in [0,\ps-\delta-\mu]$ becomes
\begin{equation*}%
\dyr(v)=\yr(v)\left(\frac{1}{w(v)-v}-\frac{1}{\mu}\right), \quad \yr(0)=\frac{\mu}{\vz}\cdot e^{\frac{\vz-\mu}{\mu}},\quad v\in[0,\vz-\mu].
\end{equation*}
This has solution,
\begin{equation*}
    \yr(v)=\frac{\mu}{\vz-v}\exp\left(-\frac{v-\ps+\delta+\mu}{\mu}\right).
\end{equation*}

For $v\in [\ps-\delta-\mu,\ps-\delta]$
\eqref{eq:ode-1} becomes $\dyr(v)=0$ with
$\yr(\ps-\delta-\mu)=1$. Hence, in this interval $\yr(v)=1$.

 \eqref{eq:dde-1} becomes
\begin{equation}\label{eq:dde-y-p2}
\dyr(v)+\frac{e^{-1}}{\mu} \cdot \yr(v-\mu)=0,\quad v\in [\ps-\delta,\ps+\delta]; \quad \yr(v)=1, \:\: \forall v\in [\ps-\delta-\mu,\ps-\delta].
\end{equation}
The delayed differential equation in \eqref{eq:dde-y-p2}
is a linear delayed differential equation with constants coefficients. Its solution is unique and it can be found in
\cite{norkin1973introduction} p.8 and is given by
\begin{equation*}
\yr(v) = \sum_{j=0}^{\floor*{\frac{v-\ps+\delta}{\mu}} + 1} \frac{(-1)^je^{-j}}{j!}\left(\frac{v-\ps+\delta}{\mu}+1-j\right)^j,\quad v\in [\vz,\vm].
\end{equation*}
This solution can also be verified by an inductive argument. In summary, we have that
\begin{equation}\label{eq:solyr}
\yr(v)=
\threepartdef{\frac{\mu}{\vz-v}\exp\left(-\frac{v-\ps+\delta+\mu}{\mu}\right)}{v\leq \vz-\mu;}{1}{v\in [\vz-\mu,\vz];}{\sum_{j=0}^{\floor*{\frac{v-\ps+\delta}{\mu}} + 1} \frac{(-1)^je^{-j}}{j!}\left(\frac{v-\ps+\delta}{\mu}+1-j\right)^j}{ v\in [\vz,\vm].}
\end{equation}
Finally, we argue that $\xr$ is strictly positive and  monotone non-decreasing in $[0,\ps + \delta]$.
 \Cref{prop:x-mon} in \Cref{sec:aux-sec4} establishes that $\xr$ is non-decreasing in $[0,\ps+\delta]$. Since $\xr(0)=\eps/(\ps-\delta)>0$ and $\xr(\cdot)$ is non-decreasing, we deduce that  $\xr(\cdot)$ is strictly positive.
\end{proof}

\begin{proof}{\underline{\bfseries\sffamily Proof of \Cref{prop:feas}}}
First fix $\eps>0$ and $\mu>\eps\cdot e$. To show that $\lim_{v\uparrow\ps+\delta}\xr(v)= 1$ it suffices to prove
$\xr(\vm)$ can be equal to one. Given the change of variables from $\xr(v)$ to $\yr(v)$ (see the proof of \Cref{prop:well-posed}), the latter condition is equivalent to:
\begin{equation}\label{eq:1-pf-prop:sc-rel}
\yr(\vm)=\frac{1}{\eps}\frac{\mu}{e}e^{\frac{-2\delta}{\mu}}.
\end{equation}
Note that from \cref{eq:solyr}, we have that
\begin{equation*}
\yr(\vm)=\sum_{j=0}^{\floor*{\frac{2\delta}{\mu}} + 1} \frac{(-1)^je^{-j}}{j!}\left(\frac{2\delta}{\mu}+1-j\right)^j\triangleq \Gamma\left(\frac{2\delta}{\mu}\right).
\end{equation*}
Then we can cast \eqref{eq:1-pf-prop:sc-rel} as
\begin{equation}\label{eq:recast-1-aa}
\Gamma\left(\frac{2\delta}{\mu}\right)=\frac{\mu}{e\eps}e^{\frac{-2\delta}{\mu}}.
\end{equation}
For fixed $\eps$ and $\mu$ we show that we can find a solution $\delta(\mu,\eps)$ to \Cref{eq:recast-1-aa}. { Let $z = \frac{2\delta}{\mu}$, then we need to find $z$ that solves $\Gamma(z)=\frac{\mu}{e\eps}e^{-z}$ or equivalently $L(z) = 0$ with $L(z) = \Gamma(z) - \frac{\mu}{e\eps}e^{-z}$. Note that for $z=0$, we have $L(0) = 1- \frac{\mu}{e\eps} < 0$ because $\Gamma(0)=1$ and, by assumption, $\frac{\mu}{e\eps}>1$. Additionally, from \Cref{lem:new-rev-bound-on-gamma} we have that $L(z) \ge \left( 2e^{-1} z - \frac{\mu}{e\eps} \right) e^{-z}$ and, thus, $L(z) > 0$ for all $z$ large enough. Finally, we have that $L(z)$ is continuous because $\Gamma(z) = \yr(\mu z +\ps -\delta)$ and $\yr(v)$ is the continuous solution of a delayed differential.
The intermediate value theorem implies that we can always solve \Cref{eq:recast-1-aa}.} From the bounds \Cref{lem:new-rev-bound-on-gamma} we deduce that a solution $\delta= \delta(\mu,\eps)$ to \Cref{eq:1-pf-prop:sc-rel} must satisfy
$$
\frac{\mu}{e\eps}e^{\frac{-2\delta}{\mu}}= \Gamma\left(\frac{2\delta}{\mu}\right) >2\left(\frac{2\delta}{\mu}\right) e^{-1}e^{-\frac{2\delta}{\mu}} \Leftrightarrow
\frac{\mu^2}{4\eps}>  \delta,
$$
and, also,
$$
e^{\frac{2\delta}{\mu}+1} \Gamma\left(\frac{2\delta}{\mu}\right) \leq 2\left(\frac{2\delta}{\mu}+2\right) \Leftrightarrow
e^{\frac{2\delta}{\mu}+1} \left(\frac{\mu}{e\eps}e^{\frac{-2\delta}{\mu}}\right) \leq 2\left(\frac{2\delta}{\mu}+2\right)\Leftrightarrow
\frac{\mu}{4\eps} \leq \frac{\delta}{\mu}+1.
$$
This shows that $\delta(\mu,\eps) \in \left[ \frac{\mu^2}{4\eps}-\mu, \frac{ \mu^2}{4\eps}\right]$.

 Next, we show that $\mr\in \mathcal{M}(\eps)$. First, note by assumption  $\vsr$ is well defined.
  Now from  \Cref{prop:alloc-tra-vs} and since $\delta(\mu(\eps),\eps)$ solves \eqref{eq:1-pf-prop:sc-rel}, we know that $\vsr(\cdot)$
 is a best reporting function for $\mr$. Given this property, we can verify that $\mr$ is feasible. To see why \eqref{eq:IC} holds,
note that because $\vsr(\cdot)$ is a best reporting function we have that $\max_{w}\{v\cdot \xr(w)-\tr(w)\}
=v\cdot \xr(\vsr(v))-\tr(\vsr(v))$. This, together with the
envelope theorem (see \cref{eq:env-report}), implies that \eqref{eq:IC} can be cast as
\begin{equation*}
v\cdot \xr(v)-\tr(v)\geq \int_{0}^v \xr(\vsr(s))ds-\epsilon,\quad \forall v\in [0,\bv].
\end{equation*}
If $v\leq \ps-\delta$ then the left hand-side above is 0; while the right hand-side is $v\xr(0)-\eps$ which, by the boundary condition of $\xr(\cdot)$, is non-negative. If $v\geq \ps-\delta$ then the inequality above is binding.
That is, $\mr$ verifies \eqref{eq:IC}. The \eqref{eq:IR} constraint can be  verified in a similar fashion. This concludes the proof.

\end{proof}

\begin{proof}{\underline{\bfseries\sffamily Proof of \Cref{thm:low-bd-per}}}
We rely on the characterization for $\Pi(\mr)$ provided in \Cref{prop:low-bd-obj}.
We analyze the order of $\Pi(\mr)$.  We have
\begin{flalign*}
\Pi(\mr) &=
\int_{0}^{\vm}\all(v)\left(v-\frac{\overline{F}(v+\mu)}{f(v)}\dot{w}(v)\right)f(v)dv+R(\vm)+\int_{\vm}^{\vm+\mu}\overline{F}(v)dv\\
&=  \int_{0}^{\ps-\delta-\mu}\all(v)\left(v-\frac{\overline{F}(v+\mu)}{f(v)}\dot{w}(v)\right)f(v)dv + \int_{\ps-\delta-\mu}^{\vm}\all(v)\left(v-\frac{\overline{F}(v+\mu)}{f(v)}\dot{w}(v)\right)f(v)dv\\
&+R(\vm)+\int_{\vm}^{\vm+\mu}\overline{F}(v)dv\\
&=  \int_{0}^{\ps-\delta-\mu}\all(v)vf(v)dv + \int_{\ps-\delta-\mu}^{\vm}\all(v)\left(v-\frac{\overline{F}(v+\mu)}{f(v)}\right)f(v)dv
+R(\vm)+\int_{\vm}^{\vm+\mu}\overline{F}(v)dv\\
&\ge \underbrace{\int_{\ps-\delta-\mu}^{\ps+\delta} vx(v) f(v)dv- \int_{\ps-\delta-\mu}^{\ps+\delta}x(v)\bar{F}(v+\mu)dv}_{(A)}
+\underbrace{R(\ps+\delta)+\int_{\ps+\delta}^{\ps+\delta+\mu}\bar{F}(v)dv}_{(B)}.
\end{flalign*}
where we have used that $\dot{w}(v)=0$ for $v\leq \ps-\delta-\mu$ and $\dot{w}(v)$ in $[\ps-\delta-\mu,\vm]$.
Let us bound $(B)$ first.  We use \Cref{prop:feas}  with
$$
\mu(\eps)= K\cdot \eps^\beta,\quad \beta\in(0,1).
$$
Note that we need $\beta\in(0,1)$ so that $\mu >\eps \cdot e$ for all $\eps>0$ small and $\mu(\eps)/\eps \uparrow \infty$ as $\eps\downarrow0$. With this choice, we have
$$
\delta \leq \frac{ K^2}{4} \cdot \eps^{2\beta-1}.
$$
Note that the above further constraints $\beta$ such that $\beta>1/2$.
We have
\begin{flalign*}
(B)&\geq \bar{F}(\ps+\delta+\mu)\mu +R(\ps +\delta)\\
&\geq \bar{F}(\ps+\delta+\mu)\mu +\rs-\kappa_U\delta^\alpha\\
&\geq \bar{F}(\ps+\delta+\mu)\cdot K\cdot \eps^\beta +\rs-\kappa_U \left(\frac{K^2}{4}\right)^\alpha \eps^{(2\beta-1)\alpha}
\end{flalign*}
Note that in the second inequality above we are using that
\begin{equation}
R(\ps +\delta)-\rs = \int_{\ps}^{\ps+\delta}\dot{R}(v)dv \geq \int_{\ps}^{\ps+\delta} -\kappa_u\alpha(v-\ps)^{\alpha-1} dv =-\kappa_U\delta^\alpha,
\end{equation}
here we are using \Cref{def1} and $\delta>0$ small such that $\delta<\ell$. The latter is possible by considering $\eps>0$ small enough. We set $\beta=\alpha/(2\alpha-1)$ then
\begin{flalign*}
(B)&\geq \rs+ \bar{F}(\ps+\delta+\mu)\cdot K\cdot \eps^{\alpha/(2\alpha-1)} -\kappa_U \left(\frac{K^2}{4}\right)^\alpha \eps^{\alpha/(2\alpha-1)}\\
&= \rs+  \left(\bar{F}(\ps+\delta+\mu)\cdot K -\kappa_U \left(\frac{K^2}{4}\right)^\alpha \right)\cdot \eps^{\alpha/(2\alpha-1)}
\end{flalign*}
We must make sure that the parenthesis in the last expression above is strictly positive:
$$
\bar{F}(\ps+\delta+\mu) >\kappa_U\left(\frac{1}{4}\right)^\alpha K^{2\alpha-1},
$$
since $\alpha>1/2$, $\delta(\eps)+\mu(\eps)\downarrow 0$, and $\bar{F}(\ps)>0$,  we can always choose $K>0$ small enough such that the above is true.

Let us bound $(A)$. To do this, we use \Cref{def1} again.

\begin{flalign*}
(A) &= \int_{\ps-\delta-\mu}^{\ps+\delta} vx(v) f(v)dv- \int_{\ps-\delta-\mu}^{\ps+\delta}x(v)\bar{F}(v+\mu)dv\\
&\geq \int_{\ps-\delta-\mu}^{\ps+\delta} vx(v) f(v)dv- \int_{\ps-\delta-\mu}^{\ps+\delta}x(v)\bar{F}(v)dv\\
&= \int_{\ps-\delta-\mu}^{\ps+\delta} x(v)\left(v f(v)-\bar{F}(v)\right)dv\\
&= \int_{\ps-\delta-\mu}^{\ps+\delta} x(v)\left(-\dot{R}(v)\right)dv\\
&= -x(\ps-\delta)R(\ps-\delta) + x(\ps-\delta-\mu)R(\ps-\delta-\mu) + \int_{\ps-\delta-\mu}^{\ps+\delta} \dot{x}(v)R(v)dv\\
&= -R(\ps-\delta) + \frac{\eps}{\mu}R(\ps-\delta-\mu) + \int_{\ps-\delta-\mu}^{\ps+\delta} \dot{x}(v)R(v)dv\\
&\geq -\left(\rs -\kappa_L \delta^\alpha\right) + \frac{\eps}{\mu}\left(\rs -\kappa_U |\delta+\mu|^\alpha\right) + \int_{\ps-\delta-\mu}^{\ps+\delta} \dot{x}(v)\left(\rs -\kappa_U |v-\ps|^\alpha\right)dv\\
&\geq -\left(\rs -\kappa_L \delta^\alpha\right) + \frac{\eps}{\mu}\left(\rs -\kappa_U (\delta+\mu)^\alpha\right) + \left(\rs -\kappa_U (\delta+\mu)^\alpha\right)\int_{\ps-\delta-\mu}^{\ps+\delta} \dot{x}(v)dv\\
&= -\left(\rs -\kappa_L \delta^\alpha\right) + \frac{\eps}{\mu}\left(\rs -\kappa_U (\delta+\mu)^\alpha\right) + \left(\rs -\kappa_U (\delta+\mu)^\alpha\right)\left(1-\frac{\eps}{\mu}\right)\\
& = \kappa_L\delta^\alpha-\kappa_U (\delta+\mu)^\alpha\\
&\geq -\kappa_U (\delta+\mu)^\alpha\\
&\stackrel{(a)}{\geq} -\kappa_U (2\delta)^\alpha\\
&\geq -\kappa_U \left(2 \frac{K^2}{4} \cdot \eps^{2\beta-1}\right)^\alpha\\
&\geq -\kappa_U \left(2 \frac{K^2}{4} \cdot \right)^\alpha \eps^{\alpha/(2\alpha-1)}
\end{flalign*}
In (a) we use that $\delta>\mu$ when $\lim_{\eps\downarrow0}\mu(\eps)/\eps=+\infty$ because
$\lim_{\eps\downarrow 0}\delta(\mu(\eps),\eps)/\mu(\eps)=+\infty$ which is true because $\frac{\mu}{4\eps} \leq \frac{\delta}{\mu}+1 $, see
\Cref{prop:feas}. This, together with the bound we have for (B), yields
$$
\Pi \geq
 \rs+  \left(\bar{F}(\ps+\delta+\mu)\cdot K -\kappa_U \left(\frac{K^2}{4}\right)^\alpha -\kappa_U \left(2 \frac{K^2}{4} \cdot \right)^\alpha\right)\cdot \eps^{\alpha/(2\alpha-1)}.
$$
The parenthesis can be guaranteed to be positive by using a the same reasoning as before. This concludes the proof.

\end{proof}

\subsection{Auxiliary results for Section \ref{sec:low-bound-perf}}\label{sec:aux-sec4}
\begin{lemma}\label{prop:x-mon}
Let $\all(\cdot)$ be a solution to \eqref{eq:ode-1}-\eqref{eq:dde-1}. Then $\all(\cdot)$ is non-decreasing in $[0,\ps+\delta]$.
\end{lemma}
\begin{proof}{\underline{\bfseries\sffamily Proof of \Cref{prop:x-mon}}}

To see why $\all(\cdot)$ is non-decreasing we show that if $\all(\cdot)$ is non-decreasing in $[\vz+(n-1)\cdot \mu,\vz+n\cdot \mu]$ then it is also non-decreasing in $[\vz+n\cdot \mu,\vz+(n+1)\cdot \mu]$ for all $n\ge 0$. Indeed, fix $n$ and  denote by $\all_0$ to $\all(\cdot)$ in
$[\vz+(n-1)\cdot \mu,\vz+n\cdot \mu]$. Suppose that $\all_0$ in non-decreasing.
 Now consider $v\in[\vz+n\cdot \mu,\vz+(n+1)\cdot \mu]$, $\all(\cdot)$ must satisfy the following equation
\begin{equation}\label{eq:altern-ode-0}
\all(v)=\all_0(\vz+n\cdot \mu),\quad \dot{\all}(v)= \frac{\all(v)-\all_0(v-\mu)}{\mu},\quad \forall v\in [\vz+n\cdot \mu,\vz+(n+1)\cdot \mu].
\end{equation}
Consider the alternative ordinary differential equation
\begin{equation}\label{eq:altern-ode-1}
r(v)=\all_0(\vz+n\cdot \mu),\quad \dot{r}(v)= \frac{r(v)-\all_0(\vz+n\cdot \mu)}{\mu},\quad \forall v\in [\vz+n\cdot \mu,\vz+(n+1)\cdot \mu].
\end{equation}
Let $f(v,x)$ and $g(v,r)$ be the right-hand side of \cref{eq:altern-ode-0} and \cref{eq:altern-ode-1}, respectively. Note because $\all_0(\cdot)$ is non-decresing in $[\vz+(n-1)\cdot \mu,\vz+n\cdot \mu]$ we have that
$f(v,x)\geq g(v,x)$. In turn, a standard comparison argument for ordinary differential equations implies that
$\all(v)\geq r(v)$ for all $v\in [\vz +n\cdot \mu,\vz+(n+1)\cdot \mu].$ Since the solution to \cref{eq:altern-ode-1} is $r(v)=\all_0(\vz+n\cdot \mu)$, we can conclude that $\all(v)\ge\all_0(\vz+n\cdot \mu)$ for all $v\in[\vz +n\cdot \mu,\vz+(n+1)\cdot \mu].$ Therefore, from
\cref{eq:altern-ode-0} we deduce that $\mu\dot{\all}(v)\geq \all_0(\vz+n\cdot \mu)-\all_0(v-\mu)\geq 0$
 for all $v\in[\vz +n\cdot \mu,\vz+(n+1)\cdot \mu].$ That is, $\all(\cdot)$ is non-decreasing in $\in [\vz +n\cdot \mu,\vz+(n+1)\cdot \mu]$.
 To conclude, note that in $[\vz-\mu,\vz]$ $\all(\cdot)$ is the solution to an ordinary differential equation for which it can be easily verified that the solution in strictly increasing.
\end{proof}

\begin{lemma}\label{lem:new-rev-bound-on-gamma}
Let $\Gamma(t)$ be defined by
$$
\Gamma(t) \triangleq \sum_{j=0}^{\floor*{t} + 1} \frac{(-1)^je^{-j}}{j!}\left(t+1-j\right)^j,\quad t\geq 0.
$$
Then,
\begin{equation*}
2(t+1) e^{-1}e^{-t}< \Gamma(t) \leq 2(t+2)e^{-(t+1)}, \quad \forall t\geq 0.
\end{equation*}

\end{lemma}
\begin{proof}{\underline{\bfseries\sffamily Proof of \Cref{lem:new-rev-bound-on-gamma}}}
To prove the bounds we establish a connection between $\Gamma(x)$ and renewal theory. We will show that for an appropriate renewal function $m(t)$, we have that $m(t+1) = e^{t+1}\Gamma(t) - 1$. Then we will use natural bounds on $m(t)$ to derive bounds for $\Gamma(t)$.

Consider the uniform distribution with support in $[0,1]$ with cdf denoted by $G(t)$. That is, we have a sequence of independent random variables $\{X_i\}_{i\ge 1}$ with uniform distribution with support in $[0,1]$ that represent the time between the $(i-1)$st and the $i$th events. Let
 $$
 Z_0=0,\quad Z_k = \sum_{i=1}^k X_i,\quad k\ge1,
 $$
 that is, $Z_k$ is the time of the $k$th event. Then, the number of events up to time $t$ and the renewal function are given by
 $$
 N(t) = \sup\{k: \: Z_k\leq t\},\quad \text{and} \quad m(t) = \E[N(t)].
 $$
respectively. We will prove that  for all $n\geq 0$
\begin{equation}\label{eq:renewal_1}
m(t) = \sum_{k=0}^n (-1)^ke^{t-k}(t-k)^k/k! -1,\quad \text{for}\: n\leq t\leq n+1.
\end{equation}
The renewal function can be written as (see, e.g., Proposition 3.2.1 in \citealt{ross1996stochastic})
$$
m(t)= \sum_{\ell=1}^\infty \PP(Z_\ell\leq t).
$$
Because $Z_\ell$ is the sum of $\ell$ independent uniform distributions, it follows that $Z_\ell$ has an Irwin-Hall distribution given by
$$
 \PP(Z_\ell\leq t) = \twopartdef{\frac{1}{\ell!} \sum_{k=0}^{\floor{t}} (-1)^k \binom{\ell}{k} (t-k)^\ell}{t<\ell}{1}{t\geq\ell.}
$$
So for $n\leq t\leq n+1$,
\begin{flalign*}
m(t) &= \sum_{\ell=1}^\infty \PP(Z_\ell\leq t)\\
&= \sum_{\ell=1}^n 1 + \sum_{\ell=n+1}^\infty  \frac{1}{\ell!} \sum_{k=0}^{\floor{t}} (-1)^k \binom{\ell}{k} (t-k)^\ell\\
&= n + \sum_{\ell=n+1}^\infty  \frac{1}{\ell!} \sum_{k=0}^{n} (-1)^k \binom{\ell}{k} (t-k)^\ell\\
&= n + \sum_{k=0}^{n} \frac{(-1)^k}{k!}\sum_{\ell=n+1}^\infty  \frac{1}{(\ell-k)!}   (t-k)^\ell\\
&= n + \sum_{k=0}^{n} \frac{(-1)^k(t-k)^{k}}{k!}\sum_{\ell=n+1-k}^\infty  \frac{1}{\ell!}   (t-k)^{\ell}\\
&= n + \sum_{k=0}^{n} \frac{(-1)^k(t-k)^{k}}{k!}\left(e^{t-k}- \sum_{\ell=0}^{n-k}  \frac{1}{\ell!}   (t-k)^{\ell}\right)\\
&= n + \sum_{k=0}^{n} \frac{(-1)^k(t-k)^{k}}{k!}e^{t-k}
- \sum_{k=0}^{n} \sum_{\ell=0}^{n-k}  \frac{(-1)^k(t-k)^{k}}{k!} \frac{(t-k)^{\ell}}{\ell!} \\
&=  \sum_{k=0}^{n} \frac{(-1)^k(t-k)^{k}}{k!}e^{t-k} -1,
\end{flalign*}
where the second equation follows from the Irwin-Hall distribution, the third because $n\leq t\leq n+1$, the fourth from exchanging the order of summation, the sixth from the exponential series, and the last equality from the identity (which can be verified by induction)
$$
 \sum_{k=0}^{n} \sum_{\ell=0}^{n-k}  \frac{(-1)^k(t-k)^{k}}{k!} \frac{(t-k)^{\ell}}{\ell!}=n+1 ,\quad \forall n>0,t.
$$
This proves  \Cref{eq:renewal_1}. Evaluating the renewal function at $t+1$ gives
\begin{align}\label{eq:renewal_x}
m(t+1) &=  \sum_{k=0}^{\floor{t+1}} (-1)^ke^{t+1-k}(t+1-k)^k/k! -1
= \sum_{k=0}^{\floor{t}+1} (-1)^ke^{t+1-k}(t+1-k)^k/k! -1 \nonumber\\
&= e^{t+1} \Gamma(t)-1.
\end{align}
Now we obtain bounds on $m(t)$. Note that $Z_{N(t)+1}>t$ and that $N(t)+1$ is a stopping time, hence Wald's equation yields
$$
\frac{1}{2} (m(t)+1) >t.
$$
Hence, from \eqref{eq:renewal_x} using $t+1$ we have that
\begin{equation}\label{eq:renewal_x2}
\Gamma(t) >2(t+1) e^{-1}e^{-t}>2t e^{-1}e^{-t}
\end{equation}

Next, we prove the upper bound  for $\Gamma(t)$. Note that $Z_{N(t)+1}\leq t +1$ because the time between renewals are uniformly distributed in [0,1].  Hence,  by Wald's equation we have that
$$
\frac{1}{2}(m(t)+1) \leq t+1.
$$
Using $t+1$ instead of $t$ in the above we deduce that
$$
e^{t+1} \Gamma(t) \leq 2(t+2).
$$

\end{proof}

\begin{proof}{\underline{\bfseries\sffamily Proof of \Cref{prop:low-bd-obj}}}
We exploit the characterization of $\mr=(\xr,\tr)$ given in \cref{prop:alloc-tra-vs}. For ease of exposition we
use $(\all,\tra)$ instead of $(\xr,\tr)$.  We have
\begin{flalign*}
\Pi(\mr) &=\int_{0}^{\bv}\tr(v)f(v)dv\\
&=
\int_{0}^{\bv}v\all(v)f(v)dv+ \eps\overline{F}(\vz)-\int_{\vz}^{\bv}\int_{0}^{v}\all(\vsr(s))f(v)dsdv \\
&=\int_{0}^{\bv}v\all(v)f(v)dv+ \eps\overline{F}(\vz)
-\int_{0}^{\vz}\int_{\vz}^{\bv}\all(\vsr(s))f(v)dvds\\
&-\int_{\vz}^{\bv}\int_{s}^{\bv}\all(\vsr(s))f(v)dvds\\
&=\int_{0}^{\bv}v\all(v)f(v)dv+ \eps\overline{F}(\vz)
-(\vz)\all(0)\overline{F}(\vz)
-\int_{\vz}^{\bv}\all(\vsr(v))\overline{F}(v)dv\\
&\stackrel{(a)}{=}\int_{0}^{\bv}v\all(v)f(v)dv
-\int_{\vz}^{\vm+\mu}\all(v-\mu)\overline{F}(v)dv-\all(\vm)\int_{\vm+\mu}^{\bv}\overline{F}(v)dv\\
&=\int_{0}^{\bv}v\all(v)f(v)dv
-\int_{\vz-\mu}^{\vm}\all(v)\overline{F}(v+\mu)dv-\int_{\vm+\mu}^{\bv}\overline{F}(v)dv\\
&\stackrel{(b)}{=}\int_{0}^{\vm}\all(v)\left(v-\frac{\overline{F}(v+\mu)}{f(v)}\dot{w}(v)\right)f(v)dv+\int_{\vm}^{\bv}v\all(v)f(v)dv-\int_{\vm+\mu}^{\bv}\overline{F}(v)dv\\
&=\int_{0}^{\vm}\all(v)\left(v-\frac{\overline{F}(v+\mu)}{f(v)}\dot{w}(v)\right)f(v)dv
+\int_{\vm}^{\bv}(vf(v)-\overline{F}(v))dv +\int_{\vm}^{\vm+\mu}\overline{F}(v)dv\\
&\stackrel{(c)}{=}\int_{0}^{\vm}\all(v)\left(v-\frac{\overline{F}(v+\mu)}{f(v)}\dot{w}(v)\right)f(v)dv+R(\vm)+\int_{\vm}^{\vm+\mu}\overline{F}(v)dv,
\end{flalign*}
where $(a)$ holds because $\all(0)=\eps/(\ps-\delta)$;
$(b)$ holds because $\dot{\w}(v)=0$ if $v<\vz-\mu$ and $\dot{\w}(v)=1$ if $v\in(\vz-\mu,\vm)$; and $(c)$ holds because the derivative of $v\overline{F}(v)$ equals $\overline{F}(v)-vf(v)$. This concludes the proof.
\end{proof}

\section{Auxiliary Results }\label{sec:app_rev_aux-11}

\begin{lemma}\label{lem:app_rev_aux-11-1}
Suppose that $R(v)$ is locally concave around $\ps$ and let $\alpha \in (1,+\infty)$. Then
 there exists positive constants $\tilde{\kappa}_L,\tilde{\kappa}_U$ and a neighborhood of $\ps$, $\mathcal{N}_{\tilde{\ell}}=(\ps-\tilde{\ell},\ps+\tilde{\ell})\subset (0,\bv)$, such that
\begin{equation}\label{eq:def1-equiv-1}
\tilde{\kappa}_L\cdot |v - \ps|^\alpha \le R(\ps) - R(v) \le \tilde{\kappa}_U \cdot|v - \ps|^\alpha, \quad \forall v\in  \mathcal{N}_{\tilde{\ell}}
\end{equation}
if and only if there exists positive constants $\kappa_L,\kappa_U$ and a neighborhood of $\ps$, $\mathcal{N}_\ell=(\ps-\ell,\ps+\ell)\subset (0,\bv)$, such that
\begin{equation}\label{eq:def1-equiv-2}
\kappa_L\alpha\cdot |v-\ps|^\alpha \leq (\ps-v)\cdot \dot{R}(v)\leq \kappa_U\alpha\cdot |v-\ps|^\alpha, \quad \forall v\in \mathcal{N}_\ell.
\end{equation}

\end{lemma}

\begin{proof}{\underline{\bfseries\sffamily Proof of \Cref{lem:app_rev_aux-11-1}}}
We show that under local concavity \Cref{eq:def1-equiv-1} is equivalent to \Cref{eq:def1-equiv-2}. To prove that \Cref{eq:def1-equiv-2} implies \Cref{eq:def1-equiv-1} we can simply consider the cases $v<\ps$ and $v>\ps$ and then integrate \Cref{eq:def1-equiv-2}. The latter immediately delivers \Cref{eq:def1-equiv-1}. To show  that \Cref{eq:def1-equiv-1} implies \Cref{eq:def1-equiv-2}, we must use the local concavity assumption. Without loss of generality we will assume that all the calculation below are performed in the neighborhood where  $R(v)$ is  concave, $\mathcal{N}_c$, intersected with the neighborhood where \Cref{eq:def1-equiv-1} holds.
We have that
$$
\tilde{\kappa}_L\cdot |v-\ps|^\alpha\leq R(p)- R(v) =\int_{v}^p \dot{R}(s)ds \leq \dot{R}(v)(p-v),
$$
where the last inequality uses that $\dot{R}(v)$ is non-increasing by concavity of $R(v)$. Now we establish the upper bound in \Cref{eq:def1-equiv-2}. We consider the case $v<\ps$ (the case $v>\ps$ is symmetric). Since $R(v)$ is locally concave, we have that:
$$
R(x) \leq R(y) + \dot{R}(y) (x-y),\quad \forall x,y \in \mathcal{N}_c.
$$
Consider $z\leq v$, use the previous inequality with $x= z$ and $y = v$, and also use that $R(z) \geq R(\ps) - \tilde{\kappa}_U (\ps-z)^\alpha$ to obtain:
$$
\dot{R}(v) \leq  \frac{R(v)-R(z)}{v-z}\leq  \frac{R(v)-R(\ps) + \tilde{\kappa}_U (\ps-z)^\alpha}{v-z}\leq  \frac{\tilde{\kappa}_U (\ps-z)^\alpha}{v-z},
$$
where in the last inequality we used that $R(v)\leq R(\ps)$. Consider $z^\star = (\alpha v -\ps)/(\alpha-1)$ (which can be obtained by optimizing the right-hand-side in the last inequality above). By replacing this choice of $z$ above, we obtain
$$
\dot{R}(v) \leq  \tilde{\kappa}_U  \alpha\cdot \left(\frac{\alpha}{\alpha-1}\right)^{\alpha-1} (\ps-v)^{\alpha-1}
\leq   \alpha\cdot e \tilde{\kappa}_U (\ps-v)^{\alpha-1},
$$
where we have used that $\left(\frac{\alpha}{\alpha-1}\right)^{\alpha-1}$ is bounded above by $e$. In order to conclude the argument we need to verify that $z^\star \in (0,v)$.  If $v$ is taken sufficiently close to $\ps$ then because $\alpha>1$ we have that $z^\star$ can be made positive.
Additionally, $z^\star<v$ if and only if $\alpha v -\ps < (\alpha-1)v$ which holds  because we are assuming that $v<\ps.$ This yields a new neighborhood around $\ps$. Taking the intersection of this neighborhood with $\mathcal{N}_c$ and $\mathcal{N}_{\tilde{\ell}}$ implies \Cref{eq:def1-equiv-2}. This concludes the proof.
\end{proof}

\section{The Revelation Principle}\label{app-rev}
The standard revelation principle can be stated as follows: Given a mechanism and an equilibrium for that mechanism, there exists a direct mechanism in which (1) it is an equilibrium for each buyer to report his or her type truthfully and (2) the outcomes are the same as in the given equilibrium of the original mechanism.  In our case, the statement is similar: Given a mechanism and an  $\eps-$equilibrium for that mechanism, there exists a direct mechanism in which (1) it is an $\eps$-equilibrium for the buyer to report her type truthfully and (2) the outcomes are the same as in the given $\eps$-equilibrium of the original mechanism.

Mathematically, the buyer's set of messages is denoted by $\Theta$. The seller aims to design an indirect selling mechanism given by $\hat{\all}:\Theta\rightarrow \RR$ and $\hat{\tra}:\Theta\rightarrow \RR$, where $\hat{\all}$ denotes the  allocation probability and $\hat{\tra}$  the  transfers. The buyer's strategic response to the mechanism is a function $\hat{\theta}: \mathcal{S}\rightarrow \Theta$. Given a particular mechanism, %
we are interested in $\eps$-equilibria that satisfy:
 $$v\cdot \hat{\all}(\hat{\theta}(v))-\hat{\tra}(\hat{\theta}(v))\geq v\cdot \hat{\all}(\hat{\theta}')-\hat{\tra}(\hat{\theta}')-\eps, \quad \forall v,\hat{\theta}',$$
 and the participation constraint:
  $$v\cdot \hat{\all}(\hat{\theta}(v))-\hat{\tra}(\hat{\theta}(v))\geq 0, \quad \forall v.$$
That is, the buyer aims to select a reporting strategy that ensures that he collects non-negative utility from participating,   and the buyer is satisficing in that he aims to collect the maximum surplus up to $\varepsilon$, i.e., for all $v\in \mathcal{S}$. When there are multiple $\eps$-optimal best responses for the buyer, we assume that the buyer chooses the one that is the most favorable to the principal. An implication of this assumption, which is common in the mechanism design literature, is that the seller's problem reduces to simultaneously choosing a mechanism together with a best response for the buyer, which together are individually rational and approximately incentive compatible. The seller aims to maximize the expected revenues from trade $\E_v[\hat{\tra}(\hat{\theta}(v))]$ subject to the above constraints. We let
$$
\all(v) = \hat{\all}(\hat{\theta}(v) \quad \text{and}\quad t(v)=\hat{\tra}(\hat{\theta}(v)).
$$
With this transformation, we can optimize over mechanisms that take as input values. Indeed, we just saw how a solution from the indirect mechanism problem can be transformed into a solution to the direct mechanism problem. A solution from the direct mechanism problem is also a solution to the indirect mechanism problem because every direct mechanism is an indirect mechanism. For this to work, we use the assumption that we are choosing the best $\eps$-best response. %

\end{APPENDICES}

\bibliographystyle{informs2014} %
\bibliography{references.bib} %

\end{document}